\documentclass[%
reprint,
superscriptaddress,
 amsmath,amssymb,
aps,
]{revtex4-1}

\usepackage[english]{babel}
\usepackage[utf8]{inputenc}

\setcitestyle{sort&compress,numbers}
\usepackage{amsthm, color, mathtools, amsmath, amssymb, setspace,bbm, tensor, subfigure, mathtools, placeins, graphicx, wrapfig,float, verbatim, xcolor}
\usepackage{graphicx}
\usepackage[normalem]{ulem}
\usepackage[unicode=true, breaklinks=false, pdfborder={0 0 1}, backref=false, colorlinks=true, linkcolor=blue, citecolor=blue]{hyperref}

\usepackage{cancel,bm}

\let\oldsqrt\sqrt
\def\sqrt{\mathpalette\DHLhksqrt}
\def\DHLhksqrt#1#2{%
\setbox0=\hbox{$#1\oldsqrt{#2\,}$}\dimen0=\ht0
\advance\dimen0-0.2\ht0
\setbox2=\hbox{\vrule height\ht0 depth -\dimen0}%
{\box0\lower0.4pt\box2}}

\newcommand{\tr}{\operatorname{tr}}
\newcommand{\ptr}[1]{\operatorname{tr}_{#1}}

\newcommand{\Bcal}{\mathcal{B}}
\newcommand{\Ccal}{\mathcal{C}}

\newcommand{\Ecal}{\mathcal{E}}
\newcommand{\Fcal}{\mathcal{F}}

\newcommand{\Hcal}{\mathcal{H}}
\newcommand{\Ical}{\mathcal{I}}
\newcommand{\Mcal}{\mathcal{M}}

\newcommand{\Tcal}{\mathcal{T}}
\newcommand{\Ucal}{\mathcal{U}}

\newcommand{\Lcal}{\mathcal{L}}

\newcommand{\Jcal}{\mathcal{J}}
\newcommand{\Pcal}{\mathcal{P}}
\newcommand{\Pprob}{\mathbb{P}}

\newcommand{\Xcal}{\mathcal{X}}
\newcommand{\ident}{\mathbbm{1}}

\newcommand{\DOthm}{{\normalfont D{\O}}\ }
\newcommand{\DOpr}{D{\O}\ }
\newcommand{\inp}{\text{\normalfont \texttt{i}}}
\newcommand{\out}{\text{\normalfont \texttt{o}}}

\newcommand{\sbt}{\,\begin{picture}(-1,1)(-1,-3)\circle*{2.5}\end{picture}\ }

\def\bra#1{\mathinner{\langle{#1}|}}

\def\ket#1{\mathinner{|{#1}\rangle}}

\def\braket#1{\mathinner{\langle{#1}\rangle}}

\newcommand*\xbar[1]{%
   \hbox{%
     \vbox{%
       \hrule height 0.5pt 
       \kern0.5ex
       \hbox{%
         \kern-0.2em
         \ensuremath{#1}%
         \kern-0.0em
       }%
     }%
   }%
}

\def\BraVert{\egroup\,\mid\,\bgroup}

\def\ketbra#1#2{\ket{#1\vphantom{#2}}\!\bra{#2\vphantom{#1}}}

\def\bra#1{\mathinner{\langle{#1}|}}

\def\ket#1{\mathinner{|{#1}\rangle}}

\def\braket#1{\mathinner{\langle{#1}\rangle}}

\usepackage{bbm, braket}
\usepackage[mathscr]{euscript}
\allowdisplaybreaks
\newtheorem*{theorem*}{Theorem}
\newtheorem{theorem}{Theorem}
\newtheorem{definition}{Definition}
\newtheorem{lemma}{Lemma}

\newtheorem{example}{Example}
\newtheorem{corollary}{Corollary}
\newtheorem*{corollary*}{Corollary}

\makeatletter
\newcommand{\neutralize}[1]{\expandafter\let\csname c@#1\endcsname\count@}
\makeatother

\newtheorem{thm}{Theorem}

\newenvironment{thmbis}[1]
  {\renewcommand{\thethm}{\ref*{#1}$'$}%
   \neutralize{thm}\phantomsection
   \begin{thm}}
  {\end{thm}}

\theoremstyle{definition}
\newtheorem{manualtheoreminner}{Example}
\newenvironment{manualtheorem}[1]{%
  \renewcommand\themanualtheoreminner{#1}%
  \manualtheoreminner
}{\endmanualtheoreminner}

\graphicspath{{Images/}}

\begin{document}

\title{When is a non-Markovian quantum process classical?}

\author{Simon Milz}
\email{simon.milz@oeaw.ac.at} 
\affiliation{School of Physics and Astronomy, Monash University, Clayton, Victoria 3800, Australia}
\affiliation{Institute for Quantum Optics and Quantum Information, Austrian Academy of Sciences, Boltzmanngasse 3, 1090 Vienna, Austria}

\author{Dario Egloff}
\email{dario.egloff@mailbox.tu-dresden.de}
\affiliation{Institute  of  Theoretical  Physics and IQST,  Universit{\"a}t  Ulm,  Albert-Einstein-Allee  11D-89069  Ulm,  Germany}
\affiliation{Institute  of  Theoretical  Physics,  Technical University Dresden, D-01062 Dresden, Germany}

\author{Philip Taranto}
\affiliation{Institute for Quantum Optics and Quantum Information, Austrian Academy of Sciences, Boltzmanngasse 3, 1090 Vienna, Austria}

\author{Thomas Theurer}
\affiliation{Institute  of  Theoretical  Physics and IQST,  Universit{\"a}t  Ulm,  Albert-Einstein-Allee  11D-89069  Ulm,  Germany}

\author{Martin B. Plenio}
\affiliation{Institute  of  Theoretical  Physics and IQST,  Universit{\"a}t  Ulm,  Albert-Einstein-Allee  11D-89069  Ulm,  Germany}

\author{Andrea Smirne}
\affiliation{Institute  of  Theoretical  Physics and IQST,  Universit{\"a}t  Ulm,  Albert-Einstein-Allee  11D-89069  Ulm,  Germany}
\affiliation{\mbox{Dipartimento di Fisica ``Aldo Pontremoli", Universit{\`a} degli Studi di Milano, via Celoria 16, I-20133 Milan,
Italy}}
\affiliation{Istituto Nazionale di Fisica Nucleare, Sezione di Milano, via Celoria 16, I-20133 Milan, Italy}

\author{Susana F. Huelga}
\email{susana.huelga@uni-ulm.de}
\affiliation{Institute  of  Theoretical  Physics and IQST,  Universit{\"a}t  Ulm,  Albert-Einstein-Allee  11D-89069  Ulm,  Germany}

\begin{abstract}
More than a century after the inception of quantum theory, the question of which traits and phenomena are fundamentally quantum remains under debate. Here  we give an answer to this question for temporal processes which are
    probed sequentially by means of projective measurements of the same observable. Defining classical processes as those that can---in principle---be simulated by means of classical resources only, we fully characterize the set of such processes. Based on this characterization, we show that for non-Markovian processes (i.e., processes with memory), the absence of coherence does not guarantee the classicality of observed phenomena and furthermore derive an experimentally and computationally accessible measure for non-classicality in the presence of memory. We then provide a direct connection between classicality and the vanishing of quantum discord between the evolving system and its environment. Finally, we demonstrate that---in contrast to the memoryless setting---in the non-Markovian case, there exist processes that are \textit{genuinely quantum}, i.e., they display non-classical statistics independent of the measurement scheme that is employed to probe them.
\end{abstract}

\date{\today}
\maketitle
\section{Introduction}
Quantum coherence is considered to be one of the fundamental traits that distinguishes quantum from classical mechanics~\cite{streltsov_measuring_2015,streltsov2017colloquium,hu2018quantum}. Beyond its mathematical deviation from classical theory, it plays an important role in the enhancement of quantum metrology tasks~\cite{napoli2016robustness,marvian2016quantify}, constitutes a fundamental requirement for many quantum algorithms~\cite{hillery2016coherence,matera2016coherent}, and has been conjectured to be
necessary for the formulation of efficient transport models in biology
that are consistent with spectroscopic data~\cite{huelga2013vibrations,Scholes2017,Wang2019}. Consequently, the resource theory of coherence~\cite{aberg2006quantifying,baumgratz2014quantifying, levi2014quantitative, piani_robustness_2016,chitambar2016comparison,winter2016operational,chitambar2016critical,chitambar2018quantum,theurer2018quantifying} has been of tremendous interest in recent years, and has seen rapid development both on the theoretical as well as the experimental side~\cite{wu2019quantum}.

Despite such progress and the growing wealth of accompanying evidence that links coherence to non-classical phenomena, the explicit connection between the two remains unclear and subject to active debate~\cite{wilde_mark_m_could_2010,briggs_equivalence_2011,miller_perspective_2012,leon-montiel_highly_2013,oreilly_non-classicality_2014}. Put differently, the mere presence of coherence does not guarantee the existence of effects that cannot be explained on purely classical grounds, and an unambiguous relationship between coherence and non-classicality has not been established yet.

In order to provide such a connection, an operationally meaningful and clear-cut definition of classicality is crucial. One such possible definition is based on experimentally attainable quantities only, namely the joint probability distributions obtained from sequential measurements of an observable~\footnote{For a different demarcation line between classical and quantum physics, based on the memory cost required to simulate a given process, see, e.g., Ref.~\cite{budroni_memory_2019}}. If these satisfy the Kolmogorov consistency conditions for all considered sets of measurement times---which provide the starting point for the formulation of the theory of classical stochastic processes~\cite{kolmogorov_foundations_1956,breuer_theory_2007}---then they can, in principle, be explained by a fully classical model and there is therefore nothing inherently quantum about the observed phenomenon. If they do not, then there exists no underlying classical stochastic process that could lead to the observed joint probability distributions, and the corresponding process is considered non-classical. This characterization of classicality is in the spirit of the derivation of Leggett-Garg inequalities, where, instead of classicality, non-invasiveness and macroscopic realism are put to the test ~\cite{leggett1985quantum,leggett_realism_2008}. Indeed, any set of probability distributions that satisfies the Kolmogorov conditions does not violate the corresponding Leggett-Garg inequalities~\cite{emary_leggettgarg_2014, asano_violation_2014}.

Following this line of reasoning, and in a sense to be further specified later more precisely, in Ref.~\cite{smirne_coherence_2019} a one-to-one connection was derived between the notion of classicality based on the Kolmogorov conditions and the coherence properties of the dynamics of Markovian (i.e., memoryless) quantum processes: such a process is classical iff the corresponding dynamical propagators can never create coherence that can be detected at any later time. Thus, a direct relation between the mathematical notion of coherence and an operationally well-defined and broadly applicable notion of classicality has been established. In turn, this relation provides a direct interpretation of Markovian processes that violate Leggett-Garg inequalities in terms of the underlying quantum resources. However, this connection only holds in the memoryless case and does not straightforwardly apply to the non-Markovian scenario, where, amongst other issues, such propagators cannot be used to compute multi-time statistics~\cite{milz_introduction_2017}. 

Here, we go beyond this paradigm of memoryless processes and consider the general case of non-Markovian dynamics. Such general processes can be described in terms of higher-order quantum maps, so-called quantum combs~\cite{chiribella_transforming_2008, chiribella_transforming_2008, chiribella_theoretical_2009}. Recently, this framework has been tailored to the description of open quantum system dynamics~\cite{pollock_non-markovian_2018,pollock_operational_2018}, and has---amongst others---found direct application in the characterization of multi-time memory effects~\cite{taranto_2019L,taranto_2019A,taranto_2019S,TarantoThesis} and within the field of stochastic thermodynamics~\cite{strasberg_operational_2018,strasberg_stochastic_2019, strasberg_repeated_2019}. Here, we employ it to extend the results of Ref.~\cite{smirne_coherence_2019} to the non-Markovian case. In particular, we link spatial quantum correlations or, more precisely, the discord between an observed system and an environment to the non-classicality of the observed measurement statistics. Somewhat surprisingly, for the case of general processes---where memory effects play a non-negligible role---the presence of non-classical phenomena is not solely dependent on the ability of the process to create or detect coherence, in stark contrast to the memoryless case. As we will show, the absence of detectable coherence is not necessarily sufficient to enforce classical behavior in general. Rather, classicality of multi-time statistics is inherently linked to quantum discord---which was originally introduced as a means to distinguish classical spatial correlations from non-classical ones~\cite{Zurek_2000,henderson_classical_2001, ollivier_quantum_2001,modi_classical-quantum_2012}---between the evolving system and its environment. We characterize the complete set of classical processes and derive a concrete relation between the presence and detectability of discord and the non-classicality of observed multi-time measurement statistics. This, in turn, allows for the derivation of experimentally accessible quantifiers of non-classicality and the categorization of the resources required for the implementation of a non-classical, non-Markovian process, paving the way to a clear-cut understanding of non-classicality on operational grounds. 

In a similar manner to the analysis of coherences, our results will predominantly be phrased with respect to measurements in an arbitrary, but \textit{predetermined} basis i.e., with respect to a fixed observable, raising the question if classicality is merely a question of perspective; in principle, for every process, there could exist a sequential measurement scheme, that yields classical statistics. While this \textit{always} holds true for processes in classical physics, as well as memoryless quantum processes, we show by means of an explicit example, that this is not necessarily the case for quantum processes with memory; in the presence of quantum memory, there exists a fundamentally new class of processes, which we will call \textit{genuinely} quantum processes, that lead to non-classical statistics \textit{independent} of how they are probed.

Throughout this article, we investigate the question of when a physical process---with or without memory---can be considered classical, and what classicality implies if we assume the underlying theory to be quantum mechanics. Concretely, for the most part, we consider the scenario of a quantum system of interest that is sequentially probed in a fixed basis, that is, interrogated at successive points in time---like, for example, in Leggett-Garg type experiments---and we are interested in characterizing when the multi-time measurement statistics resulting from such a scenario can be simulated by a classical stochastic process, and thus be reasonably considered \textit{classical}.

As we will make no assumption about the underlying dynamics, the system of interest can be coupled to an environment that is out of the experimenter's control and can thus undergo an open evolution that displays complex classical and quantum memory effects. The classicality of the observed statistics then depends on the interplay of the dynamics of the system of interest, the pertinent memory effects, and the way in which the system is probed. We derive both the structural as well as dynamical properties of general classical non-Markovian processes, providing an answer to the question: \emph{What is a non-classical process, and what are its key features?} 

Finally, by dropping the restriction to fixed instruments, we show that an \textit{observer-independent} notion of non-classicality exists, i.e., that there are processes that, no matter how they are probed, display statistics that cannot be simulated by classical stochastic processes. As such processes cannot exist in the absence of memory, the interplay of quantum memory effects and quantum dynamics leads to a fundamentally new class of processes---genuinely quantum processes---that cannot hide their non-classicality.

\section{Summary of the main results}
Before providing detailed derivations in the subsequent sections, here, we give a more concrete overview of the main results of our work. Throughout this article, we define the classicality of a process based on observed multi-time statistics $\Pprob_n(x_n,t_n;\hdots;x_1,t_1)$ for measurements at different times $\{t_1,\hdots,t_n\}$. The number of possible outcomes is always considered to be finite, and, unless stated otherwise, the measurements are given by measurements in the computational basis $\{\ketbra{x_k}{x_k}\}$. With respect to these statistics, a process is considered classical (on $K$ times), if the made measurements are non-invasive, i.e., they satisfy the Kolmogorov conditions 
\begin{align}
	&\Pprob_{n-1}(x_n,t_n;\ldots;\cancel{x_j,t_j};\ldots;x_1,t_1)
	\\ = &\nonumber  \sum_{x_j} \Pprob_n(x_n,t_n;\ldots;x_j,t_j;\ldots;x_1,t_1)	\quad \forall \ n\leq K, \,\forall\ j \, .
\end{align}
On the other hand, it is Markovian, i.e., memoryless, if the respective conditional probabilities satisfy 
\begin{gather}
\Pprob(x_{n}|x_{n-1},\ldots, x_1) = \Pprob(x_{n}|x_{n-1}) \quad \forall \ n\leq K \, .
\end{gather}
In quantum mechanics, such a process can be modeled by means of completely positive trace preserving maps $\{\Lambda_{t_j,t_{j-1}}\}$, which act on the probed system and describe the dynamics between measurements, as well as an initial system state $\rho_{t_0}$. 

Going beyond the results of Ref.~\cite{smirne_coherence_2019}, we show that (see Theorem~\ref{thm::NCGD_classical}) a Markovian process is classical iff it can be modeled by a state $\rho_{t_0}$ that is diagonal in the measurement basis $\{\ketbra{x_k}{x_k}\}$ and non-coherence-generating-and-detecting (\textbf{NCGD}) maps $\Lambda_{t_k,t_{k-1}}$, i.e., maps that satisfy 
\begin{align}
\notag
    \Delta\circ\Lambda_{t_{j+1},t_j}\circ&\Delta \circ \Lambda_{t_j,t_{j-1}}\circ \Delta \\
&= \Delta \circ\Lambda_{t_{j+1},t_j}\circ \Lambda_{t_j,t_{j-1}}\circ \Delta \quad \forall j\, ,
\end{align}
where $\Delta$ is the completely dephasing map in the measurement basis, and $\circ$ denotes composition. Intuitively, maps that satisfy the above equation can create coherences, but not in a way that can be detected at a later time by means of the employed measurement basis. Thus, Theorem~\ref{thm::NCGD_classical} provides a direct connection between coherence and an experimentally testable notion of classicality in the Markovian case. 

Going beyond the Markovian case we show that this direct connection between coherence and classicality breaks down when memory is present. We provide an explicit example (Example~\ref{ex::class_state}) of a dynamics $U_{t_j,t_i} \ket{\ell, p} = e^{i \phi_\ell p (t_j-t_i)} \ket{\ell, p}$ acting on a qubit system (represented by $\ell$) coupled to a continuous degree of freedom (represented by $p$) that---for the right choice of initial environment state---never displays coherences in the system state, but exhibits non-classical statistics nonetheless. 

When memory plays a non-negligible role, individual CPTP maps that act on the system alone are insufficient for the computation of multi-time probabilities. Rather, probabilities are computed by means of higher order quantum maps, called quantum combs~\cite{chiribella_quantum_2008, chiribella_theoretical_2009}. These maps contain all information about the underlying process at hand, and multi-time joint probabilities can then be expressed as 
\begin{gather}
    \Pprob_{K}(x_K,t_K;\ldots;x_1,t_1) = \Ccal_{K}[\Pcal_{x_K},\hdots, \Pcal_{x_1}]\, ,
\end{gather}
where $\Ccal_{K}$ is the quantum comb of the process and $\{\Pcal_{x_j}\}$ are the CP maps corresponding to measurements with outcome $x_j$, i.e., $\Pcal_{x_j}[\rho] = \braket{x_j|\rho|x_j}\ketbra{x_j}{x_j}$. 

We derive a full characterization of combs that lead to classical statistics in Theorem~\ref{thm::KclassComb}, and make this characterization more concrete in Theorem~\ref{thm::KclassComb_prime}, employing the Choi-Jamio{\l}kowski isomorphism (\textbf{CJI}) that allows one to map higher order quantum maps $\Ccal_n$ onto multipartite quantum states $C_n$. 

Using this full characterization, a measure  $M(C)$ for the non-classicality of a process $C$ can be derived. We phrase this problem in terms of the operational task of deciding whether or not a given comb $C$ is classical, and show that the corresponding maximum probability to guess correctly is given by (see Eq.~\eqref{eqn::success_distinction})
\begin{gather}
    \Pprob(C) = \tfrac{1}{2}(1+M(C))\, ,
\end{gather}
where $M(C)$ can both be computed efficiently via a linear program (see Eq.~\eqref{eqn::linearProgMeas}) and is accessible experimentally---and could be evaluated based on already existing experimental data (e.g., in Ref.~\cite{smirne_experimental_2019}). We show that, e.g., in the two-time case
\begin{gather}
    M(C) \leq \sum_{x_2} \left|\Pprob(x_2) - \sum_{x_1} \Pprob(x_2,x_1) \right|\, ,
\end{gather}
holds, where the right hand side of the above equation is a natural quantifier of classicality, that is used both theoretically, as well as experimentally (for example in Leggett-Garg type scenarios) to quantify the non-classicality of sequential measurement statistics. 

In the same vein as in the Markovian case, the dynamical properties (in contrast to the aforementioned structural ones) of classical processes can be obtained. In the non-Markovian case, a process is fully defined by an initial system-environment state $\eta_{t_0}^{se}$ and intermediate system-environment CPTP maps $\Gamma_{t_j,t_{j-1}}$. We show that in the non-Markovian case, rather than the coherences of the system it is the (basis dependent) system-environment discord ~\cite{Zurek_2000, ollivier_quantum_2001,henderson_classical_2001, modi_classical-quantum_2012} that determines the classicality of the observed statistics. In particular, we demonstrate (see Thms.~\ref{thm::NDCG} and~\ref{thm::NDGD_dilation}) that a process is classical iff it can be modeled by an initial state $\eta_{t_0}^{se}$ with vanishing (basis dependent) discord, i.e., $\eta_{t_0}^{se} = \sum_m p_m\, \ketbra{x_m}{x_m} \otimes \xi_m$, and a set of system-environment maps that is non-discord-generating-and-detecting (\textbf{NDGD}), i.e., 
\begin{align}
\notag
    \Delta \circ \Gamma_{t_{j+1},t_j} \circ \Delta \circ &\Gamma_{t_j,t_{j-1}} \circ \Delta\\
    &\phantom{asdf}= \Delta \circ \Gamma_{t_{j+1},t_j} \circ \Gamma_{t_j,t_{j-1}} \circ \Delta\, ,
\end{align}
where the completely dephasing map $\Delta$ acts on the system alone. Analogously to the Markovian case, the above equation implies that the maps $\{\Gamma_{t_j,t_{j-1}}\}$ can create discord, but said discord cannot be detected by means of later measurements on the system in the chosen measurement basis. In turn, this result provides a direct connection between quantum discord and the classicality of a quantum process. Additionally, it also gives an \textit{a posteriori} explanation why the absence of coherence in Example~\ref{ex::class_state} did not lead to classical statistics (for an explicit discussion of the discord that leads to of non-classical statistics in Example~\ref{ex::class_state}, see its continuation Example~~\ref{ex::classical_prime}).

While, in principle, these aforementioned results do not rely on the fact that we assume measurements in \textit{one} fixed basis, but could similarly be obtained for different (but fixed) instruments at every time, they still depend on the fact that one specific measurement scheme is chosen beforehand. Classicality (or the absence thereof) of the observed statistics could thus depend on the respective choice of measurement schemes. This holds true in the Markovian case, where there is always a choice of measurement bases that renders the observed statistics classical. However, as we show by explicit example (see Sec.~\ref{sec:genuinelyquantumprocess}), there are processes with memory---dubbed genuinely quantum---that display non-classical statistics \textit{independent} of the employed measurement scheme.

The Paper is structured as follows: In Sec.~\ref{sec::gen_frame} we introduce the basic concepts that will be employed throughout this article to examine classicality. In Sec.~\ref{sec::Coh_Cla}, we reiterate and slightly generalize the results of Ref.~\cite{smirne_coherence_2019} linking non-classicality and coherence for the Markovian case, and discuss their breakdown when memory effects are present. This motivates our consideration of the non-Markovian case in Sec.~\ref{sec::combs}, where we fully characterize the set of general classical processes by means of the quantum comb framework. This characterization then enables us to formulate a quantifier of non-classicality, that is both experimentally accessible and can be computed efficiently. Based on these results, in Sec.~\ref{sec::DiscClass}, we subsequently establish the direct connection between (basis dependent) quantum discord and the classicality of temporal processes. Finally, in Sec.~\ref{sec:genuinelyquantumprocess}, we go beyond the paradigm of measurements in a fixed basis, and provide an example for processes that appear quantum independent of the scheme that is used to probe them. The paper concludes in Sec.~\ref{sec::ConclOut} with a summary and an outlook on  further research directions and open problems.

\section{General framework}
\label{sec::gen_frame}
The overarching aim of this paper is to characterize when a general quantum mechanical process can be considered classical in an operationally consistent manner and identify the structural properties consequently implied on the underlying evolution. Importantly, our investigation will be operational in the sense that it is based solely on experimentally accessible quantities; as such, it applies to situations where the underlying theory is classical mechanics, quantum mechanics, or some more general theory~\cite{chiribella_probabilistic_2010}.

Ultimately, any physical theory provides predictions about possible observations---only these can be tested by experiments. That is, any theory must (in principle) provide the correct probabilities for measurement outcomes (or sequences thereof) to occur when a system of interest is experimentally probed. The difference between predictions made regarding such observable quantities by classical physics and quantum (or post-quantum) theory can then be used to unambiguously demarcate between the theories on the investigated spatial and temporal scales. 

Following Ref.~\cite{smirne_coherence_2019}, we will thus define our notion of classicality by means of joint probability distributions pertaining to sequences of measurement outcomes, as these are precisely what is obtained when a temporal process is probed.

\subsection{Kolmogorov conditions and classicality}
In classical physics, a stochastic process on a set of $K$ times is fully described by a joint probability distribution 
\begin{gather}
	\label{eqn::stochClass}
	\Pprob_K(x_{K},t_{K};\dots;x_{1},t_1)\, 
\end{gather}
which yields the probability to measure the realizations $\{x_{K},\dots,x_1\}$ of the random variables $\{X_{K},\dots,X_1\}$ at times $\{t_{K},\dots, t_1\}$. For example, $\Pprob_2(x_2,t_2;x_1,t_1)$ could describe the probability to obtain both outcomes $\{x_2,x_1\}$ when measuring the position of a particle undergoing Brownian motion at times $t_1$ and $t_2>t_1$. In what follows, we will often omit the explicit time label, with the understanding that $x_j$ denotes an outcome of a measurement at time $t_j$. 

Crucially, in classical physics, joint probability distributions describing a stochastic process for different sets of times satisfy the so-called \textit{Kolmogorov consistency conditions}~\cite{kolmogorov_foundations_1956, feller_introduction_1968, breuer_theory_2007,tao_introduction_2011}: given a joint probability distribution $\Pprob_K$ for a set of times, the probability distributions for all subsets of times can be obtained by marginalization, that is 
\begin{align}\label{eqn::Kolmo_cond}
	&\Pprob_{n-1}(x_n,t_n;\ldots;\cancel{x_j,t_j};\ldots;x_1,t_1)
	\\ = &\nonumber  \sum_{x_j} \Pprob_n(x_n,t_n;\ldots;x_j,t_j;\ldots;x_1,t_1)	\quad \forall \ n\leq K, \,\forall\ j \, .
\end{align}
Just like the Leggett-Garg inequalities~\cite{leggett1985quantum, leggett_realism_2008,emary_leggettgarg_2014} for temporal correlations, the satisfaction of these requirements is based on the assumptions of realism \textit{per se}, i.e., the assumption that $x_j$ has a definite value at any time $t_j$, and the possibility to implement non-invasive measurements~\cite{milz_kolmogorov_2017}.

Importantly, an experimenter obtaining a family of joint probability distributions that satisfies the Kolmogorov conditions when probing a temporal process at different sets of times would not be able to distinguish said process from a classical one, as every such finite family can be obtained from a---potentially exotic---underlying \textit{classical stochastic process}. More generally, the \textit{Kolmogorov extension theorem} states that if all joint probability distributions for finite subsets of a time interval $[0,t]$ satisfy the consistency conditions of Eq.~\eqref{eqn::Kolmo_cond} amongst each other, then there exists an underlying classical stochastic process on said time interval that leads to the observed probability distributions~\cite{kolmogorov_foundations_1956, feller_introduction_1968, breuer_theory_2007, tao_introduction_2011}. In other words, if the Komogorov consistency conditions of Eq.~\eqref{eqn::Kolmo_cond} are satisfied (for all considered choices of $t_j$), then there is nothing inherently quantum mechanical about the observed process. We therefore define:
\begin{definition}[$K$-classical process~\cite{smirne_coherence_2019}]\label{def::N-classical_statistics}
Let $\Xcal$ be a finite set. A process defined on a set of times $\Tcal$, with $|\Tcal| = K$, that is described by the joint probabilities $\Pprob_n\left(x_n,t_n;\ldots;x_1,t_1\right)$, with $t_n\geq\dots\geq t_1$, $t_i\in \mathcal{T}$, $n \leq K$ and 
 $x_i\in \Xcal$, is said to be $K$-classical if the Kolmogorov consistency conditions of Eq.~\eqref{eqn::Kolmo_cond} are satisfied up to $n=K$. 
\end{definition}

Throughout this article, we will call a family of joint probabilities on a set of $K$ times a $K$-\textit{process}
and denote it by $\left\{\Pprob_n(x_n,\ldots, x_1)\right\}_{n\leq K}$. Here, the label $n\leq K$ is a short-hand notation for all the subsets of $\Tcal$ with $n$ ordered times $t_n \geq \ldots \geq t_1$, where $t_i \in \mathcal{T}$, for any $n\leq K$; moreover from here on we will not indicate explicitly the time arguments in the probability distributions, implying that the outcome
$x_j$ refers to time $t_j$.

While the above definition of classicality seems intuitive, some comments are in order. First, we choose to define classicality for a finite set of $K$ times.
 While this is motivated on a practical ground, the general definition of a classical stochastic process involves the joint probability distributions
associated with any number of ordered time instants $t_K \geq \ldots \geq t_1$, with $K\in\mathbbm{N}$, and any
choice of such instants. In particular, as said, the Kolmogorov extension theorem
infers the existence of a stochastic process from the validity of the consistency conditions on all such joint distributions. Here, instead, we fix a finite value of $K$ and the sequence of time instants beforehand, so that, 
given the $K$-time joint probability distribution of a $K$-classical process, the involved hierarchy of probability distributions can be constructed by iteratively applying the consistency conditions, at any intermediate time.

Second, the above definition of classicality is \textit{a priori} device independent, as it only relies on the inferred statistics without any assumptions on the underlying theory and/or measurement devices; as a consequence,  the classicality of a process according to the above definition depends upon the manner in which the system of interest is probed. Although often overlooked, this is also the case in classical physics: given some underlying classical stochastic process, not \textit{every} set of measurements that an experimenter might be able to perform will lead to a family of probability distributions that satisfies the above definition of $K$-classicality. In fact, if performing such measurements might potentially disturb the system (i.e., the measurement is invasive), the Kolmogorov condition fails in general, even if the underlying evolution is classical~\cite{milz_kolmogorov_2017}. 

For example, suppose that instead of merely measuring the position of a particle at different times when probing a Brownian motion process, an experimenter chooses to displace the particle at each time depending on where it was found. In this case, Eq.~\eqref{eqn::Kolmo_cond} would generally fail to hold for the joint probability distributions observed. Consequently, the Kolmogorov consistency conditions in Eq.~\eqref{eqn::Kolmo_cond} are in fact a statement of the \textit{non-invasiveness} of the performed measurements: if they hold true, then \emph{not} performing a measurement at any given time cannot be distinguished (for the given experimental situation) from averaging over their probabilities (i.e., forgetting the outcomes of the measurements performed).

In classical physics one assumes that, in principle, one could measure the system without disturbing it, and that therefore there exists a family of joint probability distributions that can consistently explain all possible outcome probabilities. Such a non-invasive and complete measurement is often referred to as an `ideal measurement' in the literature~\cite{piron_ideal_1981}.

On the other hand, in quantum mechanics any measurement disturbs \emph{some} system state and therefore ideal measurements do not exist in general in the strong sense discussed above. As a consequence, quantum mechanical processes generically do not satisfy Kolmogorov conditions~\cite{BreuerEA2016,milz_kolmogorov_2017}, a fact that fundamentally distinguishes them from the classical realm. 

More generally,
the violation of Bell, Kochen-Specker, or Leggett-Garg inequalities, which can be observed in quantum mechanics, 
are different manifestations of the impossibility to obtain the measured data by non-invasive measurements. Particularly, in the case of Leggett-Garg inequalities~\cite{fine1982hidden,leggett1985quantum}, it is precisely the breakdown of Kolmogorov conditions that is being probed~\cite{milz_kolmogorov_2017,smirne_coherence_2019}, and our above definition of classicality is hence in line with the wider program of determining fundamentally quantum traits of nature.  

\begin{figure}
    \centering
    \includegraphics[width=0.9\linewidth]{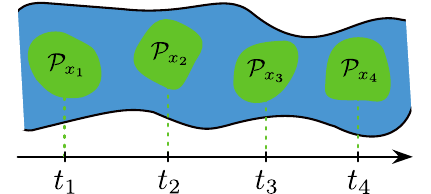}
    \caption{{\bf{Probing a process with projective measurements.}} At each time $t_j$, the process (depicted in blue) is probed by a projective measurement (depicted in green) with outcomes $x_j$,
    where each $x_j$ belongs to the same finite set $\mathcal{X}$. If the resulting family of probability distributions $\Pprob_n$ (depicted are the cases $n \leq 4$) satisfies the Kolmogorov consistency conditions, then not performing a measurement at a time $t_j$ cannot be distinguished from performing a measurement and averaging over the outcomes. In this case, this experiment cannot be distinguished from a classical one, even though the underlying evolution might be quantum mechanical.}
    \label{fig::Process}
\end{figure}

\subsection{Measurement setup}
As mentioned above, the structural properties of families of joint probability distributions depend on the way in which a system of interest is probed. Consequently, before being able to analyze the set of  quantum processes, it is crucial to fix the measurements that are used to probe a process at hand. Although there are no ideal measurements in quantum mechanics, projective measurements share some basic features with the classical ideal measurements discussed above, and are thus a natural choice. In particular, they guarantee repeatability, i.e., that two sequential measurements (without any evolution in between) would give the same value with unit probability, as well as a weaker form of ideality, namely that if an outcome occurs with certainty, then the state of the system before the measurement is not disturbed by the latter~\cite{Heinosaari2012}. It therefore suggests itself to start our analysis on the classical reproducibility of quantum processes by focusing on projective measurements; moreover, also following Ref.~\cite{smirne_coherence_2019}, we will further restrict to the case of orthogonal rank-1 (sharp) projectors, like, e.g., projective measurements with respect to the eigenbasis of any non-degenerate self-adjoint
operator.

In many experimental situations of interest, there is a preferred basis to select. For instance, if the dynamics is such that the system dephases to a given basis, the latter provides a natural choice. This occurs, e.g., in the case of open quantum systems dynamics that are subject to environmental fluctuations. In other cases it may make sense to choose the basis more arbitrarily (in advance), for instance when analyzing a specific protocol, or attempting to optimize it (see Ref.~\cite{Egloff2018} for more details). Finally, the experimental setup might only allow for a measurement of one particular observable, in which case the chosen basis would correspond to the eigenbasis of said observable. 

In what follows, we will analyze the classicality of a process based on the joint probability distributions obtained from sequential sharp measurements in a fixed basis $\{\ket{x}\}_{x=1}^d$---henceforth also called the classical, standard, or computational basis---with the action of a measurement with outcome $x$ on a state $\rho$ given by 
\begin{gather}
    \rho\mapsto \Pcal_x[\rho]\coloneqq \ketbra{x}{x}\rho\ketbra{x}{x}\,.
\end{gather}
See Fig.~\ref{fig::Process} for a graphical depiction. 

This freedom in the considered measurements makes the property of classicality fundamentally contingent on the respective choice of measurement basis. However, this basis dependence is unsurprising and mirrored by coherence theory~\cite{streltsov2017colloquium}. There, the existence of off-diagonal elements $\braket{m|\rho|n}$, i.e., coherences, depends on the choice of the basis a quantum state is represented in.  As they are considered to be a fundamentally quantum property, it is a natural question to ask how coherences (with respect to the computational basis) and classicality of a process (with respect to the same basis) are interrelated. Importantly, while the existence of coherences cannot be determined by projective measurements in the computational basis alone, the prevalence of non-classical effects can be. Thus, as we shall see below, providing an operationally accessible notion of classicality allows one to link coherence (and, more generally, quantum correlations) in a quantitative manner to experimentally observable deviations from classical physics.

\subsection{Open (quantum) system dynamics and memory effects}
The definition of classicality we use (introduced in Ref.~\cite{smirne_coherence_2019}) answers the question of whether or not there exists a classical stochastic process that can explain the multi-time probabilities obtained by measuring a quantum system at given times in the computational basis. To make our analysis as general as possible, we will consider the possibility that the measured system interacts with a surrounding environment, which can influence the resulting statistics. Explicitly, assuming that the system and environment in state $\eta$ are together closed and described by quantum mechanics, their joint dynamics between measurements is given by unitary evolution : $\Ucal_{t_{j+1},t_j}[\eta] = U_{t_{j+1},t_j} \eta\, U_{t_{j+1},t_j}^\dagger$. The resulting joint probability distributions read
\begin{align}
\label{eqn::Gen_Prob}
    \notag \Pprob_n(x_n,\dots,x_1) &= \tr\left\{(\Pcal^s_{x_n}\otimes \Ical^e)\circ \Ucal_{t_n,t_{n-1}} \circ \cdots \right. \\
     &\phantom{=}\left.\cdots\circ (\Pcal^s_{x_1}\otimes \Ical^e) \circ \Ucal_{t_1,t_0}[\eta^{se}_{t_0}]\right\}\, ,
\end{align}
where $\eta_{t_1}^{se}$ is the system-environment state at time $t_1$, $\Ical^e$ signifies the identity channel on the environment, $\Pcal^s_{x_j}$ corresponds to a measurement on the system in the computational basis at time $t_j$ with outcome $x_j$ and $\circ$ denotes composition (see Fig.~\ref{fig::Gen_Dilation} for a graphical representation). Whenever there is no risk of confusion, we will drop the additional superscripts $s$ and $e$ throughout this paper. Naturally, the classicality of the family of joint probability distributions obtained via Eq.~\eqref{eqn::Gen_Prob} crucially depends on the properties of the intermediate evolutions $\Ucal_{t_{j+1},t_j}$ and the initial state $\eta^{se}_{t_0}$. 

In general, such a multi-time statistics displays memory effects, i.e., it is non-Markovian: at any point in time $t_j$, the future statistics does not only depend on the measurement outcome $x_j$ at time $t_j$, but also on (potentially) all previous outcomes $x_{j-1}, \hdots, x_1$. Indeed, \textit{all} information about future statistics at $t_j$ is contained in the joint state of system \textit{and} environment, which depends upon the previous measurement outcomes. As this total state cannot be accessed by measurements on the system alone, this dependence on past measurements manifests itself as memory effects on the system level (see Sec.~\ref{sec::combs} for a detailed discussion). 

However, under some specific circumstances, the influence of such memory effects on the multi-time statistics
can be neglected; this is essentially the case when the \emph{quantum regression formula} \textbf{(QRF)} can be applied
\cite{Lax1968,Carmichael1993,breuer_theory_2007,Gardiner2004}. Under this assumption, the observed statistics can be understood in terms of dynamical propagators that act on the system alone, which, in turn, enables one to directly link the classicality of a process to the properties of said propagators in terms of coherence production and detection. The corresponding result has been obtained in Ref.~\cite{smirne_coherence_2019}, and we will reiterate and expand upon it in the coming section. 
Subsequently, employing quantum combs--- a powerful framework for the description of general, possibly non-Markovian open quantum processes---we characterize the set of quantum processes that can be described classically.

\section{Coherence and classicality}
\label{sec::Coh_Cla}

In this section, we reiterate the main result of Ref.~\cite{smirne_coherence_2019} on the connection between coherence and classicality for the memoryless case, generalizing it to the case of a divisible (but not necessarily semigroup~\cite{gorini1976,lindblad1976,breuer_theory_2007}) dynamics. As mentioned above, such a direct connection may be established, because memoryless processes can be understood in terms of propagators that are defined on the system alone, while this property fails to hold in the general, non-Markovian, case. 

After introducing an operational notion of Markovianity associated with the multi-time statistics due to sequential measurements of a (non-degenerate) observable, we present a one-to-one connection between the non-classicality of such statistics and the capability of the open system dynamics to generate and detect coherences with respect to the relevant basis. We also clarify the relation between the notion of Markovianity used in this paper and the QRF, which allows us to straightforwardly recover
the main result of Ref.~\cite{smirne_coherence_2019}. Finally, we lay out the subtleties that arise when generalizing the framework to allow for memory effects, motivating the main results of this work.

\subsection{One-to-one connection in the Markovian case}\label{sec:otoma}

\begin{figure}
    \centering
    \includegraphics[width=0.95\linewidth]{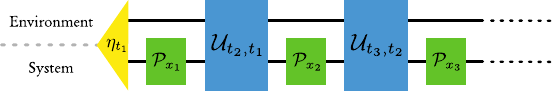}
    \caption{{\bf{General open quantum process.}} The state of the system at time $t_1$ is correlated with the environment (depicted by the yellow triangle representing the joint state). Measurements on the system (green boxes) are performed at times $t_1, t_2, \dots$. In between, the system and the environment undergo a unitary evolution (blue boxes). The distinction between system and environment is given by the degrees of freedom that the experimenter controls (system) and those that remain inaccessible to experimental control (environment).}
    \label{fig::Gen_Dilation}
\end{figure}

\begin{figure}
    \centering
    \includegraphics[width=0.95\linewidth]{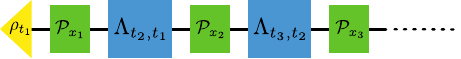}
    \caption{{\bf{Markovian process.}} For a Markovian process, the system dynamics in between intermediate times (depicted as the blue boxes) can be modeled by maps $\Lambda_{t_{j+1},t_j}$ that do not depend on previous outcomes  (i.e., there is no memory). The measurement statistics are obtained by measuring in the classical basis at times $t_1, t_2, t_3, \dots$ (depicted in green); before the first measurement the system is in the state
    $\rho_{t_1}$ (depicted in yellow)} 
    \label{fig::Markov_Dil}
\end{figure}

Classically, a process is Markovian (i.e., memoryless), if, for any chosen time $t_j$,
the future statistics only depend upon the outcome at time $t_j$, but not on any prior outcomes at $t_{j-1}, t_{j-2}, \cdots$; explicitly, a classical stochastic process is Markovian if its statistics satisfy
\begin{gather}
\label{eqn::Markov_prob}
\Pprob(x_j|x_{j-1},\dots,x_1) = \Pprob(x_j|x_{j-1})\, \quad \forall \ j,
\end{gather}
where $\Pprob(x_j|x_{j-1},\dots,x_1)$ is the conditional probability to obtain outcome $x_j$ at time $t_j$ given that outcomes $x_{j-1}, x_{j-2},\dots$ were measured at earlier times $t_{j-1}, t_{j-2}, \dots$~\cite{breuer_theory_2007}. Extending this definition to general (i.e., not necessarily classical) statistics and taking into account that, in practice, one only deals with systems probed at a finite number of times, we obtain the following definition of $K$-Markovianity:
\begin{definition}\label{def::markov}
	Let $\Xcal$ be a finite set. A process defined on a set of times $\Tcal$, with $|\Tcal| = K$ is called $K$-Markovian if it satisfies:
	\begin{gather}
	\label{eq::markov}
	\Pprob(x_{n}|x_{n-1},\ldots, x_1) = \Pprob(x_{n}|x_{n-1}) \quad \forall \ n\leq K \, ,
	\end{gather}
	for all ordered tuples of times $t_n \geq \ldots \geq t_1$, with $t_i\in \mathcal{T}$,
	and $x_i \in \mathcal{X}$.
\end{definition}
Just like our earlier definition of classicality and coherence, the absence of memory effects as defined in Definition~\ref{def::markov} is basis dependent: a process that appears Markovian in one basis may appear non-Markovian when probed in a different one. While there exist basis independent notions of Markovianity in the quantum case~\cite{Lindblad1979,accardi_quantum_1982, pollock_operational_2018,pollock_non-markovian_2018,li_concepts_2018}, the basis dependent one introduced here is best suited for the experimental situation we envision; as such, in what follows, we predominantly understand Markovianity with respect to measurements in the computational basis. We will briefly return to the relation between this basis dependence and the basis independent notion of Markovianity in Sec.~\ref{sec::combs}.

To establish a connection between non-classicality of a Markovian process and the coherence properties of the underlying dynamics, we need to introduce the maps that characterize the dynamical evolution of the open system. To this end, assume that at an initial time $t_0$ (with $t_0\leq t_1$)
the system and the environment are in a product state
 $\eta^{se}_{t_0}=\rho_{t_0}\otimes \sigma_{t_0}$ (for some fixed environment state $\sigma_{t_0}$), so that we
can define the
completely positive and trace preserving \textbf{(CPTP)} dynamical maps $\{\Lambda_{t_j,t_0}\}$ of the open system evolution between the initial time and the measurement times $t_j$~\cite{breuer_theory_2007,rivas_book_2012}
\begin{equation}\label{eq:map}
   \rho_{t_j}=  \Lambda_{t_j,t_0} [\rho_{t_0}] = \mbox{tr}_e\left[U_{t_j,t_0} \left(\rho_{t_0}\otimes \sigma_{t_0}\,\right) U_{t_j,t_0}^\dagger\right]\, ,
\end{equation}
where $\tr_e$ denotes the trace over the environmental degrees of freedom. 
Additionally, let us also assume that the dynamics is divisible~\cite{wolf_dividing_2008}, i.e, we can define the corresponding \textit{propagators} $\{\Lambda_{t_k,t_j}\}$ between \textit{any} two times via the composition rule
\begin{equation}\label{eq:prop}
    \Lambda_{t_k,t_0} = \Lambda_{t_k,t_j} \circ \Lambda_{t_j,t_0} \quad \forall \ t_k\geq t_j \geq t_0\, ,
\end{equation}
and they satisfy the composition law $\Lambda_{t_\ell,t_j} = \Lambda_{t_\ell,t_k} \circ \Lambda_{t_k,t_j}$ for all times $t_\ell \geq t_k \geq t_j$.
Under these assumptions, it is natural to ask, what conditions the propagators $\{\Lambda_{t_k,t_j}\}$ must satisfy in order for the resulting statistics to be classical. However, Eq.~\eqref{eq:prop} does not yet tell us how to obtain multi-time statistics~\cite{milz_cp_2019}. 

The relation we seek is provided by the QRF, which, for example, holds in the weak coupling and the singular coupling limits~\cite{Dumcke1983}, and constitutes a relation between the definition of Markovian processes given by Definition~\ref{def::markov} and the corresponding open system dynamics 
(see also Ref.~\cite{li_concepts_2018} for an extensive discussion of the QRF and its generalizations).
For the case of rank-1 projective measurements (in the computational basis), the QRF states that the multi-time probability distributions in Eq.~\eqref{eqn::Gen_Prob} can be equivalently expressed by
\begin{align}
\label{eqn::Markov1}
    &\Pprob_n(x_n,\dots,x_1) \\ \notag
    & = \tr\left[\Pcal_{x_n}\circ \Lambda_{t_n,t_{n-1}} \circ \cdots \circ \Lambda_{t_2,t_1} \circ \Pcal_{x_1} \circ \Lambda_{t_1,t_0}[\rho_{t_0}]\right]\,.
\end{align}
Importantly, this means that the full multi-time statistics can be obtained by means of maps that are independent of the respective previous measurement outcomes and which act on the system alone (see Fig.~\ref{fig::Markov_Dil} for a graphical representation). 

It is straightforward to see that satisfaction of the QRF (see Eq.~\eqref{eqn::Markov1}) implies Markovian statistics in the sense of Eq.~\eqref{eq::markov} and in particular we have the identities
\begin{align}
\braket{x_{k}|\Lambda_{t_{k},t_j}[\ketbra{x_j}{x_j}]|x_{k}} 
&= \Pprob(x_{k}|x_j)\quad \forall \ j\geq 1\, ,\label{eqn::extra1}\\
\text{and} \ \braket{x_1|\Lambda_{t_1,t_0}[\rho_{t_0}]|x_1} &= \Pprob(x_1)
\end{align}
In other words, the action of the propagators on the \textit{populations} (i.e., the diagonal terms of $\rho_{t_j}$, the state of the system at $t_j$) can be identified with the conditional probabilities between any two times. Crucially, this is not generally the case, and breaks down in situations where the QRF cannot be applied~\cite{Vacchini2011}. 

More generally, even if the QRF applies, the composition rule on the level of propagators does \textit{not} imply a composition rule on the level of the resulting measurement statistics, i.e., for a divisible process that satisfies the QRF, we generally have 
\begin{gather}
    \sum_{x_k} \Pprob(x_\ell|x_k) \Pprob(x_k|x_j)\neq \Pprob(x_\ell|x_j)\, ,
\end{gather}
which captures the deviation of quantum Markovian processes from classical ones. As mentioned previously, in order for the resulting process to be classical, not performing a measurement must be indistinguishable from performing a measurement and averaging over all possible outcomes. Put differently, for an observer that can only perform measurements in a fixed basis, the process is classical if they cannot detect the invasiveness of measurements in said basis. 

A measurement at time $t_j$ in the fixed basis where the measurement outcomes are averaged over
can be represented by the \textit{completely dephasing map}
\begin{gather}
    \Delta[\rho] = \sum_{x_j} \Pcal_{x_j} [\rho] = \sum_{x_j} \braket{x_j|\rho|x_j} \ketbra{x_j}{x_j}\,.
\end{gather}
The natural property of the propagators to look at in relation to classicality is thus that for all $t_j$:
\begin{align}
\label{eqn::NCGD}
&\Delta_{j+1}\circ\Lambda_{t_{j+1},t_j}\circ\Delta_{j} \circ \Lambda_{t_j,t_{j-1}}\circ \Delta_{j-1}
\\ \notag &= \Delta_{j+1}\circ\Lambda_{t_{j+1},t_j}\circ\Ical_{j} \circ \Lambda_{t_j,t_{j-1}}\circ \Delta_{j-1} \\ \notag
&= \Delta_{j+1}\circ\Lambda_{t_{j+1},t_{j-1}}\circ \Delta_{j-1}\, ,
\end{align} 
where $\Ical_j$ and $\Lambda_j$ are the identity map and the completely dephasing map at time $t_j$, respectively (see Fig.~\ref{fig::NCGD} for a graphical representation). In the last line of Eq.~\eqref{eqn::NCGD}
we used the composition law
$\Lambda_{t_{j+1},t_{j-1}} = \Lambda_{t_{j+1},t_j} \circ \Lambda_{t_j,t_{j-1}}$. Eq.~\eqref{eqn::NCGD} is, e.g., satisfied if none of the maps $\{\Lambda_{t_{j+1}, t_j}\}$ create coherences. More generally, each of the maps in Eq.~\eqref{eqn::NCGD} can in principle create coherences, as long as these coherences cannot be detected at the next time by means of measurements in the classical basis. Therefore, such a collection of maps satisfying Eq.~\eqref{eqn::NCGD} has been named \textit{non-coherence-generating-and-detecting} \textbf{(NCGD)}~\cite{smirne_coherence_2019}.
\begin{figure}
    \centering
    \includegraphics[width=0.85\linewidth]{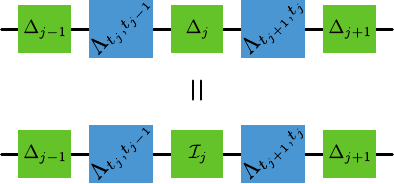}
    \caption{{\bf{NCGD dynamics.}} If the process is NCGD, then for a classical observer, `doing nothing' (i.e., performing the map $\Ical$) cannot be distinguished from a measurement in the classical basis and averaging over the outcomes (i.e., performing the map $\Delta$) at any point in time.}
    \label{fig::NCGD}
\end{figure}
The precise connection between NCGD and classicality is expressed by the following theorem: 
\begin{theorem}\label{thm::NCGD_classical}
	Let $\left\{\mathbb{P}_n(x_n,\ldots,x_1)\right\}_{n\leq K}$ be a $K$-Markovian process (Definition~\ref{def::markov}). 
	Then, the process is also $K$-classical (Definition~\ref{def::N-classical_statistics}) 
	if and only if there exist
	a system state $\rho_{t_0}$ (at a time $t_0\leq t_1$) which is diagonal 
	in the computational basis $\left\{\ket{x}\right\}_{x\in \mathcal{X}}$
	and a set of propagators $\left\{\Lambda_{t_j,t_{j-1}}\right\}_{j=1,\ldots,K}$ which are NCGD with respect to $\left\{\ket{x}\right\}_{x\in \mathcal{X}}$, 
	such that $\rho_{t_0}$ and $\left\{\Lambda_{t_j,t_{j-1}}\right\}_{j=1,\ldots,K}$ yield
	$\left\{\mathbb{P}_n(x_n,\ldots,x_1)\right\}_{n\leq K}$ via Eq.~\eqref{eqn::Markov1}.
\end{theorem}

\begin{proof}
We first show that if a Markovian process can be reproduced by means of NCGD propagators $\{\Lambda_{t_{j+1},t_j}\}$
and an initial diagonal state
(both properties with respect to the computational basis), then it yields classical statistics. If the statistics is Markovian, then it follows from Eq.~\eqref{eq::markov} that the joint probability distribution on any set of times $t_n \geq \ldots \geq t_1$, with $t_i\in \mathcal{T}$, is given by 
\begin{gather}
\label{eqn::Markov_Comp}
    \Pprob_n(x_n,\dots,x_1) = \Pprob(x_n|x_{n-1})\cdots\Pprob(x_2|x_1)\Pprob(x_1)\,.
\end{gather}
As the process can, by assumption, be reproduced by the maps $\{\Lambda_{t_{j},t_{j-1}}\}$
via Eq.~\eqref{eqn::Markov1}, then for any time $t_j$ we have
\begin{align}
&\sum_{x_j} \Pprob(x_{j+1}|x_{j})\Pprob(x_{j}|x_{j-1}) \\
\notag
&=\sum_{x_j}\tr\{\Pcal_{x_{j+1}} \circ \Lambda_{t_{j+1},t_j}[\Pi_{x_j}]\} \tr\{\Pcal_{x_j}\circ \Lambda_{t_j,t_{j-1}}[\Pi_{x_{j-1}}]\} \\  
\notag
&= \tr\{\Pcal_{x_{j+1}} \circ \Lambda_{t_{j+1},t_j}\circ \Delta_j \circ  \Lambda_{t_{j},t_{j-1}}[\Pi_{x_{j-1}}]\} \\
\notag
&= \tr\{\Pcal_{x_{j+1}} \circ \Lambda_{t_{j+1},t_{j-1}}[\Pi_{x_{j-1}}]\}\, , 
\end{align}
where we have set $\Pi_{x_{j}} = \ketbra{x_{j}}{x_{j}}$ and the NCGD property was used in the last line. This equation implies 
\begin{gather}
\label{eqn::composition}
   \sum_{x_j} \Pprob(x_{j+1}|x_{j})\Pprob(x_{j}|x_{j-1}) = \Pprob(x_{j+1}|x_{j-1})\,.
\end{gather}
Moreover, the (initial) diagonal state $\rho_{t_0}$ guarantees that we have
\begin{gather}
\label{eqn::initial}
   \sum_{x_1} \Pprob(x_{2}, x_{1}) = \Pprob(x_{2})\,.
\end{gather}

As a consequence of these two previous relations, the family of joint probability distributions computed via Eq.~\eqref{eqn::Markov_Comp} satisfies Kolmogorov conditions, and is thus classical.

Conversely, if the process is classical and Markovian, Eq.~\eqref{eqn::composition} holds. We can then define the maps 
\begin{gather}\label{eq:extra6}
    \widetilde{\Lambda}_{t_{j+1},t_{j}}[\ketbra{x_j}{y_j}] = \delta_{x_jy_j} \sum_{x_{j+1}} \Pprob(x_{j+1}|x_j) \Pi_{x_{j+1}}\, ,
\end{gather}
and the initial diagonal state
\begin{gather}
    \widetilde{\rho}_{t_0} = \sum_{x_1} \Pprob(x_{1}) \Pi_{x_{1}},
\end{gather}
which also means that we identify the initial time as the time of the first measurement, $t_1=t_0$.
 The set of maps $\{\widetilde{\Lambda}_{t_{j+1},t_j}\}$ defined in this way, in conjunction with $\widetilde{\rho}_{t_0}$, reproduces the correct statistics via Eq.~\eqref{eqn::Markov1}. As they are diagonal in the computational basis
 for any pair of times $t_j$ and $t_{j+1}$, they form an NCGD set.
\end{proof}
Crucially, the connection between classicality and NCGD dynamics is one-to-one: If the obtained Markovian statistics
cannot be reproduced by a set of maps that are NCGD, then the process is non-classical. Before discussing classicality in the presence of memory effects below, it is worth discussing the intuitive meaning of this theorem, and NCGD dynamics in particular. 

If the process at hand is Markovian and classical, the maps $\{\widetilde \Lambda_{t_{j+1},t_j}\}$
(as well as the initial state $\widetilde{\rho}_{t_0}$) introduced in the proof of Theorem~\ref{thm::NCGD_classical} define an \emph{artificial} reduced dynamics of the system, whose propagators correctly reproduce all joint probability distributions for measurements in the (fixed) classical basis via Eq.~\eqref{eqn::Markov1}. 
Note that the actual propagators of the dynamics (i.e., those fixed by the unitary evolution
in Eq.~\eqref{eqn::Gen_Prob} via Eqs.~\eqref{eq:map} and~\eqref{eq:prop}) might differ
from the maps $\widetilde{\Lambda}_{t_{j+1},t_{j}}$ above
(and $\widetilde{\rho}_{t_0}$ might differ from the actual initial state $\rho_{t_0}$); indeed, the fact that they do not coincide is simply a manifestation of the basis dependence of the (sequential) measurement scheme we are focusing on here. 

Crucially, a composition rule on the level of the actual propagators does not imply a composition rule on the level of the propagators of the populations. 
This implication only holds if the propagators of the dynamics are NCGD and the resulting statistics can be computed via Eq.~\eqref{eqn::Markov1}, in which case Eq.~\eqref{eqn::NCGD} results in 
\begin{gather}
    \widetilde{\Lambda}_{t_{j+1},t_{j-1}} = \widetilde{\Lambda}_{t_{j+1},t_{j}} \circ \widetilde{\Lambda}_{t_{j},t_{j-1}} \quad \forall \ t_j,\label{eq:widett}
\end{gather}
with
\begin{equation}
    \widetilde{\Lambda}_{t_{k},t_{j}}[\ketbra{x_j}{y_j}] = \delta_{x_jy_j} \braket{x_{k}|\Lambda_{t_{k},t_j}[\ketbra{x_j}{x_j}]|x_{k}} \Pi_{x_{j}}\,
\end{equation}
(see Eqs.~\eqref{eqn::extra1} and~\eqref{eq:extra6}). These reduced propagators still produce the correct populations, which is the only relevant part for the considered statistics, and set all coherences to zero.
This composition law is then---as already seen in Eq.~\eqref{eqn::composition}---equivalent to the well-known classical Chapman-Kolmogorov equations
\begin{gather}
    \sum_{x_j} \Pprob(x_{j+1}|x_j) \Pprob(x_j|x_{j-1}) = \Pprob(x_{j+1}|x_{j-1})\,,
\label{eqn::composition1}
\end{gather}
which hold for classical Markovian processes: If the measurement statistics of a Markovian process can be reproduced by a set of NCGD maps $\{\Lambda_{t_j,t_{j-1}}\}$, then it can also be reproduced by the set of maps $\{\widetilde \Lambda_{t_j,t_{j-1}}\}$, which act non-trivially on only the populations of the computational basis and satisfies a composition law, thus the process is classical.

Conversely, if the classical composition rule of Eq.~\eqref{eqn::composition1} holds for a Markovian process, then there exists a set $\{\widetilde{\Lambda}_{t_{j+1},t_j}\}$ of propagators (e.g., those defined in Eq.~\eqref{eq:extra6}) that are NCGD and correctly reproduce all joint probability distributions for measurements in the (fixed) classical basis.

Theorem~\ref{thm::NCGD_classical} is a generalization of the main result of Ref.~\cite{smirne_coherence_2019} in two ways. First, it does not impose any restriction on the propagators of the underlying quantum evolution, while in Ref.~\cite{smirne_coherence_2019} these were required to form a semigroup,
i.e., $\Lambda_{t_{j+1},t_{j}} = e^{\Lcal(t_{j+1} - t_j)}$, for some Lindbladian $\Lcal$~\cite{gorini1976,lindblad1976}. 

Second, the definition of Markovianity used here coincides with the standard definition of classical stochastic processes, whereas in Ref.~\cite{smirne_coherence_2019}, a definition based on Eq.~\eqref{eqn::Markov1} (for
semigroups) was used. Consequently, while the maps $\{\Lambda_{t_{j+1},t_j}\}$ cannot be fully probed by measurements in the computational basis alone, the requirement of Eq.~\eqref{eqn::composition1} can be tested for by simply performing sequences of measurements in the classical basis at the relevant times, thus making our theorem fully operational.
However, this comes at the cost of dealing with propagators 
$\{\widetilde{\Lambda}_{t_{j+1},t_j}\}$
which possibly do not correspond to those of the actual reduced dynamics.

On the other hand, as we show in Appendix~\ref{app::prev_results}, a one-to-one correspondence between the dynamical propagators $\Lambda_{t_{j+1},t_j}$
and the non-classicality of the multi-time statistics 
can be established also in the general
(non-semigroup)
divisible case, when the QRF applies, 
provided that one assumes a proper invertibility condition 
on the restriction of the dynamical maps to the populations of the computational basis. 
Indeed, this also allows one to recover in a straightforward way the main result
of Ref.~\cite{smirne_coherence_2019} as a corollary by further imposing the semigroup composition law. 

Importantly, Theorem~\ref{thm::NCGD_classical} characterizes the connection between coherences and the classicality of a Markovian process. While it is not necessary that the underlying propagators do not create coherences in order for a Markovian process to be classical, it is necessary and sufficient that coherences---should they be created---cannot be detected at a later point in time by means of measurements in the computational basis. Put differently, the propagators must be such that a classical observer could not decide whether at any point in time an identity map or a completely dephasing map was performed (which is depicted in Fig.~\ref{fig::NCGD}). This requirement is exactly encapsulated in the NCGD property of the propagators.

\subsection{Coherence in the non-Markovian case: preliminary analysis}

The above connection between quantum coherence and non-classicality fails to hold in the non-Markovian case. On the one hand, in this case propagators between two times are no longer sufficient to fully characterize the multi-time statistics~\footnote{For a characterization of non-Markovian processes in terms of \textit{collections} of CPTP maps (or sequences thereof), see Refs.~\cite{sakuldee_non-markovian_2018, paz-silva_dynamics_2018}. Notably, the characterization employed in these references is equivalent to the one provided here.}. On the other hand, even if the state of the system is diagonal in the computational basis at all times, dephasing can still be invasive due to correlations with the environment, breaking the connection between coherences and the classicality of statistics. We will discuss the former problem in the subsequent sections. Using an open system model from Refs.~\cite{lindblad_1980, accardi_quantum_1982,arenz_distinguishing_2015}, an explicit \textit{ante litteram} example of the latter case has already been provided in Ref.~\cite{smirne_coherence_2019} (note also a similar investigation in Ref.~\cite{budini_conditional_2019}), albeit not with an emphasis on the lack of coherence in the system state at all times (even in between the measurements). Here, we reiterate this example, focusing on the absence of coherences in the state of the system. The details of this discussion can be found in Appendices~\ref{app::Absence} and~\ref{app::CombPocket}. A simpler, although non-continuous, example for a non-Markovian process that yields non-classical statistics but never displays coherences in the system state is provided in Appendix~\ref{app::MIC_Ex}.

	\begin{example}
		\label{ex::class_state}
		\normalfont
		Let the system of interest $s$ consist of a qubit described by $\rho_s(t)$ which is coupled to a continuous degree of freedom $p$ of the environment. The global dynamics of system and environment is governed by the unitary evolution $U_{t_j,t_i}$, acting as
		\begin{equation}
		U_{t_j,t_i} \ket{\ell, p} = e^{i \phi_\ell p (t_j-t_i)} \ket{\ell, p}\, ,
		\end{equation}
		where $\left\{\ket{\ell}\right\}_{\ell=0,1}$ is the eigenbasis of the system Pauli operator $\hat{\sigma}_z$ and $\phi_\ell = (-1)^\ell$. The initial system-environment state is assumed to be of product form $\eta(0) = \rho_s(0)\otimes \ket{\varphi^e}\bra{\varphi^e}$, with
		$
		\ket{\varphi^e} = \int_{-\infty}^{\infty} d p f(p) \ket{p}$, where $f(p)$ satisfies the normalization condition $\int_{-\infty}^{\infty} d p |f(p)|^2 =1$. By defining 
		\begin{equation}\label{eq:kt}
		k(t) := \int_{-\infty}^{\infty} d p |f(p)|^2 e^{2 i p t}\, ,
		\end{equation}
		it is straightforward to show that, expressed in the eigenbasis of $\hat \sigma_z$, the free open evolution of the state of the system (i.e., without intermediate measurements) is given by 
		\begin{equation}
		\rho_s(t)=\begin{pmatrix}
		\rho_{0 0} &  k(t)\rho_{01}\\ 
		k^*(t)\rho_{10} & \rho_{11}
		\end{pmatrix}\, ,
		\end{equation}
		where $\rho_{mn}:= \braket{m|\rho_s(0)|n}$. 
		
		If $\rho_s(0)$ is initialized in a convex mixture of the eigenvectors $\{\ket{\pm} = (\ket{0}\pm\ket{1})/\sqrt{2}\}$ of the $\hat \sigma_x$ operator, i.e., $\rho_s(0) = \alpha\ket{+}\bra{+}+(1-\alpha)\ket{-}\bra{-}$, then 
		\begin{align}
		\rho_s(t)= \frac{1}{2} &\begin{pmatrix}
		1 &  k(t)(2\alpha -1)\\ 
		k^*(t)(2\alpha -1) &1
		\end{pmatrix}\nonumber\\
		= &\frac{1}{2} \Big\{ \ketbra{+}{+} \left[ 1+\left(2\alpha-1\right) \mbox{Re} \left(k(t)\right) \right] \nonumber\\
		&- \ketbra{+}{-} \left(2\alpha-1\right)\mbox{Im}\left(k(t)\right)\nonumber\\
		&+  \ketbra{-}{+} \left(2\alpha-1\right)\mbox{Im}\left(k(t)\right) \nonumber\\
		&+ \ketbra{-}{-} \left[ 1-\left(2\alpha-1\right) \mbox{Re} \left(k(t)\right) \right] \Big\}, \label{eq:st}
		\end{align}
		i.e., no coherence w.r.t. $\hat{\sigma}_x$ will be generated if $k(t)$ is a real function of time (as noted in Ref.~\cite{smirne_coherence_2019}); this is, e.g., the case if $f(p)$ corresponds to a Lorentzian distribution centered around zero,
		\begin{equation}
		|f(p)|^2 = \frac{\Gamma}{\pi(\Gamma^2+p^2)} \mapsto k(t)= e^{-2 \Gamma |t|}\, .
		\end{equation} 
		\textit{A priori}, the fact that there are no $\hat{\sigma}_x$-coherences created in the free evolution does not mean that none are created if the system is probed at intermediate times. However, here, \textit{no} $\hat{\sigma}_x$-coherence is generated even when we take into account how the measurements modify the system's state. Specifically, immediately after a measurement in the $\hat \sigma_x$-basis is performed at time $t_1$ (yielding outcome $\pm$), the total system-environment state is of product form 
		\begin{gather}\label{eq:extra3}
		\eta^{(\pm)}(t_1) = \ketbra{\pm}{\pm} \otimes \xi^{(\pm)}(t_1)\, ,
		\end{gather}
		where $\xi^{(\pm)}(t_1)$ is a state of the environment that depends on the measurement outcome. As we show in Appendix~\ref{app::Absence}, any state of the system
		evolved from the post-measurement state of Eq.~\eqref{eq:extra3} according to the described dynamics remains diagonal in the $\{\ket{\pm}\}$ basis; this also holds true for the state of the system after any \textit{sequence} of such measurements. Together with the fact that the statistics resulting from measurements in the $\{\ket{\pm}\}$ basis  is non-classical (i.e., it does not satisfy Kolmogorov conditions, as has been shown in Ref.~\cite{smirne_coherence_2019}), this constitutes an example of a non-classical process without any coherence with respect to the measured observable ever being generated. Evidently, this behavior is only possible since the chosen example is non-Markovian. 
	\end{example}
Unlike in the Markovian case, where the absence of coherences trivially leads to classical statistics, when memory effects are present, it is the coherences of the system state \textit{as well as} the non-classical correlations between the system and its environment that can lead to non-classical behavior---in a way which will be specified in the following. Intuitively, while the completely dephasing map leaves the system unchanged if no coherences are created, it does not necessarily leave the \textit{overall} system-environment state invariant. In detail, in general we can have $\Delta[\rho^s_{t_j}] = \Ical[\rho^s_{t_j}] \ \forall \ t_j$, without it implying $\Delta\otimes \Ical^e[\eta^{se}_{t_j}] = \Ical[\eta^{se}_{t_j}] \ \forall \ t_j$. As we will see, the latter property is sufficient, but not necessary, for the satisfaction of the Kolmogorov conditions. First, though, in order to be able to go beyond the investigation of Markovian processes, and extend the existing connection between classicality and coherences, it is important to introduce \textit{quantum combs}---a suitable framework to describe general quantum processes~\cite{chiribella_theoretical_2009,pollock_non-markovian_2018}.

\section{non-Markovian classical processes}\label{sec::combs}

The previous example illustrates the subtle relation between coherence and classicality in the case of open quantum processes with memory. There, although no coherence is ever generated on the level of the system with respect to the chosen measurement basis, the system-environment correlations built up throughout the dynamics lead to non-classical statistics. To develop a more in-depth understanding of the interplay between coherences and classical phenomena, we require a suitable operational framework for approaching such scenarios. We can then employ this framework to comprehensively characterize \textit{all} quantum processes that display classical statistics. 

\subsection{Classicality and processes with memory}
The necessity of such a novel framework for the description of quantum processes that display memory effects stems from the breakdown of their modeling in terms of propagators that could be used in the Markovian case; this can already be seen for classical stochastic processes. Here, a joint probability distribution $\Pprob_K(x_K,\dots,x_1)$ fully describes a $K$-process. This probability distribution can equivalently be represented in terms of multi-time conditional probabilities as 
\begin{align}
&\Pprob_K(x_K,\dots,x_1) \\
\notag
&= \Pprob_{K}(x_{K}|x_{K-1},\dots,x_1)\cdots \Pprob_2(x_2|x_1) \Pprob_1(x_1)\, .
\end{align} 
Importantly, all of the above conditional probabilities generally depend upon \textit{all} preceding measurement results, in contrast to the Markovian case where they only depend on the most recent outcome. Consequently, two-point transition probabilities of the form $\Pprob(x_j|x_{j-1})$ are not sufficient in general to build up all joint probability distributions and thus do not completely describe the process. Similarly, two-time propagators $\{\Lambda_{t_j,t_{j-1}}\}$ are generally not sufficient to compute multi-time joint probabilities in the quantum case and therefore fail to fully characterize the process~\cite{Vacchini2011,rivas_quantum_2014}.

For classical statistics, the joint probability distribution $\Pprob_K(x_K,\dots,x_1)$ contains all information about the $K$-process, since all distributions for fewer times, as well as all conditional probabilities, can be derived once $\Pprob_K$ is known. In exactly the same way, a general quantum $K$-process is fully characterized by the joint probabilities for \emph{all possible sequences} of $K$ measurements (at times $t_1,\dots,t_K$), including non-projective and non-orthogonal ones.

As discussed in the previous section, if the complete system-environment dynamics is known, then all joint probability distributions (on times $\{t_j\}_{j=1}^n$) obtained from sequential measurements of the system can be computed via
\begin{align}
\label{eq:multitimeoutputstate}
   &\Pprob_n(x_n,\dots,x_1) \\ 
  \notag &= \tr\left[(\Pcal_{x_n}\otimes \Ical^e)\circ \Ucal_{t_n,t_{n-1}} \circ \cdots \circ (\Pcal_{x_1}\otimes \Ical^e) [\eta^{se}_{t_1}]\right]\, .
\end{align}
Here, $\{\Pcal_{x_j}\}$ correspond to projective measurements in the computational basis, but
evidently the same relation can also be used to obtain the correct probabilities when using different probing \textit{instruments}, e.g., instruments that measure sharply in a different basis or those that perform generalized measurements. More formally, an instrument $\Jcal_{k} = \{ \Mcal_{x_k} \}$ (at time $t_k$) is a collection of CP maps that add up to a CPTP map~\cite{Heinosaari2012}. For instance, the instrument corresponding to a measurement in the computational basis is given by $\Jcal_k = \{\Pcal_{x_k}\}$, and all of its elements add up to the CPTP map $\sum_{x_k} \Pcal_{x_k} = \Delta_k$. Intuitively, each outcome of an instrument corresponds to one of its constituent CP maps, which, in turn, describes how the state of the system changes upon the realization of a specific measurement outcome. With this, the probability to obtain the sequence of outcomes $x_1,\dots, x_K$, given that the instruments $\Jcal_1,\dots, \Jcal_K$ were used to probe the system, is given by
\begin{align}\label{eqn::combintroprob}
    &\Pprob_K(x_K,\dots,x_1|\Jcal_K,\dots,\Jcal_1) \\ 
    \notag &= \tr\left[(\Mcal_{x_K}\otimes \Ical^e)\circ \Ucal_{t_K,t_{K-1}} \circ \cdots \circ (\Mcal_{x_1}\otimes \Ical^e) [\eta^{se}_{t_1}]\right]\, \\
    \notag &=: \Ccal_K[\Mcal_{x_K},\dots,\Mcal_{x_1}]\,;
\end{align}
indeed the joint probability distribution for any subset of ordered times $t_n \geq \ldots \geq t_1$, with $n<K$,
can be obtained by replacing in the formula above $\Mcal_{x_j}$ with the identity operator, in
correspondence with the times $t_j$ not contained in the subset.

In what follows, whenever we drop the explicit instrument labels, it is understood that the probabilities were the result of a measurement in the computational basis at each time. The multi-linear functional $\Ccal_K$ introduced above is a special case~\footnote{In contrast to the combs discussed in Refs.~\cite{chiribella_quantum_2008,chiribella_theoretical_2009}, the combs we consider do not start on an open input line, and do not end on an open output line; or, equivalently, in our case, the Hilbert spaces of this initial input and final output space are trivial. Such combs are also called \textit{testers} in the literature.} of a \textit{quantum comb}~\cite{chiribella_quantum_2008,chiribella_theoretical_2009} and provides a natural generalization to the concept of quantum channels that by construction allows for the inclusion of memory effects.~\cite{Kretschmann2005,Caruso2014,pollock_operational_2018,pollock_non-markovian_2018} (see Fig.~\ref{fig::general_comb} for a graphical representation). 
\begin{figure}
    \centering
    \includegraphics[width=0.9\linewidth]{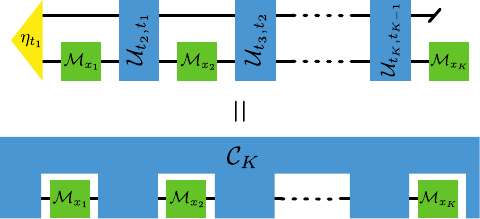}
    \caption{{\bf{Comb of a general open quantum evolution.}} The probabilities characterizing a quantum process can be understood as the action of a comb $\Ccal_K$ on the sequence of CP maps $\{\Mcal_{x_j}\}$ that correspond to the respective measurement outcomes.}
    \label{fig::general_comb}
\end{figure}
It maps \emph{any} sequence of possible experimental transformations enacted on the system to the corresponding joint probability of their occurrence. In this sense, $\Ccal_K$ plays exactly the same role that the joint probability distribution $\Pprob_K$ plays in the classical setting, and thus allows one to decide on the classicality of the resulting statistics. For example, for the completely memoryless case, i.e, the case of Markovianity with respect to measurements in any basis, the evolution between any two points in time is described solely by a sequence of independent CPTP maps that act on the system alone~\cite{pollock_operational_2018,1367-2630-18-6-063032}, and we have
\begin{align}
\label{eqn::Markov_Comb}
     &\Ccal^{\text{Markov}}_K[\Mcal_{x_K},\dots,\Mcal_{x_1}] \\
     \notag &= \tr\left[\Mcal_{x_K} \circ \Lambda_{t_K,t_{K-1}}\circ \cdots \circ \Mcal_{x_2}\circ \Lambda_{t_2,t_{1}}\circ \Mcal_{x_1}[\rho_{t_1}]\right]\, .
\end{align}
In general, however, the comb of a $K$-process does not split in the way above into independent portions of evolution between times. Thus, when analyzing the relation between coherence and classicality in the presence of memory, instead of investigating the properties of individual CPTP maps, one must consider those of the multi-time comb $\Ccal_K$. 

The comb $\Ccal_K$ is an operationally well-defined object that can---just like the joint probability distribution $\Pprob_K$---be obtained by means of probing measurements on the system alone through a generalized tomographic scheme~\cite{pollock_non-markovian_2018, milz_reconstructing_2018}. Specifically, for its reconstruction, it is \textit{not} necessary to explicitly know the system-environment dynamics: the comb does not contain direct information about the environment, but solely that of its influence on the multi-time statistics observed from measurements on the system. As such, it encapsulates all that is out of control of the experimenter
and thereby clearly separates the underlying process at hand from what can be controlled (i.e., the experimental interventions). An explicit example of the comb formalism is provided in Appendix~\ref{app::CombPocket}, where we rephrase Example~\ref{ex::class_state} in terms of the comb description.

Crucially, the comb framework allows us to consider what it means for a stochastic process \emph{with memory} to be classical, thereby permitting an extension of the results of Ref.~\cite{smirne_coherence_2019} to the non-Markovian case: given the comb $\Ccal_K$ of a process on times in $\Tcal$, \textit{all} combs correctly describing the process on fewer times $\Tcal' \subseteq \Tcal$ can be deduced by letting $\Ccal_K$ act on the identity map at the appropriate superfluous times~\cite{pollock_non-markovian_2018, milz_kolmogorov_2017}. For example, we have (see also Fig.~\ref{fig::marginal_gen})
\begin{align}
\label{eqn::Marginal_Quantum}
     &\Ccal_{K-1}[\Mcal_{x_K},\dots,\Mcal_{x_{j+1}},\Mcal_{x_{j-1}},\dots,\Mcal_{x_1}] \notag \\
     &= \Ccal_{K}[\Mcal_{x_K},\dots,\Mcal_{x_{j+1}},\Ical_j,\Mcal_{x_{j-1}},\dots,\Mcal_{x_1}] \, .
\end{align}
\begin{figure}
    \centering
    \includegraphics[width=1\linewidth]{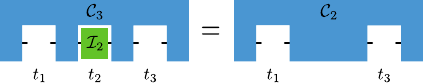}
    \caption{{\bf{Consistency condition for combs.}} Letting a comb defined on times $\Tcal$ act on identity maps at a set of times $\Tcal\setminus \Tcal'$ (i.e., the set of times in $\Tcal$ but not in $\Tcal'$) yields the correct comb on $\Tcal'$. Depicted is the situation for $\Tcal = \{t_1,t_2,t_3\}$ and $\Tcal' = \{t_1,t_3\}$.}
    \label{fig::marginal_gen}
\end{figure}
As we have discussed in the previous sections, classicality of a process means that the action of the completely dephasing map cannot be distinguished (by means of measurements in the classical basis) from not performing an operation. With the method of `generalized marginalization' given by Eq.~\eqref{eqn::Marginal_Quantum} at hand, we obtain the following characterization of classical combs:
\begin{theorem}[$K$-classical quantum combs]
\label{thm::KclassComb}
A comb $\mathcal{C}_K$ on times $\Tcal$, with $|\Tcal| = K$, yields a $K$-classical process via Eq.~\eqref{eq:multitimeoutputstate} iff it satisfies
\begin{align}
\label{eqn::Class_Comb}
    &\Ccal_K\left[\bigotimes_{t_j \in \Tcal'} \Ical_j, \bigotimes_{t_k\in \Tcal \setminus \Tcal'} \Pcal_{x_k} \right] \\ \notag
    &=\Ccal_K \left[\bigotimes_{t_j \in \Tcal'} \Delta_j, \bigotimes_{t_k\in \Tcal \setminus \Tcal'} \Pcal_{x_k}\right]\, ,
\end{align}
for all subsets $\Tcal'\subseteq \Tcal$ and all possible sequences of outcomes on $\Tcal\setminus \Tcal'$. 
\end{theorem}
In slight abuse of notation, here, the argument $\bigotimes_{t_j \in \Tcal'} a_j, \bigotimes_{t_k\in \Tcal \setminus \Tcal'} b_{x_k}$ of the comb $\Ccal_K$ signifies that it acts on the maps $a_j$ at times $t_j\in \Tcal'$ and on $b_{x_k}$ at times $t_k\in  \Tcal \setminus \Tcal'$. 

Theorem~\ref{thm::KclassComb} expresses in a concise way that a general process is $K$-classical iff measurements in the computational basis cannot distinguish the action of completely dephasing maps from the action of identity maps. Let us emphasis again that the completely dephasing map does not only destroy coherences of the systems reduced state, but also quantum correlations between the system and the environment. Therefore, Theorem~~\ref{thm::KclassComb} does not directly link coherence and non-classicality as Theorem~\ref{thm::NCGD_classical} did for the case without memory.
\begin{proof}
The proof of Theorem~\ref{thm::KclassComb} is thus straightforward: If a comb satisfies Eqs.~\eqref{eqn::Class_Comb}, then the resulting statistics satisfy Kolmogorov conditions. Conversely, any joint probability distribution on a set of times $\Tcal'\subseteq \Tcal$ can either be obtained by direct measurement, or by marginalization of the corresponding distribution on $\Tcal$. The former can be computed via the first line of Eq.~\eqref{eqn::Class_Comb}, the latter via the second one. If the statistics of the process appear classical, then both resulting distributions have to coincide, and Eq.~\eqref{eqn::Class_Comb} must hold. 
\end{proof}

In the (basis dependent) Markovian case that we discussed in the previous section, Eq.~\eqref{eqn::Class_Comb} directly reduces to Eq.~\eqref{eq:widett},
the NCGD property at the level of propagators of populations. Theorem~\ref{thm::KclassComb} therefore provides the proper generalization of the results of Ref.~\cite{smirne_coherence_2019} to the non-Markovian case. Nonetheless, its consequences for the structural properties of classical combs, and, in particular, the relation of classicality and coherence remain somewhat opaque in the way Theorem~\ref{thm::KclassComb} is phrased. In order to address these questions, we now  introduce a representation of quantum combs that is favorable for the purposes of our work.

\subsection{Choi-Jamio{\l}kowski representation of general quantum processes}
Both the quantum comb describing the $K$-process at hand and the experimental interventions applied at each time are linear maps (the former being a higher-order multi-linear map). Any such map can be represented in a variety of ways, but the most natural for our present purposes makes use of the Choi-Jamio{\l}kowski isomorphism~\cite{jamiolkowski_linear_1972,Choi1975} between quantum maps and positive semi-definite Hermitian matrices. 

A general quantum map---e.g., one that corresponds to a generalized measurement---at time $t_k$ is a CP transformation $\Mcal_{x_k}: \Bcal(\Hcal_k^\inp) \rightarrow \Bcal(\Hcal_k^\out)$ that takes bounded linear operators on the (input) Hilbert space $\Hcal_k^\inp$ onto bounded linear operators on the (output) Hilbert space $\Hcal_k^\out$. Throughout this paper, we will consider the input and output spaces of such maps to be isomorphic (and of finite dimension), and the labels $\inp$ and $\out$, as well as the time label, are merely introduced for better accounting of the involved spaces. Any such quantum map $\Mcal_{x_k}$ can be isomorphically mapped onto a positive semi-definite Hermitian matrix that we will call its \textit{Choi state}, $M_{x_k} \in \Bcal(\Hcal_k^\out \otimes \Hcal_k^\inp)$, by letting it act on one half of an unnormalized maximally-entangled state $\Phi^+ = \sum_{x_k,y_k} \ketbra{x_kx_k}{y_ky_k} \in \Bcal(\Hcal_k^\inp \otimes \Hcal_k^\inp )$, i.e.,
\begin{gather}
    \label{eqn::Choi}
    M_{x_k} := (\Mcal_{x_k} \otimes \Ical )[\Phi^+] \in \Bcal(\Hcal_k^\out \otimes \Hcal_k^\inp).
\end{gather}
This isomorphism implies, e.g., the following identifications:
\begin{align}
    \label{eqn::Ident_map}
    &\text{Identity Map}: &&\Ical_k \Leftrightarrow \Phi^+_k\,,  \\
    \label{eqn::proj}
    &\text{Proj. Map}: &&\Pcal_{x_k} \Leftrightarrow \ket{x_k} \bra{x_k} \otimes \ket{x_k} \bra{x_k}\,, \\
    \label{eqn::Comp_Deph_Map}
    &\text{C. Deph. Map}: &&\Delta_k \Leftrightarrow \sum_{x_k} \ket{x_kx_k} \bra{x_kx_k}:=D_k. 
\end{align}
Here and throughout this article, we typically denote maps with calligraphic upper-case letters (as we have already done above) and their Choi state with the corresponding non-calligraphic variant---with the exception of the identity map (Eq.~\eqref{eqn::Ident_map}) and the completely dephasing map (Eq.~\eqref{eqn::Comp_Deph_Map}). For better orientation, we will continue to denote the respective time at which the maps act by an additional subscript.

Analogously, as a quantum comb $\Ccal_K$ is a multi-linear map it can---in a similar way to Eq.~\eqref{eqn::Choi}---be mapped onto a positive semi-definite Hermitian matrix $C_K$~\cite{chiribella_theoretical_2009, pollock_non-markovian_2018, milz_introduction_2017}. The action of a quantum comb on a sequence of CP maps $\{\Mcal_{x_K},\hdots,\Mcal_{x_1}\}$ is then equivalently given by~\cite{chiribella_theoretical_2009}
\begin{gather}
\label{eqn::Born_rule}
    \Ccal_K[\Mcal_{x_K},\dots,\Mcal_{x_1}] = \tr\left[(M_{x_K}^\mathrm{T} \otimes \cdots \otimes M_{x_1}^\mathrm{T})\, C_K\right]\, ,
\end{gather}
where $\sbt^\mathrm{T}$ denotes the transposition with respect to the computational basis. Eq.~\eqref{eqn::Born_rule} constitutes the Born rule for temporal processes~\cite{chiribella_memory_2008,ShrapnelCostaBorn2017}, where $C_K$ plays the role of a quantum state over time and the Choi states $M_{x_K}, \dots, M_{x_1}$ play the role that positive operator-valued measure \textbf{(POVM)} elements play in the standard Born rule.

Concretely, given an instrument sequence $\mathcal{J}_K , \hdots, \mathcal{J}_1$, by combining Eqs.~\eqref{eqn::combintroprob} and~\eqref{eqn::Born_rule}, the joint probability over the sequence of outcomes $x_K, \hdots, x_1$ is given by
\begin{align} \label{eq:extra4}
    &\mathbbm{P}_K(x_K, \hdots, x_1|\mathcal{J}_K, \hdots, \mathcal{J}_1) \\ 
    &\phantom{\mathbbm{P}_K}= \tr\left[(M_{x_K}^\mathrm{T} \otimes \cdots \otimes M_{x_1}^\mathrm{T})\, C_K\right]\, . \notag
\end{align}

Through this isomorphism, memory effects of the temporal process correspond directly to structural properties of its Choi state~\cite{milz_introduction_2017,taranto_2019L,taranto_2019A,taranto_2019S,TarantoThesis}; analogously, the classicality of a process is reflected in the properties of $C_K$. 

Represented in this way, quantum combs and the channels that they generalize have particularly nice properties. Complete positivity and trace preservation for a quantum channel $\Mcal$ correspond respectively to $M \geq 0$ and satisfaction of $\ptr{\out}{[M]} = \mathbbm{1}_\inp$. Analogously the Choi state of a quantum comb has to satisfy $C_K\geq 0$ as well as a hierarchy of trace conditions that fix the causal ordering of events~\cite{chiribella_theoretical_2009}, i.e., they ensure that later events cannot influence the statistics of earlier ones.

It is important to note that \textit{all} $K$-processes can be represented through the Choi-Jamio{\l}kowski isomorphism as (unnormalized) quantum states $C_K$. In the converse direction, any operator satisfying the aforementioned properties admits an underlying open quantum dynamics description~\cite{chiribella_quantum_2008,chiribella_theoretical_2009,pollock_non-markovian_2018}. Specifically, this means that for every proper comb, there is a (possibly fictitious) environment and a set of system-environment unitaries such that the action of the comb
on any sequence of instruments can be written as in Eq.~\eqref{eqn::combintroprob}. 
Quantum combs are hence the most general descriptors of open quantum system processes (when
the system of interest is probed at fixed times). We will call the respective underlying unitary description
that includes the environment the \textit{dilation} of the comb. As is the case for quantum channels, any such dilation is non-unique. On the other hand, the comb $\Ccal_K$ resulting from some underlying evolution is unique, and---just like the joint probability distribution $\Pprob_K$ in the classical case---constitutes the maximal descriptor of the process on the respective set of times.

\subsection{Structural properties of classical combs}
As a first step to a structural understanding of classical combs, we rephrase Theorem~\ref{thm::KclassComb} in terms of Choi states:
\begin{thmbis}{thm::KclassComb}[$K$-classical quantum combs]
\label{thm::KclassComb_prime}
A comb $\mathcal{C}_K$ on times $\Tcal$, with $|\Tcal| = K$, yields a $K$-classical process
iff its Choi state satisfies 
\begin{align}
\label{eqn::Class_Comb_prime}
    &\tr\left[\left(\bigotimes_{t_j \in \Tcal'} \Phi^+_j \bigotimes_{t_k\in \Tcal \setminus \Tcal'} P_{x_k} \right) C_K \right]
    \\
    \notag
    &=\tr\left[\left(\bigotimes_{t_j \in \Tcal'} D_j \bigotimes_{t_k\in \Tcal \setminus \Tcal'} P_{x_k} \right) C_K \right]\, . 
\end{align}
for all subsets $\Tcal'\subseteq \Tcal$ and all possible sequences of outcomes on $\Tcal\setminus \Tcal'$. 
\end{thmbis}
Using the relations~\eqref{eqn::Ident_map} --~\eqref{eqn::Comp_Deph_Map} as well as Eq.~\eqref{eq:extra4}, it is straightforward to see that this theorem is indeed equivalent to Theorem~\ref{thm::KclassComb}. Importantly, as it is stated in terms of Choi states, Theorem~\ref{thm::KclassComb_prime} allows one to derive a direct connection between general correlations and the classicality of a $K$-process. 

To see how the requirement in Eq.~\eqref{eqn::Class_Comb_prime} translates to structural constraints on classical combs, first note that any comb that yields the joint probability distribution $\Pprob_K(x_K,\dots,x_1)$ when probed in the classical basis can be written as 
\begin{gather}
\label{eqn::chiterm}
    C_K = \widetilde{C}_K^\mathrm{Cl.} + \chi\, ,
\end{gather}
where the term 
\begin{gather}
    \label{eqn::classicalComb}
    \widetilde{C}_K^\mathrm{Cl.} = \sum_{x_K,\dots,x_1} \Pprob_K(x_K,\dots,x_1) P_{x_K} \otimes \cdots  \otimes P_{x_1}\, ,
\end{gather}
contains the joint probability distribution $\Pprob_K$ on its diagonal and $\tr[(P_{x_K}\otimes \cdots \otimes P_{x_1}) \chi] =0$ for all $x_K, \dots, x_1$~\cite{milz_reconstructing_2018}. Intuitively, $\widetilde{C}_K^\mathrm{Cl.}$ corresponds to the part of $C_K$ that can be probed by measurements in the classical basis alone, while $\chi$ contains all the information about the underlying process that such measurements are blind to. If $\chi = 0$, then $C_K$ clearly satisfies the conditions of Eq.~\eqref{eqn::Class_Comb}, as $\tr[P_{x_j} \Phi^+_j] = \tr[P_{x_j} D_j]$ for all $x_j$~\footnote{For $\chi=0$, $C_K$ is actually not a proper comb, as it does not satisfy the hierarchy of trace conditions that ensure causal ordering. Nonetheless, this lack of causality could not be picked up by means of projective measurements in the classical basis alone, and does thus not pose a problem for our discussion}. In words, for $\chi=0$, the corresponding comb is classical, as it is diagonal in the classical product basis. However, this is not necessary for Eq.~\eqref{eqn::Class_Comb} to hold; rather, it suffices if $\chi$ is such that it does not allow one to distinguish between the action of the identity map and the completely dephasing map. We thus arrive at the following lemma: 

\begin{lemma}
\label{lem::coherence_non_Markov}
Let $C_K$ be the comb of a $K$-process on $\Tcal$, with $|\Tcal| = K$, and let $A_j:= \Phi^+_j - D_j$. $C_K$ yields a $K$-classical process iff it is of the form
\begin{gather}
\label{eqn::classLem}
    C_K = \widetilde{C}_K^{\mathrm{Cl.}} + \chi\, ,
\end{gather}
where $\widetilde{C}_K^{\mathrm{Cl.}}$ is obtained from some joint probability distribution $\Pprob_K$ via Eq.~\eqref{eqn::classicalComb} and $\chi$ satisfies
\begin{gather}
\label{eqn::chi_requirement}
    \tr\left[\left(\bigotimes_{t_j \in \Tcal'} A_j \bigotimes_{t_k\in \Tcal \setminus \Tcal'} P_{x_k} \right) \chi \right] =0
\end{gather}
for all subsets $\Tcal'\subseteq \Tcal$ and $\Tcal' = \emptyset$. 
\end{lemma}
\begin{proof}
It is straightforward to see that a comb of the form of Eq.~\eqref{eqn::classLem} satisfies Eq.~\eqref{eqn::Class_Comb_prime}, whenever $\chi$ fulfills Eq.~\eqref{eqn::chi_requirement}, and thus yields $K$-classical statistics. Conversely, \emph{any} comb $C_K$ on $K$ times can be written as $C_K = \widetilde{C}_K^{\mathrm{Cl.}} + \chi\,$, where $\widetilde{C}_K^{\mathrm{Cl.}}$ is of the form of Eq.~\eqref{eqn::classicalComb} for some $\Pprob_K$ and $\tr[(P_{x_K}\otimes \cdots \otimes P_{x_1}) \chi] =0$~\cite{milz_reconstructing_2018}. When measuring (in the computational basis) at $K$ times, the resulting joint probability distribution is given by $\Pprob_K$. As, by assumption, the process is classical, summation over outcomes obtained at any time $C_K$ is defined on must yield the same statistics as letting the comb act on the identity channel at this time. As this has to hold for any collection of times in $\Tcal$, $\chi$ has to satisfy the additional requirements given by Eq.~\eqref{eqn::chi_requirement}. 
\end{proof}
Intuitively, Eq.~\eqref{eqn::chi_requirement} ensures that the action of $\Delta_j$ cannot be detected at any point in time by means of measurements in the classical basis. Therefore Lemma~\ref{lem::coherence_non_Markov} is equivalent to Theorem~\ref{thm::KclassComb}. However, the former provides an explicit constraint on the structure of such combs that contain coherences that can be present in the process without making the resulting statistics non-classical. 

Indeed, if $\chi = 0$, then the corresponding comb $C_K$ is diagonal in the classical product basis and, as such, cannot create coherences and destroys any kind of coherences that could be fed into the process (e.g., by performing coherence creating operations at some time). On the other hand, if $\chi\neq 0$ and the comb contains off-diagonal terms (with respect to the classical basis), then coherences can be created over the course of the process. However, if $\chi$ satisfies Eq.~\eqref{eqn::chi_requirement}, then these coherences---or rather the invasiveness of the completely dephasing map---cannot be detected at any later time by measurements in the classical basis. This understanding of classical non-Markovian combs mirrors the intuition we had built in the Markovian setting for the case of NCGD dynamics. Consequently, Lemma~\ref{lem::coherence_non_Markov} fully characterizes the relation between coherences and the non-classicality of a process (see Fig.~\ref{fig::Venn} for a graphical representation of the different sets of processes we consider).
\begin{figure}
    \centering
    \includegraphics[width=0.75\linewidth]{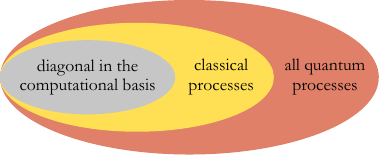}
    \caption{{\bf{Nested set of processes.}} Processes that cannot produce coherence and destroy any coherence that is fed in (i.e., their Choi states are diagonal in the computational basis) form a strict subset of processes that appear classical when sequentially probed in the computational basis. Both of these sets, as well as the set of all quantum processes, are convex.}
    \label{fig::Venn}
\end{figure}

Somewhat unsurprisingly, the above lemma implies that combs leading to classical processes are of measure zero in the set of all combs: while any comb can be written in the form of Eq.~\eqref{eqn::classLem}, Eq.~\eqref{eqn::chi_requirement} places further linear constraints on the $\chi$ term, which must be satisfied by combs leading to classical processes, but not by general combs. The set of combs leading to classical processes is thus confined to a lower dimensional subset, implying that it is of zero measure (with respect to any reasonable measure in the set of all non-Markovian combs). This fact falls in line with the intuition built above; for a randomly chosen comb, the action of a completely dephasing map in a given basis will generally be detectable. Furthermore, the vanishing volume of classical combs within the set of all combs mirrors the analogous property in the spatial setting: There, quantum states that display no discord are of measure zero in the set of all bipartite quantum states~\cite{ferraro_almost_2010} (the relation between quantum discord and classicality of processes is discussed in detail in Sec.~\ref{sec::DiscClass}).

 In the non-Markovian case, the characterization of classical processes comes at a price. In order to decide on the $K$-classicality of a given process, it is no longer sufficient to investigate propagators between pairs of times, but rather the full part of the comb $C_K$ that is relevant for sequential projective measurements must be known, due to the importance of multi-time effects. However, this behavior is to be expected, as can already be seen in the case of classical stochastic processes: the full characterization of a non-Markovian process only happens on the level of the full joint probability distribution $\Pprob_K$, and not by way of transition probabilities between adjacent times only. Despite the additional complexity brought in by the presence of memory, as we will see in the following section, measures for classicality that are both experimentally and computationally accessible can be derived based on the characterization of classical processes we have provided.

\subsection{Quantifying non-classicality}
\label{sec::Quant_non_class}
As we have seen above, the set of combs leading to classical processes is of measure zero in the set of all combs. Importantly though, this fact does not render our original definition of classicality meaningless, but rather---in conjunction with Lemma~\ref{lem::coherence_non_Markov}---allows for the derivation of a meaningful measure of non-classicality that is experimentally accessible and can be formulated by means of a linear program \textbf{(LP)}.

More specifically, we can exploit the characterization of classical processes provided by Eqs.~\eqref{eqn::classLem} and~\eqref{eqn::chi_requirement} in order to define a measure of non-classicality with a clear operational meaning. Such a measure not only classifies whether or not a comb is non-classical, but also quantifies the degree to which it is. This is crucial when assessing whether any potential non-classicality arises from inherently quantum features of the experiment or from experimental errors. In order to clarify its operational interpretation, we formulate our measure in the context of a game with two adversaries, Alice and Bob, and one referee, Rudolph. The task of Alice is to construct a classical stochastic process that is a good model for a comb she receives from Rudolph. The task of Bob is to design a test that distinguishes this model from the original comb. Let $C$ be the given comb in its Choi representation (i.e., a positive operator with some additional causality constraints). The game then proceeds as follows:

\indent \textbf{0)} Rudolph begins with a given comb $C$ and sends its description to both Alice and Bob.\\
\indent \textbf{A)} Alice prepares a classical process $C^{\textup{Cl.}}$ and sends it to Rudolph.\\
\indent \textbf{R1)} Rudolph sends the description of the classical process $C^{\textup{Cl.}}$ prepared by Alice to Bob.\\
\indent \textbf{B1)} Bob prepares a testing sequence $\{T_i(\vec{x})\}_{\vec{x}}$ and sends it to Rudolph.\\
\indent \textbf{R2)} Rudolph takes randomly either $C$ or $C^{\textup{Cl.}}$ and applies the testing sequence chosen by Bob. He yields an outcome $\vec{x}$, which he announces.\\
\indent \textbf{B2)} Bob announces whether the comb is $C$ or $C^{\textup{Cl.}}$\\
\indent \textbf{R3)} Rudolph announces whether Bob is correct or not and hence who wins the game.

Let us recall at this point that our definition of classicality relies exclusively on the statistics obtained by probing the process with projective measurements in fixed, orthonormal bases. Therefore, to only probe what is relevant within our framework, we restrict the testing sequences that Bob is allowed to prepare to only involve such measurements, i.e., the testing sequence must be of the form $T_i(\vec{x})=\bigotimes_{t_j\in \tau_i}\Phi^{+}_j \bigotimes_{t_k\in \tau_i^c} P_{x_k}$.
The figure of merit that we are interested in is the probability for Bob to win if both players play optimally. This is an operational quantity describing how well said comb can be distinguished from its best classical approximation, given that one has only access to the aforementioned restricted testing strategies that can be used to probe classicality. Making use of the arguments of Lemma~\ref{lem::coherence_non_Markov} to simplify the structure of the classical combs, in Appendix~\ref{measure_derivation} we derive this quantity; here we simply present the main results. 

The probability for Bob winning the game is given by:

\begin{align}
\label{eqn::success_distinction}
\mathbbm{P}_B(C)=\frac{1}{2}\left(1+ M(C)\right),
\end{align}
with $M(C)$ being one half of the solution of
\begin{alignat}{2}
&\text{minimize:}\ &&\max_{i} \sum_{\vec{x}} \left|\tr[(C^{\textup{Cl.}}-C) T_i(\vec{x})] \right| \label{eq:minmax} \\
&\text{subject to:}\  &&	C^{\textup{Cl.}}=\sum_{y_K,\ldots,y_1} \mathbbm{P}_K(\vec{y}) P_{y_K}\otimes\cdots \otimes P_{y_1}, \nonumber \\
& &&\mathbbm{P}_K(\vec{y})\; \mathrm{joint\, prob.\, distribution}. \nonumber
\end{alignat}
This can be transformed into the following linear program (and hence can be solved efficiently numerically; the error can be estimated and one can compute the optimal $C^{\textup{Cl.}}$ and $T_i(\vec{x})$ \cite{boyd2004convex}):
\begin{alignat}{2}
&\text{minimize:}\ && a \label{eqn::linearProgMeas} \\
&\text{subject to:}\ &&\sum_j b_{ij}-a\le0, \nonumber \\
& &&\sum_k p_k \alpha_{ijk}-\beta_{ij} -b_{ij}\le0,  \nonumber \\
& &&-\sum_k p_k \alpha_{ijk}+\beta_{ij} -b_{ij}\le0, \nonumber \\
& &&\sum_k p_k -1=0, \nonumber \\
& &&	p_k\ge0, a\ge0, b_{ij}\ge0,\nonumber
\end{alignat}
where we have defined $\alpha_{ijk}:=\tr\left[ (P_{y_K(k)}\otimes\cdots \otimes P_{y_1(k)}) T_i(\vec{x}_j) \right]$, $\beta_{ij}:=\tr \left[ C T_i\left(\vec{x}_j\right)\right]$ and $p_k:=\mathbbm{P}_K(\vec{y}(k))$. 
For completeness we also give the dual program, which by definition turns a minimization into a maximization. 
The dual problem is useful to give bounds on the found solution, to solve the problem, and potentially to find different interpretations of the quantity in question. 
The dual of the program above can be formulated as:
\begin{alignat*}{2}
&	\text{maximize:}\ &&  Z\\
& 	\text{subject to:}\ && Z \leq \sum_{ij} \left(\alpha_{ijk}-\beta_{ij}\right)\left(2 Y_{ij}-X_{i}\right)\quad\forall \ k, \\
 & && \sum_i X_i=1, \\
& &&X_i, Y_{ij}, X_i-Y_{ij} \ge0,  \\
& && Z \in \mathbbm{R}.
\end{alignat*}
It follows directly from the interpretation as the solution of the game defined above that the quantity $M(C)$ is faithful, i.e., its value is zero if the statistics is classical, and that it measures how difficult it is to simulate the given comb by a classical stochastic process. 
As such, it provides us with a properly motivated quantifier of the degree of non-classicality of quantum processes, which describes how well the obtained statistics can be simulated by a classical process.

The full evaluation of $M(C)$ would, in principle, require testing over every sequence of projective measurements
(to compute the maximization in Eq.~\eqref{eq:minmax}) and the comparison with every classical
multi-time probability distribution (to compute the minimization in Eq.~\eqref{eq:minmax}). Practically, it is then useful to consider bounds to this quantifier of non-classicality, which can be accessed via a limited number of measurements. In particular, lower bounds can be obtained by using a subset of measurement sequences 
$T_i(\vec{x})$ (in a similar way as to how one can use entanglement witnesses to construct bounds on meaningful entanglement measures~\cite{horodecki_1996,Audenaert_2006,Eisert_2007,Guehne_2007}). If such a lower bound is non-zero, this is already sufficient to conclude that the comb is non-classical.
On the other hand, upper bounds can be attained by restricting our consideration to some classical combs. As a relevant example, for any given comb $C$ one can focus on a single classical comb $\overline{C}^{\textup{Cl.}}$, 
realized by applying a dephasing map
before and after each measurement. 
This yields the statistics resulting from the marginals of the joint statistics one would obtain by measuring at every time.
Note that, while this specific choice of a classical comb only provides us with an upper bound on our measure defined above, it is nonetheless faithful.
In the simplest case where only two times are involved, $K=2$, one can easily see that by
replacing $C^{\textup{Cl.}}$ with $\overline{C}^{\textup{Cl.}}$ in Eq.~\eqref{eq:minmax}, we derive the following upper bound
\begin{equation}\label{eq:boundMeasure}
    M(C) \leq \sum_{x_2} \left|\Pprob(x_2) - \sum_{x_1} \Pprob(x_2,x_1) \right|.
\end{equation}
Such a `natural' quantifier of non-classicality has already been used 
to investigate coherence properties in transport phenomena~\cite{li2012}
and, more recently, to control the departure from any classical random walk via the manipulation of quantum
coherence in a time-multiplexed quantum walk experiment~\cite{smirne_experimental_2019}. Let us note at this point that the experimental data that was used in Ref.~\cite{smirne_experimental_2019} to evaluate the right hand side of Eq.~\eqref{eq:boundMeasure} allows one to calculate $M(C)$ too. 
Hence, $M(C)$ can be evaluated without further acquisition of experimental data, which demonstrates the applicability of our measure to current experiments. In addition, our measure---or lower bounds thereof---can be employed to investigate more complex experiments with $K>2$.

\section{Dynamical properties of \texorpdfstring{$K-\text{CLASSICAL}$}{} processes}
\label{sec::DiscClass}
Theorem~\ref{thm::KclassComb} and Lemma~\ref{lem::coherence_non_Markov} provide a full characterization of processes that yield classical statistics. Together, they allow for the derivation of classically testable quantifiers of non-classicality. For further clarification, and in order to connect non-classical processes to the respective \emph{underlying} evolution, we now discuss some concrete cases of underlying non-Markovian dynamics that lead to classical statistics. Moreover, we will connect the classicality of temporal processes to vanishing quantum discord in the joint state of the system and the environment.

\subsection{Discord and Classicality}
Recall that in the Markovian case, the classicality of a process can be decided solely in terms of propagators between pairs of times that are defined on the system of interest alone and it is linked to the ability of those maps to create and detect coherences. In particular, the set of dynamics that does not create coherences on the level of the system is contained in the set of maps that lead to classical statistics~\cite{smirne_coherence_2019}. As we have seen above, this fails to hold in the non-Markovian case, where, even if the state of the system is diagonal in the computational basis at all times, i.e., no coherence on the system level is ever generated, the statistics might not satisfy the Kolmogorov conditions.  

As soon as memory effects play a non-negligible role, it is \textit{both} the coherences of the system state \textit{and} the correlations between the system and its environment that can lead to non-classical behavior. It is thus desirable to derive a more explicit relation between coherence, correlations and classicality. 

To do so, first recall that while the completely dephasing map leaves the system unchanged if the state of the system is classical at all times, it does not necessarily leave the overall system-environment state, which, at every time $t_j$ contains all relevant memory, invariant. Specifically, in this case we have $\Delta_j[\rho^s_{t_j}] = \Ical[\rho^s_{t_j}] \ \forall \ t_j$ but not necessarily $\Delta_j \otimes \Ical_j^e [\eta^{se}_{t_j}] = \Ical^{se}[\eta^{se}_{t_j}] = \forall \ t_j$. While the latter is not necessary for the satisfaction of the Kolmogorov conditions, it is sufficient:

\begin{lemma}
\label{lem::Discord}
Let $\{p_{t_i}^m\}$ be sets of probabilities that sum to unity, $\{\Pi_{j}^m \}$ orthogonal projectors (not necessarily rank-$1$) on the system that are diagonal in the computational basis, and $\{\xi_j^m\}$ states on the environment. If at all times $t_j\in \Tcal$, with $|\Tcal| = K$, the system-environment state is of the form 
\begin{gather}
\label{eqn::basis_discord}
 \eta^{se}_{t_j} = \sum_m p_{t_j}^m\, \Pi_{j}^m \otimes \xi_j^m\, ,
\end{gather}
then the underlying process is $K$-classical, i.e., it satisfies the Kolmogorov conditions of Eq.~\eqref{eqn::Kolmo_cond}. 
\end{lemma}
Note that we assume the computational basis to be the same at every time, so that the additional subscript of $\Pi_j^m$ is somewhat superfluous and merely added to clearly signify the respective time at which the state is defined. In principle, one could define classicality with respect to projective measurements in different bases at each time $t_j$, in which case the additional subscript of $\Pi_j^m$ would denote projectors in different bases, and the above lemma would still hold. Analogously, all other results of this paper can straighforwardly be adapted to such more general probing schemes, but for simplicity, we understand classicality with respect to a \textit{fixed} basis that does not change in time (the only exception being Sec.~\ref{sec:genuinelyquantumprocess}, where we will extend the setting to allow for arbitrary measurement schemes in order to examine the nature of genuinely quantum processes.). Naturally, the environment states $\xi_j^m$ in Eq.~\eqref{eqn::basis_discord} can be diagonal in \textit{arbitrary} bases, as it is only invasiveness with respect to measurements on the system that we are concerned with. 

Before we prove Lemma~\ref{lem::Discord}, it is insightful to discuss the relation between the concept of classical temporal processes and the classical spatial system-environment correlations it introduces. Firstly, recall the full system-environment state at each time encapsulate all memory effects. Concretely, in contrast to the state of the system alone, they contain all information that is relevant to predict the future statistics. In particular, for states of the form given in Eq.~\eqref{eqn::basis_discord}, at each time $t_j$, this memory is stored in the probabilities $\{p_{t_j}^m\}$ and the enviromnment states $\{\xi_j^m\}$. States of said form have vanishing quantum discord~\cite{ollivier_quantum_2001, Zurek_2000, henderson_classical_2001, datta_condition_2010, modi_classical-quantum_2012}, i.e., they do not display any genuinely \textit{quantum} correlations between the system and the environment. 
For a general zero-discord state, the set $\{\Pi_{j}^m\}$ in Eq.~\eqref{eqn::basis_discord} could be any set of mutually orthogonal projectors, and the correlations between the system and the environment are considered to be classical, since there exists a measurement on the system with perfectly distinguishable outcomes that overall leaves the total state undisturbed~\cite{ollivier_quantum_2001,modi_classical-quantum_2012} (see also the proof below). 

As we only consider measurements on the system in a fixed basis in our setting, here, vanishing discord at all times does not yet force the resulting statistics to be classical; rather, the discord must vanish in the correct basis, i.e., the one in which the experimenter's measurements act.  While discord is often considered as a basis independent quantity---obtained by a minimization procedure over all possible measurement scenarios~\cite{modi_classical-quantum_2012}---here, and throughout the remainder of this article, we will always consider its basis dependent formulation~\cite{Zurek_2000, ollivier_quantum_2001,henderson_classical_2001, modi_classical-quantum_2012, yadin_quantum_2016, Egloff2018} and call states of the form in Eq.~\eqref{eqn::basis_discord} discord-zero \textit{with respect to the classical basis}. That is, whenever we consider a state to be of zero discord, we will always implicitly mean that it can be represented as per Eq.~\eqref{eqn::basis_discord} with the projectors being diagonal in the classical basis of the measurements. Importantly, this basis dependence mirrors the basis dependence of coherence, which is also always defined with respect to a fixed classical basis.
\begin{proof}
For states of the form in Eq.~\eqref{eqn::basis_discord}, the completely dephasing map $\Delta$ on the system has the same effect as the `do-nothing' identity channel $\Ical$, i.e.,
\begin{align}
    &\Delta_j\otimes \Ical_j^e\left[\sum_m p_{t_j}^m \Pi_{j}^m \otimes \eta_j^m\right] \\
    \notag&\phantom{asdfaaaaaaaa}= \Ical^s_j \otimes \Ical^e_j\left[\sum_m p_{t_j}^m \Pi_{j}^m \otimes \eta_j^m\right]\, .
\end{align}
Consequently, if the system-environment state is of this form at all times, the resulting statistics satisfy the Kolmogorov conditions.
\end{proof}

It is insightful to re-examine Example~\ref{ex::class_state} in light of Lemma~\ref{lem::Discord}; there, we provided an example of a process for which the state of the system never displayed coherence, but nonetheless led to non-classical statistics. Consequently, the system-environment state must have non-zero (basis dependent) discord over the course of the dynamics:

\begin{manualtheorem}{\hspace{-.18cm}\ref*{ex::class_state}$^\prime$}
\label{ex::classical_prime}
As we discuss in Appendix~\ref{app::Absence}, in Example~\ref{ex::class_state}, the system-environment state before the first measurement $(t<t_1)$ is given by
		\begin{align}
		    \rho_{se}(t)=\frac{1}{4}\sum_{i, j \in \{-,+\}} &\ketbra{i}{j} \otimes 
		    \left( 
		        i\cdot j         \ketbra{\varphi^-(t)}{\varphi^-(t)}
		    \right.
		    \nonumber\\&
		        +i (2 \alpha -1) \ketbra{\varphi^-(t)}{\varphi^+(t)}
		        \nonumber\\&
		    +j(2 \alpha -1)\ketbra{\varphi^+(t)}{\varphi^-(t)}
		    \nonumber\\&
    		\left.
    		   + \ketbra{\varphi^+(t)}{\varphi^+(t)}
		    \right)
		\end{align}
		where both
		\begin{align}
			\ket{\varphi^+(t)}= \int_{-\infty}^{\infty} d p f(p) e^{ipt} \ket{p}
		\end{align}
		and
		\begin{align}
			\ket{\varphi^-(t)}= \int_{-\infty}^{\infty} d p f(p) e^{-ipt} \ket{p}
		\end{align}
		are valid quantum states. This state has zero discord with respect to the eigenbasis of $\hat \sigma_x$ iff 
		\begin{align}\label{eq:discordZeroCond1}
			\ketbra{\varphi^+(t)}{\varphi^+(t)}-\ketbra{\varphi^-(t)}{\varphi^-(t)} =0
		\end{align}
		and either $\alpha=1/2$ or
		\begin{align}
			\ketbra{\varphi^+(t)}{\varphi^-(t)} -\ketbra{\varphi^-(t)}{\varphi^+(t)} =0.
		\end{align}
		In the case of the Lorentzian distribution, it follows from 
		\begin{align}
			\braket{\varphi^-(t)|\varphi^+(t)}=k(t)= e^{-2 \Gamma |t|}
		\end{align}
		that Eq.~(\ref{eq:discordZeroCond1}) cannot be satisfied for $t>0$, i.e., basis dependent discord is created during the evolution (and subsequently destroyed by the measurement at $t_1$). Since the state of the system itself is not altered by the measurement, but the probabilities to obtain $\pm$ at a later time are (as has been discussed in Ref.~\cite{smirne_coherence_2019}), the discord necessarily must be converted into populations by the following portion of evolution. Below, we will examine this connection between the creation and detection of basis dependent discord and non-classicality in a rigorous manner.

\end{manualtheorem}

If a state is of zero discord, it displays \textit{neither} coherences on the level of the system \textit{nor} non-classical correlations between the system and the environment, which is, to reiterate, sufficient for the classicality of the resulting process, but \textit{not} necessary. In this sense, Lemma~\ref{lem::Discord} is a direct extension of the analogous statement in the Markovian case; there, the absence of coherence in the system state at all times is also sufficient but not necessary for the process to be classical. Put differently, if all of the individual maps making up a Markovian dynamics are maximally incoherent operations \textbf{(MIO)}~\cite{aberg_catalytic_2014,streltsov2017colloquium}, i.e., they map all incoherent states onto incoherent states, then the resulting dynamics satisfies Kolmogorov conditions. However, MIO operations are a strict subset of NCGD maps~\cite{smirne_coherence_2019}.

While somewhat intuitive, the above lemma sheds light on the properties that a general non-Markovian dynamics has to satisfy in order to appear classical. For system-environment states that are discord-zero  in the computational basis (with respect to the system), a measurement on the system in the computational basis is non-invasive, i.e., it leaves the \textit{full} state unchanged (and not just the system state, as it would be the case if the system state is incoherent at all times). For comprehensiveness, in Appendix~\ref{app::discord_zero} we provide a characterization of non-discord creating processes in terms of their dynamical building blocks.

In general, the absence of discord at all times is not necessary for a process to appear classical. However, what is necessary is that at no time can there be coherences \textit{or} non-classical system-environment correlations that can be detected by means of measurements in the computational basis at a later time. This mirrors the requirement for classical processes in the Markovian case, where the individual propagators have to be NCGD, i.e., the propagators must be such that they cannot create coherences whose existence can be picked up at a later time by means of measurements in the classical basis; yet, it is still possible that the individual maps create coherences~\cite{smirne_coherence_2019}. NCGD maps are the fundamental building blocks that constitute classical Markovian combs. In what follows, utilizing the connection of classicality and discord discussed above, we will provide a characterization of the building blocks that make up classical \textit{non-Markovian} processes.

\subsection{Non-Discord-Generating-and-Detecting Dynamics and Classical Processes} 

In the Markovian case, classicality of a process can be decided on the level of CPTP maps, since in the absence of memory all higher order probability distributions can be obtained from the system state $\rho_{t_1}$ and the two-time propagators $\{\Lambda_{t_j,t_{j-1}}\}$. It suggests itself to employ this intuition in the non-Markovian case, as every non-Markovian process corresponds to a Markovian one if enough additional degrees of freedom are taken into account.

In detail, as we discussed, every non-Markovian process can be dilated to a concatenation of a (potentially correlated) system-environment state and unitary total dynamics~\cite{chiribella_theoretical_2009, pollock_non-markovian_2018}, interspersed by the operations of the experimenter on the system alone that are performed at times $\{t_j\}$ (see Fig.~\ref{fig::general_comb} for reference). If the experimenter had access to all the degrees of freedom necessary for the dilation, then the underlying process would appear Markovian, and the results of Ref.~\cite{smirne_coherence_2019} could be applied on the system-environment level for the characterization of a classical process. Here, using the Markovian case as a guideline, we aim for a similar characterization of classical processes when only the system degrees of freedom can be accessed.

To compactify notation and simplify later discussions, we can equivalently consider a general open process as a concatenation of CPTP maps that act on both the system and the environment, interspersed by the operations on the system alone. This way of describing general open system dynamics is simply a notational compression of the general case with global unitaries that allows for an easier connection to the Markovian case, but does not lead to a different set of possible combs. In what follows, we will denote these CPTP maps by $\Gamma_{t_j,t_{j-1}}$ to clearly distinguish them from the memoryless scenario (where the respective maps $\Lambda_{t_j,t_{j-1}}$ act only on the system),
so that Eq.~\eqref{eqn::Gen_Prob} generalizes to
\begin{align}
\label{eq:extra5}
    &\Pprob_n(x_n,\dots,x_1) \\ 
    \notag &= \tr\left[(\Pcal_{x_n}\otimes \Ical^e)\circ \Gamma_{t_n,t_{n-1}} \circ \cdots \circ (\Pcal_{x_1}\otimes \Ical^e) [\eta^{se}_{t_1}]\right]\,.
\end{align}
 Moreover, for the sake of generality and to ease the comparison with the Markovian case, we allow for the state before the first measurement to be evolved from some other state at an initial reference time $t_0 \leq t_1$, i.e., 
\begin{equation}\label{eq:initstdis}
    \eta^{se}_{t_1} = \Gamma_{t_1,t_{0}}\eta^{se}_{t_0};
\end{equation}
of course, if the first measurement occurs at the initial time, then $t_1=t_0$.

On this dilated level, the dynamics is Markovian---there are no additional external `wires' that can carry memory forward---and all higher order joint probability distributions could be built up when the individual CPTP maps $\{\Gamma_{t_j,t_{j-1}}\}$ (and the initial system-environment state) are known. With this, we can define \textit{non-discord-generating-and-detecting} (\textbf{NDGD}) dynamics: 
\begin{definition}[NDGD dynamics]
\label{def::NDGD}
A global system-environment dynamics with CPTP maps $\{\Gamma_{t_j,t_{j-1}}\}_{j=1}$ is called non-discord-generating-and-detecting (NDGD) if it satisfies \begin{align}
\label{eqn::NDCG}
&\Delta_{j+1} \circ \Gamma_{t_{j+1},t_j} \circ \Delta_j \circ \Gamma_{t_j,t_{j-1}} \circ \Delta_{j-1}
\\ \notag&= \Delta_{j+1} \circ \Gamma_{t_{j+1},t_j} \circ \Ical_j \circ \Gamma_{t_j,t_{j-1}} \circ \Delta_{j-1}
\end{align}
for all $\{t_{j-1},t_j,t_{j+1}\}$, where the maps $\Gamma_{t_k,t_{k-1}}$ act on the system and the environment, while $\Delta_k$ act on the system alone.

\end{definition}
We provide a graphical representation of this definition in Fig.~\ref{fig::NDGD}.
\begin{figure}
 \centering
 \includegraphics[width=1\linewidth] {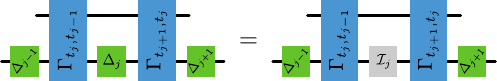}
 \caption{{\bf{NDGD system-environment dynamics.}} From the perspective of a classical observer performing projective measurements in a fixed basis, the identity map at any time $t_j$ cannot be distinguished from the completely dephasing map. Any discord (with respect to the classical basis) that is present in the system-environment state, and/or created by the system-environment CPTP maps, cannot be detected by such a classical observer.}
 \label{fig::NDGD}
\end{figure}

Formally, Definition~\ref{def::NDGD} is equivalent to the definition of NCGD dynamics, with the difference that the involved intermediary maps between times are now the system-environment maps, instead of the maps  $\{\Lambda_{t_j,t_{j-1}}\}$ acting on the system alone in the Markovian case.

Analogously to the case of NCGD, an NDGD dynamics cannot create discord (with respect to the classical basis) that can be detected at the next time (and, as such, at \textit{any} later time) by means of classical measurements. Or, equivalently, an experimenter who can only perform measurements in the classical basis cannot distinguish between a completely dephasing map and an identity map implemented at any time in $\Tcal$. As such, it provides the natural extension of NCGD to the non-Markovian case. We then have the following theorem: 
\begin{theorem}[NDGD dynamics and classicality]
\label{thm::NDCG}
 Consider a general, possibly non-Markovian, process on $\Tcal$, with $|\Tcal| = K$, obtained
 from a system-environment dynamics as in Eqs.~\eqref{eq:extra5} and~\eqref{eq:initstdis}; then the process is $K$-classical if the initial system-environment state $\eta^{se}_{t_0}$ and the set $\{\Gamma_{t_j,t_{j-1}}\}$ of maps that corresponds to it are zero discord and NDGD, respectively. 
\end{theorem}
The proof of this theorem is provided in Appendix~\ref{app::proof_NDGD}. It relies on the fact that measurements in the classical basis commute with the completely dephasing map and proceeds along the same lines as the analogous proof for NCGD dynamics in the Markovian setting provided in Ref.~\cite{smirne_coherence_2019}. Importantly, though, it is \textit{not} a necessity for classical statistics that the corresponding maps are NDGD, as we will discuss below. 

In order to further elucidate the relation of discord and classicality for general quantum stochastic processes, it is insightful to discuss the proximity of Theorem~\ref{thm::NDCG} to the corresponding results in Ref.~\cite{smirne_coherence_2019} for the Markovian case. Theorem~\ref{thm::NDCG} establishes the importance of the role of quantum discord for the classicality of non-Markovian processes. In the memoryless case, it is coherence---or the impossibility of detection thereof---that makes a process classical. Here, this role is played by discord, with the only difference being that instead of describing the process in terms of maps that are solely defined on the system of interest, we are forced to dilate the process to the system-environment space, where it is rendered Markovian. Consequently, the classicality of a process cannot be decided based on the master equation or dynamical maps that describe the evolution of the system alone (as has already been pointed out in Ref.~\cite{smirne_coherence_2019}). However, given, e.g., a Hamiltonian that generates the corresponding system-environment dynamics, whether or not the resulting process can be simulated classically can be decided by checking the validity of Eq.~\eqref{eqn::NDCG}. 

It would be desirable if NDGD dynamics were a sufficient \textit{and} necessary criterion for the classicality of non-Markovian processes; however, this is not the case. We provide an example of dynamics that is not NDGD, but nevertheless leads to classical dynamics, in Appendix~\ref{app::Class_non_NDGD}.
NDGD as defined in Eq.~\eqref{eqn::NDCG} is a statement about the entire system-environment dynamics, and holds for any possible initial state on the environment. However, by means of projective measurements on the system alone, one only has access to the system part, and the system-environment dynamics cannot be fully probed. Consequently, the criterion of Eq.~\eqref{eqn::NDCG} will, in general, be too strong for a given experimental scenario. 
Crucially, though, Theorem~\ref{thm::NDCG} allows us to understand the role of the discord generated by the system-environment interaction and subsequently detected
via projective measurements on the system in establishing non-classical statistics.

Nonetheless, even though it is not necessary for the underlying dynamics to be NDGD in order for a non-Markovian process to display classical statistics, for any $K$-classical process, there always \textit{exists} a dilation that is NDGD. That is, there exists a set $\{\widetilde \Gamma_{t_j,t_{j-1}}\}$ of system-environment CPTP maps that are NDGD and a zero-discord system-environment state $\widetilde{\eta}_{t_0}^{se}$ that yield the correct classical family of joint probability distributions when probed in the classical basis. Specifically, we have the following theorem:
\begin{theorem}
\label{thm::NDGD_dilation}
Let $\{\mathbb{P}_n(x_n,\ldots,x_1)\}_{n\leq K}$ define a process on $\mathcal{T}$, with $|\Tcal|=K$, coming from an underlying
evolution, fixed by the system-environment maps $\{\Gamma_{t_j,t_{j-1}}\}$ and the state $\eta_{t_0}^{se}$, according to Eqs.~\eqref{eq:extra5} and ~\eqref{eq:initstdis}. 
The resulting statistics $\{\mathbb{P}_n(x_n,\ldots,x_1)\}_{n\leq K}$ is $K$-classical iff there exists a NDGD evolution given by system-environment maps $\{\widetilde{\Gamma}_{t_j,t_{j-1}}\}$ defined on times in $\Tcal$
and a zero-discord state 
$\widetilde{\eta}_{t_0}^{se}$ that yield $\mathbb{P}_n(x_n,\ldots,x_1)$ when probed in the classical basis.
\end{theorem}
Before we prove this statement, it is important to contrast it with Theorem~\ref{thm::NCGD_classical}, the analogous result for Markovian processes. 
There, NCGD propagators of the system dynamics guarantee
that the process associated with sequential projective measurements is classical, and classical Markovian processes can be reproduced by a set of NCGD maps (which do not necessarily identify with the actual dynamical propagators). Analogously, here, the NDGD property of the actual system-environment evolution ensures the classicality of the process; while the converse holds for \emph{particular} dilations, there can be non-NDGD dilations that nonetheless yield classical statistics. 

In both cases the projective measurements in a fixed basis only provide a limited amount of information about the overall evolution underlying the probed statistics. While in the Markovian case the statistics can be traced back to dynamical maps acting on the open system alone, in the more general non-Markovian case it is the whole system-environment evolution that enters into play. As a consequence, only the former case allows one to establish a one-to-one correspondence between classicality and the properties of the \textit{actual} evolution
by enforcing a proper condition on the dynamics,
as discussed at the end of Sec.~\ref{sec:otoma}.
\begin{proof}
As we have already seen in the discussion of  Theorem~\ref{thm::NDCG}, the joint probability distributions obtained from an NDGD dynamics are always classical. We thus only need to prove the opposite direction. Let the underlying system-environment dynamics of the process between times be given by the maps $\{\Gamma_{t_j,t_{j-1}}\}$. As the process is classical, the set of maps $\{\widetilde{\Gamma}_{t_j,t_{j-1}} = \Delta_j \circ \Gamma_{t_j,t_{j-1}}\circ \Delta_{j-1}\}$ together with a state $\widetilde{\eta}_{t_1}^{se} = \Delta_1[\eta_{t_1}^{se}]$, where, again, $\Delta_k$ only acts on the system degrees of freedom, yields the same joint probability distributions when probed in the classical basis (see Fig.~\ref{fig::Transf_NDGD} for reference).
The process given by this set  $\{\widetilde{\Gamma}_{t_j,t_{j-1}}\}$ is NDGD by construction and $\widetilde{\eta}_{t_1}^{se}$ has vanishing discord, which means that for every $K$-classical process there is an NDGD dilation that reproduces it correctly, where we identify the initial time as the time of the first measurement, $t_0=t_1$. 
\end{proof}
\begin{figure}
    \centering
    \includegraphics[width=0.95\linewidth]{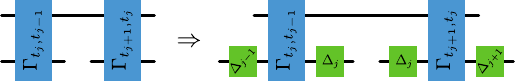}
    \caption{{\bf{Transformation to NDGD dilation.}} Any dilation of an open dynamics can be mapped onto an NDGD one by inserting completely dephasing maps on the level of the system. If the process is classical, then the transformed dilation yields the same statistics as the original one when probed in the classical basis.}
    \label{fig::Transf_NDGD}
\end{figure}

Theorems~\ref{thm::NDCG} and~\ref{thm::NDGD_dilation} complete our results for the non-Markovian setting and provide an intuitive connection between non-classical spatial correlations (i.e., discord) and classical processes.

\section{Genuinely Quantum Processes}
\label{sec:genuinelyquantumprocess}

As we have alluded to throughout this article, the classicality of a process depends on the measurement scheme that is employed to probe it; a process that appears classical in one basis---and is thus NDGD with respect to said basis---might display non-classical correlations when probed differently. This raises the question if non-classicality is merely a matter of perspective. In principle, for any process, there could exist a probing scheme that yields classical statistics. More concretely, for an experimenter that can perform arbitrary measurements, it might always be possible to `hide' the quantum nature of a process by choosing their respective measurements at the times $\{t_j\}$ such that the resulting statistics are classical. 

Naturally, such schemes with (potentially non-projective) measurements go beyond the discussion of classicality that we have conducted so far. As we will not limit the employed instruments of such schemes to be the same at every time, we will call them \textit{unrestricted} in what follows. However, we still assume that the instrument at each time is fixed in advance and is  independent of previous measurements---if the choice of instruments could depend on previous outcomes, then the employed probing scheme would be temporally correlated and marginalization at a given time would not be well-defined.

In this case, our previous results allow us to show that there exist \emph{genuinely quantum} processes, i.e., processes that display non-classical statistics with respect to \emph{every} unrestricted measurement scheme (in the sense described above) which reveals something about the probed process.

To reiterate, up to this point, our discussion of Markovianity focused on situations, where an experimenter measures in the computational basis only, thus employing the same instrument $\Jcal = \{\Pcal_{x_j}\}$ at each time, where all of the (projective) CP maps $\Pcal_{x_j}$ comprising the instrument added up to the completely dephasing map $\Delta_j$. More generally, an experimenter could use instruments $\Jcal_1 = \{\Mcal_{x_1}\}, \Jcal_2 = \{\Mcal_{x_2}\}, \dots$, each adding up to the CPTP maps $\Mcal_1, \Mcal_2, \dots$, respectively, to sequentially probe the system of interest. With this, for a process defined on times $\Tcal$, they could collect the joint probability for all subsets $\Tcal'\subseteq \Tcal$ and check if Kolmogorov consistency holds. For example, in the simplest case of two times, with $\Tcal = \{t_1,t_2\}$ and a given comb $\Ccal_2$ on $\Tcal$, an experimenter would consider the process classical, if $\Pprob(x_2|\Jcal_2) = \sum_{x_1} \Pprob(x_2,x_1|\Jcal_2,\Jcal_1)$ holds for all $x_2$, i.e., if 
\begin{gather}
    \Ccal_2[\Mcal_{x_2},\Ical_1] = \Ccal_2[\Mcal_{x_2},\Mcal_1] \quad \forall \ \Mcal_{x_2} \in \Jcal_2\, .
\end{gather}
Note that, due to causality, the second conditions, i.e., $\Ccal_2[\Ical_2,\Mcal_{x_1}] = \Ccal_2[\Mcal_{2},\Mcal_{x_1}] \ \forall \ \Mcal_{x_1} \in \Jcal_1$ holds automatically, independent of whether the process is classical or not.

In principle, there could \textit{always} exist a set of instruments $\{\Jcal_K,\dots,\Jcal_1\}$ for a given process $\Ccal_K$ on $\Tcal$, such that the resulting statistics appear classical. Naturally, for this question to make sense, the respective instruments actually have to extract information from the process at hand. In principle, an instrument could consist of a random number generator and a set of CPTP maps that the experimenter implements depending on the respective output of the random number generator. Considering these outputs as outcomes of the instruments, the experimenter could then collect statistics that are independent of the process at hand (they only depend on the statistics of the random number generators), and satisfy Kolmogorov consistency conditions (if the respective random number generators at different times are independent of each other). However, this apparent classicality would not be a statement about the properties of the underlying process, and we thus exclude such pathological instruments. We can do so by demanding that at any time $t_j$, none of the elements $\Mcal_{x_j}$ of the instrument $\Jcal_j$ is proportional to a CPTP map. Under this reasonable assumption, we now show that there are processes that are genuinely quantum, i.e., they violate Kolmogorov conditions for arbitrary choices of instruments. 

To this end, in the first step, we argue that genuinely quantum processes only exist in the non-Markovian setting, while in the memoryless case there always exists a measurement scheme that yields classical statistics. This conclusion follows from the fact that all features of a Markovian process are governed by the dynamical maps acting on the space of the system alone. Suppose then that a Markovian process is deemed to be non-classical with respect to some basis of measurements: this means that the dynamical maps constituting the process generate and detect coherence with respect to said basis. However, at each point in time throughout the process, the system to be measured is diagonal in some basis (namely, its eigenbasis); thus, in principle, if the experimenter were able to choose an unrestricted measurement scheme that is always diagonal in the same basis as the system, no coherence with respect to this basis will ever be generated and detected, implying that the statistics measured will appear classical. Consequently, in our proposed framework, genuinely quantum processes can only exist in the presence of (quantum) memory. 

A similar argument as in the Markovian case holds for the special case of non-Markovian dynamics where the system-environment state at each time is of zero discord in a basis independent sense, i.e., when there exists a basis with respect to which the joint state at each time is discord-zero. Recall that if the system-environment dynamics is NDGD (with respect to a fixed basis), then the statistics observed are classical. Now, if at each time the system-environment state has zero discord, then an experimenter can (in principle) choose the measurement basis at each time to be the one with respect to which the performed measurement is non-invasive. For such a sequence of measurements, the experimenter would not be able to distinguish between having implemented the identity map or the dephasing map (with respect to the chosen basis) at any time, since the measurement is non-invasive on the joint system-environment state (due to the lack of discord). Thus, in such a scenario, there always exists some choice of bases in which such a process looks classical. It follows then that no non-Markovian process with zero basis independent discord between system and environment at every time is genuinely quantum.

However, the above logic fails in the general setting, which we now show by explicit example. To provide intuition, we first outline the logical implication that is a consequence of the classicality demand for a chosen (two-step) process (depicted in Fig.~\ref{fig:genuinelyquantumprocess} and described below). While for two times it is \textit{always} possible to find a measurement scheme such that the statistics appear classical (even in the non-Markovian case), when a non-Markovian process extends over multiple times, finding such a measurement scheme is not possible in general. We show this in detail in Appendix~\ref{app:genuinelyquantumprocess} by considering a variant of the process shown in Fig.~\ref{fig:genuinelyquantumprocess} that is extended over four times, proving the existence of genuinely quantum processes.

\begin{figure}[b]
    \centering
    \includegraphics[width=0.75\linewidth]{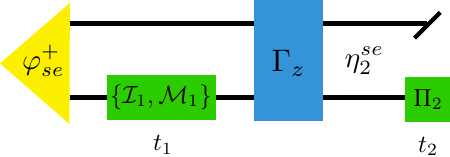}
    \caption{{\bf{First two times of a genuinely quantum process.}} The system-environment begin in a Bell state $\varphi^+_{se}$. Between times $t_1$ and $t_2$, the map $\Gamma_z$ is implemented, which biases the system in the $z$-basis if any CPTP map $\Mcal_1 \neq \mathcal{I}_1$ is performed (see Eq.~\eqref{eq:gammaz}). The label $\eta_2^{se}$ refers to the joint system-environment state immediately prior to $t_2$ (see Eqs.~\eqref{eq:sigma2seidentity1} and \eqref{eq:sigma2selambda1}). Classicality implies that the POVM $\Pi_2$ must be chosen such that it is unable to detect biases in the $z$-basis. Although this is always possible when only two times are considered, in general, classicality requires satisfaction of a growing number of constraints on the choices of later measurements, which can eventually lead to contradiction, implying the existence of genuinely quantum processes.}
    \label{fig:genuinelyquantumprocess}
\end{figure}

The explicit example process we consider begins with a two-qubit system-environment state in the Bell state $\varphi^+_{se} = \tfrac{1}{2} \sum_{ij} \ket{ii} \bra{jj}$. The experimenter can choose to measure the system (in whichever basis, or, more generally, employing any non-pathological instrument they like) at time $t_1$. Following this, the dynamics consists of a system-environment CPTP map $\Gamma_z : \mathcal{B}(\mathcal{H}^{s^\inp} \otimes \mathcal{H}^{e^\inp}) \to \mathcal{B}(\mathcal{H}^{s^\out} \otimes \mathcal{H}^{e^\out})$ whose action is to measure its joint inputs in the Bell basis, and output $\varphi^+$ if the measurement outcome indeed corresponds to $\varphi^+$, or else output a system-environment state whose system part is a pure state in the $z$-basis. The action of $\Gamma_z$ on a system-environment state $\eta^{se}$ is thus given by
\begin{align}\label{eq:gammaz}
    \Gamma_z[\eta^{se}] = &\tr(\eta^{se}\varphi_{se}^+) \, \varphi_{se}^+ \\
    \notag &+ \tr[(\mathbbm{1}_{se} - \varphi_{se}^+) \,\eta^{se}] \ketbra{0}{0}_{s} \otimes \tau_{e},
\end{align}
where $\tau_{e}$ is some quantum state on the environment. It is straightforward to check that such a map is indeed CPTP. Following this part of the dynamics, the experimenter has access to measure the system at time $t_2$.

For a genuinely quantum process, we demand that the statistics are non-classical with respect to \emph{any} possible measurement choices at times $t_1$ and $t_2$; if this is not the case, then there exists a POVM at $t_2$ that cannot distinguish between the experimenter having implemented the identity map $\mathcal{I}_1$ or an arbitrary CPTP map $\Mcal_1$ at time $t_1$, such that the statistics look classical with respect to said measurement scheme. By tracking the joint system-environment state for either choice of operation at $t_1$, we first show that such a POVM always exists. This implies that there is no genuinely quantum process defined on just two times, even in the non-Markovian setting. However, the POVM that does the trick is constrained by the demand of classicality, as we now detail. Extending the considered process to more times then imposes a number of constraints on the employed measurement devices which must be concurrently satisfied, such that finally there is no unrestricted measurement scheme that can yield classical statistics. 

Suppose that the experimenter implements $\mathcal{I}_1$ at time $t_1$; then, the system-environment state at $t_2$ is given by 
\begin{align}\label{eq:sigma2seidentity1}
    \eta_2^{se}(\mathcal{I}_1) := \Gamma_z [ (\mathcal{I}^s_1 \otimes \mathcal{I}^e) (\varphi^+_{se})] = \varphi^+_{se},
\end{align}
where the notation $\eta_2^{se}(\mathcal{I}_1)$ refers to the joint state immediately prior to $t_2$ given that the experimenter implemented the identity map at $t_1$. On the other hand, if the experimenter overall implements some CPTP map $\Mcal_1 \neq \mathcal{I}_1$ (corresponding to their instrument $\Jcal_1$ at $t_1$), then the initial Bell pair will be perturbed (as it is only locally invariant under the identity map) and therefore the system-environment state prior to $t_2$ is
\begin{align}\label{eq:sigma2selambda1}
    \eta_2^{se}(\Mcal_1) :=& \Gamma_z [(\Mcal_1^s \otimes \mathcal{I}^e) (\varphi^+_{se})] \notag \\
    =& p \varphi^+_{se} + (1-p) \ketbra{0}{0}_s \otimes \tau_e, 
\end{align}
where $p := \tr{\left[\varphi^+_{se} (\Mcal_1 \otimes \mathcal{I}^e) [\varphi^+_{se}] \right]} < 1$.
The statistics observed are gathered by making measurements on only the system, so we are now interested in the reduced system state at $t_2$ in either case: from Eq.~\eqref{eq:sigma2seidentity1}, we have the maximally mixed state $\eta_2^{s}(\mathcal{I}_1) = \tfrac{\mathbbm{1}}{2}$, whereas from Eq.~\eqref{eq:sigma2selambda1} we yield a state that is biased in the $z$-basis, $\eta_2^{s}(\Mcal_1) = \tfrac{p}{2} \mathbbm{1} + (1-p) \ket{0}\bra{0}$. As previously mentioned, classicality dictates that the POVM implemented at $t_2$ must not be able to distinguish between these two states, which leads to the fact that the chosen measurement must be blind to any bias in the $z$-basis. Mathematically, we demand
\begin{align}\label{eq:twotimeclassicality}
    \mathbbm{P}_2(x_2|\mathcal{I}_1) &\overset{!}{=} \mathbbm{P}_2(x_2|\Mcal_1), 
\end{align}
which can only be satisfied if the experimenter chooses a POVM $\Pi_2 = \{\Pi_2^{(x_2)} \}$ such that
\begin{align}\label{eq:twotimepovmdemand}
    \tr{\left[ \Pi_2^{(x_2)} \eta_2^{s}(\mathcal{I}_1)\right]} &= \tr{\left[ \Pi_2^{(x_2)} \eta_2^{s}(\Mcal_1)\right]} \quad \forall \ x_2.
\end{align}
    
A POVM that satisfies the above equation can be readily constructed: the elements $\{ \Pi_2^{(a)}, \mathbbm{1}-\Pi_2^{(a)} \}$ can always be described by $\Pi_2^{(a)} = r_2^{(0)} \mathbbm{1} + \vec{r}_2 \cdot \vec{\sigma}$, where $\vec{r}_2 = (r^{(x)}_2,r^{(y)}_2,r^{(z)}_2)$ and $\vec{\sigma} = (\sigma^{(x)}, \sigma^{(y)}, \sigma^{(z)})$ is the vector of Pauli matrices (note that we have changed notation and use the letter `$a$' to label the measurement outcome in order to avoid potential confusion with the $x$-basis direction). Demanding classicality, i.e., Eq.~\eqref{eq:twotimeclassicality}, then implies that $r^{(z)}_2 = 0$. In other words, any POVM that is not able to detect biases in the $z$-basis satisfies Eq.~\eqref{eq:twotimepovmdemand} and thus the statistics measured by such a POVM will appear classical. Importantly, here, and in what follows, we can restrict our analysis to the case of POVMs/instruments with only two elements, as any other POVM/instrument (except for the trivial case of single element ones) can always be coarse-grained to a two-element one. If such a coarse-grained instrument can detect non-classicality of statistics, then so too will the original one be able to, since it necessarily reveals more information about the process upon implementation.

However, although it might always be possible to find a basis/POVM such that the \textit{two-time} statistics for a non-Markovian process look classical, this is not the case in general. Intuitively, demanding that the experimenter cannot distinguish between implementing the identity map and an arbitrary CPTP map at different times leads to a number of constraints (e.g., above we have the constraint $r^{(z)}_2 = 0$) on the later measurement basis. In Appendix~\ref{app:genuinelyquantumprocess}, we consider a process defined across four times that is a logical extension of the two-time process considered here: in each of the first three times, depending on whether or not the system has previously been biased in either the $x$-, $y$- or $z$-basis, the process either performs an identity map (in the affirmative case) or else acts to bias the system in one of the bases. In the end, for an arbitrary CPTP map being implemented at each one of the first three times (with identity map being enacted at the others), the system state at the fourth time is biased in one of the three basis directions, and it is completely unbiased (i.e., maximally mixed) only if three consecutive identity maps are implemented. The only possible POVM at the final time that yields classical statistics must not be able to detect biases in any of the basis vector directions; the only POVM that achieves this is the one with elements proportional to the identity matrix, which corresponds to one of the measurements that we excluded because they reveal nothing about the process. Thus, the process is non-classical with respect to every possible non-pathological measurement scheme and is therefore genuinely quantum.

A relevant side-note seems in order here. Suppose that someone claims that a given process is genuinely quantum. To falsify such a statement it is enough to probe the process by whatever (non-trivial) devices one chooses; if the statistics one gets is classical, the statement is wrong. The processes that are not genuinely quantum can therefore be device-independently verified~\footnote{Note, that the term device independent does not imply that one needs no assumptions on the devices, but just that the assumptions are trivial from the setting. In the case of device independent quantum key distribution, for instance, the assumption that the devices have some kind of independence needs to be assumed. In our case, we need to assume that the measurements are not trivial, because otherwise there is nothing one can say from their statistics at all. }~\cite{Acin2007,Scarani2012,Liu2019,Kolodynski2020}. In turn, this makes the genuinely quantum processes quite peculiar, as it is impossible to hide their quantumness, and it might surprise that the set of these processes is non-empty; in fact we even conjecture that almost all many-time processes are genuinely quantum. 

\section{Conclusions and Outlook}\label{sec::ConclOut}
\subsection{Conclusions}
\label{sec::Concl}
In this paper, we have provided an operationally motivated definition of general classical stochastic processes and discussed its structural consequences and relation to quantum coherence in a system's evolution
as well as to the generation and activation of non-classical correlations between the system and the surrounding environment. While we phrased our results predominantly in the language of quantum mechanics, there is---\textit{a priori}---nothing particularly quantum mechanical about the notion of non-classicality we introduced. Rather, any process for which the potential invasiveness of performed measurements can be detected by means of said measurements is non-classical, independent of the underlying theory; as such an invasiveness is experimentally detectable, this is a fully operational notion. The question of whether or not a process is classical can thus be answered on experimentally accessible grounds and is \textit{a priori} independent of concepts that the experimenter might not be able to check for, like, e.g., coherences in the system of interest. 

Nonetheless, our definition allows for the derivation of a direct connection between the classicality of a process and coherences/non-classical correlations that might be present. While this connection can be formulated in terms of a necessary and sufficient condition for memoryless processes, there are additional subtleties to be considered in the non-Markovian case. In general, it is not sufficient for the state of the system to be diagonal in the classical basis at all times for the resulting multi-time statistics to be classical. Rather, it is the interplay of coherences, non-classical system-environment correlations, and the underlying dynamics that is of importance, as we have highlighted through a number of examples presented throughout. Using the comb framework---which can encapsulate this complex interplay---for the description of general quantum processes with memory, we have provided a characterization of quantum processes that yield classical statistics, and derived the structural properties of such processes. In principle, analogous structural properties could be derived for processes that display classical statistics when probed by means of different measurements, e.g., non-projective and/or non-orthogonal ones. However, while still enabling the derivation of structural properties, the clear connection between classicality and quantum discord would be lost as soon as sharp measurements in the computational basis are not the probing mechanism of choice anymore. In this paper, orthogonal projections were chosen as the kind of measurements that 
come closest to the ideal non-invasiveness displayed by classical measurements. More generally, our results could in principle also be extended to post-quantum theories. As the definition of classicality we provided is fully operational, the structure of classical processes in such theories could be derived in the same vein as we presented in this paper, with coherence and discord being replaced by the analogous properties of the respective theory.

Unsurprisingly, the set of classical processes turns out to be of measure zero within the set of all quantum processes. The full characterization we have provided equips the set of classical processes with an experimentally accessible measure of non-classicality that can be formulated as an linear program, thereby providing an operationally clear-cut quantification of the degree of non-classicality of a given quantum process and a general theoretical framework to define practically useful measures of non-classicality.
As an example, we showed how within our approach one can recover and motivate a quantifier of non-classicality
which is exploited in different contexts \cite{li2012} and has been used to analyze the 
properly quantum features of a given
experimental setup \cite{smirne_experimental_2019}.

Furthermore, we investigated the relation between the non-classicality of the statistics observed throughout a process and the quantumness of the prevalent spatial system-environment correlations in the underlying dynamics. While the absence of coherence in the state of the system of interest is no longer sufficient in the non-Markovian case to guarantee classicality, the absence of (basis dependent) discord is. This latter fact is somewhat intuitive, as the absence of discord at all times means that there are neither non-classical system-environment correlations nor coherences in the system that could influence the multi-time statistics deduced. Specifically, we have shown that the non-Markovian case to some extent mirrors the memoryless one: If the underlying dynamics is NDGD, i.e., any discord that is created at some point in time cannot be detected at a later time, then the process appears classical. While the converse of this statement does not hold, we have further shown that any classical process admits an NDGD dilation.

Finally, we demonstrated that, even if we extend our notion of classicality to the case of unrestricted measurement schemes, there exist processes that display non-classical statistics independent of how they are probed. This can happen only for non-Markovian processes, thus showing that genuine non-classicality can be seen as a further degree of complexity introduced by the presence of memory effects in the
multi-time statistics of quantum systems.

As our definition of classicality is tantamount to the assumptions of realism and non-invasiveness that underlie the derivation of Leggett-Garg inequalities, our results furnish experiments that test for the aforementioned properties with a clear interpretation: If the observed statistics satisfy a Leggett-Garg inequality, then the underlying process can be assumed to be NDGD. It does not have to be composed of fully classical resources, though. On the other hand, violation of a Leggett-Garg inequality implies that quantum discord must have been created (and later detected) over the course of the experiment.

\subsection{Outlook}
While we have provided a comprehensive picture of the interplay between the non-classical resources that are present in the underlying process and the non-classicality of the resulting non-Markovian multi-time statistics, the \textit{mechanisms} that lead to the emergence of classical behavior on macroscopic scales remain unclear. Na\"ively, the fact that classical processes only constitute a vanishing fraction of the set of all processes, renders it puzzling that classical processes can be observed at all. This apparent `puzzle' is reminiscent of the superposition principle which restricts the set of states that are diagonal in a fixed basis to be of measure zero in the set of all pure states, yet superpositions are generally not observed in the macroscopic domain, where one fixed basis seems to be singled out~\footnote{For experiments that detect somewhat macroscopic superpositions, see, e.g., Refs.~\cite{nairz_quantum_2003,oconnell_quantum_2010,eibenberger_matterwave_2013}}. While for the latter case, decoherence has been identified as the mechanism that fixes a preferred basis---and as such leads to the emergence of classicality in the spatial setting~\cite{zurek_decoherence_2003,schlosshauer_decoherence_2005}---an analogous investigation for temporal processes remains outstanding. Our results pave the way towards the analysis of the onset of classicality in general quantum processes when system and/or environment size increases.

Beyond this foundational perspective, the characterization of the set of classical processes, as well as the measure of non-classicality we have provided, lend themselves naturally to the development of a resource theory of non-classicality in which processes defined by Eqs.~\eqref{eqn::classLem} and~\eqref{eqn::chi_requirement} are free. Additionally, our approach yields a definite theoretical background which allows one to deal with different quantifiers of the degree of non-classicality, related to practical situations where different sets of operations are available to investigate the quantumness of physical processes.

On the structural side, we have fully characterized the set of classical processes and have shown that there exist processes that are genuinely quantum. However, the explicit partitioning of the set of quantum processes into classical, non-classical, and genuinely quantum processes remains opaque and requires further investigation. It suggests itself to assume that the set of genuinely quantum processes is of full measure: as the set of discordant states is of full measure in the set of all states~\cite{ferraro_almost_2010}, for a randomly chosen process, at any time $t_j$ there will generally not exist a measurement that leaves the respective system-environment state invariant, and the subsequent dynamics would have to be highly fine-tuned in order to disguise this invasiveness. More specifically, based on the arguments employed in the explicit construction  example of a genuinely quantum process we provided, where four measurement times were necessary to prove the genuine quantumness, we conjecture that almost all processes
associated with a $d$-dimensional system are genuinely quantum, if the system is probed
$d^2$ or more times. A rigorous proof of this statement is subject of future research. Moreover, since genuinely non-classical processes lead to non-classical statistics in a device-independent manner, their quantumness cannot be disguised. It then seems natural to explore if these processes  can be used for technological applications. 

Finally, the full characterization of general, non-Markovian quantum processes possessing an equivalent classical description, will likely be useful to better understand the different facets of memory effects in the classical and quantum realm. Although the operational framework of quantum combs does not \emph{a priori} concern any inherent timescales, as the choice of the discrete set of times is arbitrary, from a physical perspective one expects a connection between some relevant timescales of an underlying system-environment Hamiltonian generating a dynamics and the properties of the corresponding comb that arises upon specification of a set of times. Analogously, the timescales---and number of measurements---over which the non-classicality of a process can be deduced experimentally will be related to the pertinent timescales of the dynamics. However, determining the properties of an underlying system-environment Hamiltonian that leads to classicality and how the different timescales relate is an interesting, yet multi-layered and far from trivial, open problem.

The complexity arises due to the various temporal effects that play a significant role in determining the classicality (or absence thereof) of a given process and the relevant timescales over which it can be detected. For instance, we have already seen that the presence of multi-time memory effects is one such property; however, the connection between memory and classicality is a subtle one. One of the key differences between classical and quantum memory effects arises from the generically invasive nature of measurements in quantum mechanics, which leads to an inherent dependence of memory effects on the probing instruments employed~\cite{TarantoThesis}. The very notion of relevant memory timescales associated with the evolution of a quantum system therefore crucially depends on whether one wants to infer such timescales via sequential measurements over the course of the evolution or only at some final (possibly varying) time, as is done, e.g., in master equation approaches. In the latter case, the memory of the final statistics on the previous states of the system is dictated by the interplay of different timescales, related with the system of interest, its environment and their mutual interaction \cite{breuer_theory_2007}. Such a memory ultimately determines the complexity of the description of the system evolution, as provided, e.g., by memory kernels \cite{Nakajima1958,Zwanzig1960}, Green functions \cite{Zhang2012} or path integrals \cite{Chakraborty2018}.

In the case where the temporal correlations of the environment rapidly decay, the process can often be approximated as a Markovian one. When the process is indeed Markovian, i.e., described by a sequence of individual channels between times, as we have shown, it is the NCGD property of the evolution that is necessary and sufficient for classicality; however, this is not easy to relate to the relevant timescales. A property that would be sufficient for classicality, and more straightforwardly related to the inherent timescales of a Hamiltonian generating the evolution is \emph{forgetfulness} of any initial system state.

For instance, suppose one has a Markovian process generated by some Hamiltonian, which has a natural timescale of system forgetfulness, e.g., one that leads to an exponential decay of correlations between any preparations and final measurements. Then, if one probes such a process at sufficiently spaced time instants, one should expect to see classicality: the Markovianity property means that all relevant information can be determined solely on the system level, and forgetfulness ensures that any temporal correlations---in particular the ability to detect a distinction between a complete dephasing and an identity map---between adjacent times vanish. Strictly speaking, in the standard setting of testing for classicality, where a choice of measurement basis is fixed, one only requires forgetfulness with respect to projective measurements in said basis, rather than complete forgetfulness, for this argument to hold; however, besides being too strict a requirement, connecting such an instrument-specific forgetfulness to the relevant timescales is---like in the NCGD case---a difficult task.

In the presence of memory, the connection between classicality and the relevant timescales of the evolution is more involved yet. Here we have a subtle interplay between the question concerning the forgetfulness of the system of any initial non-classicality, as well as how much any non-classical effects can be transmitted through the environment via the memory mechanism. The fact that forgetfulness of the system alone here is insufficient to imply classicality is related to the crucial point that all multi-time effects must be captured in order to properly describe processes with memory. Thus, in the non-Markovian setting, the relevant timescales must typically be determined via sequential measurements over the course of the evolution. 

However, different interrogation procedures will lead to the exhibition of different multi-time memory effects. For instance, when the system is left unperturbed, the memory can be solely attributed to properties of the underlying Hamiltonian (e.g., those leading to the decay of environmental correlations), whereas when the system is measured, the effect of conditioning the environment state also plays a role. Similarly to the Markovian setting discussed above, the question of classicality of a non-Markovian process does not necessarily concern all such temporal correlations in the process (both those transmitted on the level of the system itself and the genuine memory effects due to the environment), but rather only those that can distinguish between the completely-dephasing instrument and the identity map applied to the system. These memory effects are, in turn, a special case of instrument-specific quantum Markov order, which has been recently introduced using the quantum comb formalism~\cite{taranto_2019L,taranto_2019A}. Connecting such memory effects of the process, and their subsequent impact on the classicality of observed statistics, with the timescales associated to the corresponding Hamiltonian that generates a given process poses a promising avenue for future research.

While we anticipate the above open questions to generate much theoretical interest, we also expect our results to find immediate application in a broad range of situations where it is relevant to assess whether experimental outcomes are not amenable to a classical description in order to certify some type of quantum advantage or benchmark some genuinely quantum behavior. The former include metrological schemes operating beyond the standard quantum limit~\cite{Demkowicz2015,Smirne2016,Haase2018,Yang2019} while the latter can refer to the simulation of many body quantum systems~\cite{Gemmer2001,Prior2010,Tamascelli2015,Farrelly2017,Tamascelli2019}. Also, the role that the emergence of classicality plays in system thermalization and homogenization can be investigated in a systematic and quantitatively tractable manner within our proposed approach.

\begin{acknowledgments}
Independently from the present work, Philipp Strasberg and Mar{\'i}a Garc{\'i}a D{\'i}az derived related results in Ref.~\cite{strasberg_classical_2019}, where, in particular, the case  of  non-rank-1  projectors  is  considered,  and  a definition of classicality for temporal quantum processes similar to our Theorem~\ref{thm::KclassComb} is put forward. 

We thank Konstantin Beyer, Benjamin Desef, Mar{\'i}a Garc{\'i}a D{\'i}az, Nana Liu, Kimmo Luoma, Kavan Modi, Felix A. Pollock, Philipp Strasberg, and Walter Strunz for valuable discussions. SM is grateful to the Monash Postgraduate Publication Award for financial support. This work was supported by the Swiss National Science Foundation SNSF (Grant No. P2SKP2\_184068), the Austrian Science Fund (FWF): ZK3 (Zukunftkolleg) and Y879-N27 (START project), the  European  Union's Horizon 2020 research and innovation programme under the Marie Sk{\l}odowska Curie grant agreement No 801110, the Austrian Federal Ministry of Education, Science and Research (BMBWF), and the ERC Synergy grant BioQ.
\end{acknowledgments}

\FloatBarrier
\bibliographystyle{prx_bib}
\bibliography{references.bib}

\appendix

\section{Connection to previous results}
\label{app::prev_results}

In this section we show that the result derived in the main text for the Markovian case (that is, Theorem~\ref{thm::NCGD_classical}) implies the preceding one in Ref.~\cite{smirne_coherence_2019}. 
For the ease of the reader, we restate both results here (changing slightly the terminology of the latter to the one used here).
\begin{thmbis}{thm::NCGD_classical}
Let $\left\{\mathbb{P}_n(x_n,\ldots,x_1)\right\}_{n\leq K}$ be a $K$-Markovian process (Definition~\ref{def::markov}). Then, the process is also $K$-classical (Definition~\ref{def::N-classical_statistics}) if and only if there exist a system state $\rho_{t_0}$ (at a time $t_0\leq t_1$) which is diagonal in the computational basis $\left\{\ket{x}\right\}_{x\in \mathcal{X}}$ and a set of propagators $\left\{\Lambda_{t_j,t_{j-1}}\right\}_{j=1,\ldots,K}$ which are NCGD with respect to $\left\{\ket{x}\right\}_{x\in \mathcal{X}}$, 
	such that $\rho_{t_0}$ and $\left\{\Lambda_{t_j,t_{j-1}}\right\}_{j=1,\ldots,K}$ yield
	$\left\{\mathbb{P}_n(x_n,\ldots,x_1)\right\}_{n\leq K}$ via Eq.~\eqref{eqn::Markov1}.
\end{thmbis}

\begin{theorem}[Theorem 2 of Ref.~\cite{smirne_coherence_2019}]
Let $\left\{\mathbb{P}_n(x_n,\ldots,x_1)\right\}_{n\leq K}$ be the process
fixed by the QRF Eq.~\eqref{eqn::Markov1},
with respect to a set of propagators forming a CPTP semigroup,
i.e., $\Lambda_{t_{l},t_j}=e^{\Lcal(t_{l} - t_j)}$ for any $t_l\geq t_j$
with $\Lcal$ a Lindblad generator~\cite{gorini1976,lindblad1976}, and an initial state $\rho_{t_0}$.
Then the process $\left\{\mathbb{P}_n(x_n,\ldots,x_1)\right\}_{n\leq K}$ is $K$-classical (Definition~\ref{def::N-classical_statistics})
for any $\rho_{t_0}$ diagonal in the computational
basis if and only if the 
family of propagators is NCGD in the sense that
\begin{align}
\label{eqn::NCGD22}
&\Delta\circ\Lambda_{{s_3},s_2}\circ\Delta \circ \Lambda_{s_2,s_{1}}\circ \Delta
\\ \notag
&= \Delta\circ\Lambda_{{s_3},{s_1}}\circ \Delta\quad \forall \ s_3\geq s_2\geq s_1\geq t_0\, .
\end{align}
\end{theorem}

While the two theorems are clearly related, there are two relevant differences. The new result is more operational in the sense that the statements only depend on the statistics one obtains by making the measurements in the classical basis at the specified times, whereas the statement in Ref.~\cite{smirne_coherence_2019} relies on two underlying assumptions on the Markovianity of the quantum dynamics. The first of these assumptions is that the system
multi-time statistics satisfy the QRF (Eq.~\eqref{eqn::Markov1}), and the second is that the dynamics forms a semigroup. As we will see below, the second of these assumptions can be relaxed, but the first is crucial if one wants to have the benefit of the statement in Ref.~\cite{smirne_coherence_2019}, which not only relates possible models for the statistics~\footnote{namely that it is equivalent to be able to model the statistics as coming from an evolution satisfying NCGD, as to be able to model it coming from a classical evolution}, but makes also a statement about how the possibility of modelling a process classically implies that the 
propagators referred to the actual underlying evolution have
to satisfy NGCD. To be able to make this connection between the statistics and the underlying quantum evolution, we need to restrict by assumption the types of evolutions we are considering. For the Markov case, considered here, the natural choice is the QRF (Eq.~\eqref{eqn::Markov1}), as we discussed in the main text that they are closely related.

To prove the connection between the two theorems, it is useful
to consider the following corollary to Theorem~\ref{thm::NCGD_classical} of the main text:
\begin{corollary}\label{cor::NCGD_classical}
Let $\left\{\mathbb{P}_n(x_n,\ldots,x_1)\right\}_{n\leq K}$ be the process 
fixed by the QRF Eq.~\eqref{eqn::Markov1},
with respect to a set of divisible propagators and an initial state $\rho_{t_0}$.

Let the classical dynamics of this process be invertible, that is, $\mathbb{P}_1(x_j) \neq 0$ for an initial diagonal state that is full-rank, for any $t_j<\infty$. Then, the process $\left\{\mathbb{P}_n(x_n,\ldots,x_1)\right\}_{n\leq K}$ is
$K$-classical (Definition~\ref{def::N-classical_statistics}) 
for any $\rho_{t_0}$ diagonal in the computational
basis
	if and only if the family of propagators is NCGD, see Eq.~\eqref{eqn::NCGD22}.
\end{corollary}
\begin{proof}
Let $\left\{\mathbb{P}_n(x_n,\ldots,x_1)\right\}_{n\leq K}$ be a process 
satisfying the QRF Eq.~\eqref{eqn::Markov1},
with respect to a set of divisible propagators satisfying Eq.~\eqref{eqn::NCGD22}.
Since the latter implies Eq.~\eqref{eqn::NCGD}
and the QRF implies that the process is $K$-Markovian, for any initial diagonal state
in the computational basis $K$-classicality follows from Theorem~\ref{thm::NCGD_classical}.
	
	Conversely, let the assumptions hold and the process be K-classical,
	in particular for an initial diagonal full-rank state. NCGD follows from the equation
	\begin{align}
	&
	\tr\left[{\Pcal}_{x_3}\circ\Lambda_{s_3,s_2}\circ\Delta \circ\Lambda_{s_2,s_1}\circ {\Pcal}_{x_1}\circ \Lambda_{s_1}[\rho_0]
	\right]
	\nonumber\\
	&
	=\sum_{x_2}\tr\left[{\Pcal}_{x_3} \circ \Lambda_{s_3,s_2}\circ {\Pcal}_{x_2} \circ \Lambda_{s_2,s_1}\circ {\Pcal}_{x_1} \circ \Lambda_{s_1} [\rho_0]\right]
	\notag\\
	&	
	=\sum_{x_2} \Pprob_3\left(x_3,x_2, x_1\right)
	\notag\\
	&	
	= \Pprob_2\left(x_3, x_1\right)
	\notag\\
	&
	=\tr\left[{\Pcal}_{x_3}\circ\Lambda_{s_3,s_1}\circ{\Pcal}_{x_1}  \circ\Lambda_{s_1}[\rho_0]
	\right]\, 
	\end{align}
	(for $s_3\geq s_2\geq s_1$ in $\mathcal{T}$) by linearity, since from the assumptions (invertibility of the classical dynamics and taking a diagonal, full-rank initial state) we have that ${\Pcal}_{x_1} \circ \Lambda_{s_1} [\rho_0]\neq 0 \ \forall \ x_1, s_1 < \infty$
	(for $s_1, s_2, s_3 \rightarrow \infty$,  $\Lambda_{s_i,s_j}\rightarrow \mathbbm{1}$ and NCGD holds trivially).
\end{proof}
The only difference between this corollary and Theorem~2 of Ref.~\cite{smirne_coherence_2019} is that here we have the
divisibility of the `full' propagators and
invertibility of the classical propagators in the assumptions, while there the dynamics was assumed to be of Lindblad type. This latter assumption is however strictly stronger, as
it implies divisibility and that ${\Pcal}_{x_j} \circ e^{\mathcal{L}t_j}[\rho]\neq 0 \ \forall \ x_j, \, t_j<\infty$
and for any full-rank $\rho$, since (finite-dimensional) semigroup evolutions cannot decrease the rank of a state
on a finite time~\cite{Baumgartner2008}.

In total, we have shown in this section that Theorem~2 of Ref.~\cite{smirne_coherence_2019} can be interpreted as a corollary of Theorem~\ref{thm::NCGD_classical} by using the connection between the QRF and Markovianity and further restricting to the case of Lindblad evolution.
Moreover, Corollary~\ref{cor::NCGD_classical} shows how, by relaxing such restriction and assuming
a proper invertibility condition on the classical dynamics, it is possible to establish
a one-to-one correspondence between the classicality of a process satisfying the QRF
and the NCGD property, where the latter is referred to the propagators of the actual dynamics.

\section{Absence of coherence for a model system: qubit coupled to a continuous degree of freedom}
\label{app::Absence}
In this Appendix, we provide the mathematical details missing in the main text for Example~\ref{ex::class_state}.
We begin with the expression of the global state at time $t_1$, immediately before the first measurement:
\begin{eqnarray}
	\rho_{se}&&(t_1)= \int_{-\infty}^{\infty} dp dp' f(p) f^*(p')\left(
	\rho_{00}e^{i(p-p')t_1}\ket{0p}\bra{0p'} \right. \nonumber \\
	&&\left.+\rho_{0 1}e^{i(p+p')t_1}\ket{0p}\bra{1p'} + \rho_{10}e^{-i(p+p')t_1}\ket{1p}\bra{0p'}\right.\nonumber\\
	&&\left.+
	\rho_{1 1}e^{-i(p-p')t_1}\ket{1p}\bra{1p'}
	\right).
\end{eqnarray}
After a measurement at time $t_1$ with outcome $\pm$, the state is subsequently given by
\begin{align}
	\rho^{(\pm)}_{se}(t_1)
	&=& \ket{\pm}\bra{\pm} \otimes \int_{-\infty}^{\infty} dp dp' f^{(\pm)}_{1;t_1}(p,p') \ket{p}\bra{p'},
\end{align}
where we emphasize that we have a tensor product state and have introduced the amplitude
\begin{eqnarray}
	f^{(\pm)}_{1;t_1}&&(p,p') \equiv \frac{1}{C_{t_1}^{(\pm)}}f(p) f^*(p')\left(\rho_{00}e^{i(p-p')t_1} \right.\nonumber\\
	&&\left.\pm
	\rho_{0 1}e^{i(p+p')t_1}\pm\rho_{10}e^{-i(p+p')t_1}+
	\rho_{1 1}e^{-i(p-p')t_1}\right), \nonumber
\end{eqnarray}
as well as the normalization factor $C_{t_1}^{(\pm)}=\int_{-\infty}^{\infty} dp |f(p)|^2(1\pm2\mbox{Re}(\rho_{0 1}e^{2ipt_1}))$.
Note that no $\hat{\sigma}_x$-coherence is present at this stage.

If we now let the system-environment evolve up to a certain time $\tau > t_1$, the global state will be 
\begin{eqnarray}
	\rho^{(\pm)}_{se}&&(\tau)
	= \frac{1}{2}\int_{-\infty}^{\infty} dp dp' f^{(\pm)}_{1;t_1}(p,p')\left(
	e^{i(p-p')(\tau-t_1)}\ket{0p}\bra{0p'}\right.\nonumber\\
	&&\left.\pm
	e^{i(p+p')(\tau-t_1)}\ket{0p}\bra{1p'}\pm e^{-i(p+p')(\tau-t_1})\ket{1p}\bra{0p'}\right.\nonumber\\
	&&\left.+
	e^{-i(p-p')(\tau-t_1)}\ket{1p}\bra{1p'}
	\right),
\end{eqnarray}
where the superscript $\pm$ refers to the outcome of the first measurement at time $t_1$. The corresponding system state at time $\tau$ is then given by tracing out the environmental degrees of freedom, resulting in
\begin{equation}
	\rho^{(\pm)}_{s}(\tau) = \frac{1}{2}\begin{pmatrix}
		1 &  \pm k^{(\pm)}(\tau,t_1)\\ 
		\pm k^{(\pm) *}(\tau,t_1) &1,
	\end{pmatrix}
\end{equation}
with 
\begin{align}
	k^{(\pm)}&(\tau,t_1) = \int_{-\infty}^{\infty} dp f^{(\pm)}_{1;t_1}(p,p) e^{2 i p(\tau-t_1)} \\
	&=\frac{1}{C^{(\pm)}_{t_1}} \int_{-\infty}^{\infty} dp |f(p)|^2\left(1\pm\rho_{0 1}e^{2 i p t_1}\pm\right.\nonumber\\
	&\left.\rho_{1 0}e^{-2 i p t_1}  \right) e^{2 i p(\tau-t_1)}
	\nonumber\\
	&=\frac{1}{C^{(\pm)}_{t_1}}\left( k(\tau-t_1)\pm \rho_{0 1} k(\tau) \pm \rho_{1 0} k(\tau-2 t_1)\right). \nonumber
\end{align}
Once again, we see that if the initial system state is a convex mixture of $\ket{+}$ and $\ket{-}$
and $k(t)$ is real (e.g., a Lorentzian distribution centered at 0) then no $\hat{\sigma}_x$-coherence is present at any time $\tau$. This can be seen because the reduced state can be written as in Eq.~(\ref{eq:st}) for the real $\alpha=(\pm k^{(\pm)}(\tau,t_1)+1)/2$.  
As a side remark, we note that even if the initial state had some coherences w.r.t. $\hat{\sigma}_x$, these would have been destroyed after the first measurement at time $t_1$ and, as long as $\rho_{01}\in \mathbbm{R}$, would not have been `re-generated' by the subsequent evolution.

Indeed, the argument above can be reiterated for the subsequent measurements; for instance, if we consider the global state after the second measurement at time $t_2$, we find
\begin{eqnarray}
	\rho^{(s)}_{se}(t_2)
	&=& \ket{\pm}\bra{\pm} \otimes \int_{-\infty}^{\infty} dp dp' f^{(s)}_{2;t_2,t_1}(p,p') \ket{p}\bra{p'} \nonumber
\end{eqnarray}
with
\begin{eqnarray}
f^{(s)}_{2;t_2,t_1}&&(p,p') = \frac{1}{ C^{(s)}_{t_2,t_1}}f^{(\pm)}_{1;t_1}(p,p')
\left(e^{i(p-p')(t_2-t_1)}\right.\nonumber\\
	&&\left.+sg(s)
e^{i(p+p')(t_2-t_1)} +sg(s)e^{-i(p+p')(t_2-t_1)}\right.\nonumber\\
	&&\left.+
e^{-i(p-p')(t_2-t_1)}\right),
\end{eqnarray}
where $s$ denotes the sequence of $+$ and $-$ obtained in the measurements and $sg(s)$ the sign of the corresponding product. The entire procedure can be iterated, by replacing $f^{(\pm)}_{1;t_1}(p,p')$
with $f^{(s)}_{2;t_2,t_1}(p,p')$, so that the state at any subsequent time would remain in the form of Eq.~(\ref{eq:st}), with the off-diagonal
elements given by a linear combination with real coefficients of the real function $k(t)$
evaluated at different times. In Appendix~\ref{app::CombPocket}, we will show how Example~\ref{ex::class_state} can be described using a comb representation as introduced in Section~\ref{sec::combs}.

\FloatBarrier
\begin{widetext}
\section{Comb representation of a model system: qubit coupled to a continuous degree of freedom}
\label{app::CombPocket}

In Appendix~\ref{app::Absence}, we showed the absence of coherence in the state of the system at all times for the dynamics of Example~\ref{ex::class_state}. To do so, we computed the full system-environment dynamics; however, the full knowledge of the system-environment dynamics is not necessary to understand the multi-time probabilities of observables of the system alone. Moreover, the state of the environment is often not experimentally accessible in practice, as it is typically highly complex. Therefore, it is convenient to only describe the influence that the environment has on the multi-time probabilities. Importantly, this influence, and the resulting correct descriptor of the underlying process, can be deduced by probing the system alone. 

Such a descriptor can be derived using the concept of quantum combs~\cite{chiribella_quantum_2008,chiribella_theoretical_2009}, which we briefly reviewed in Section~\ref{sec::combs}. A quantum comb contains all statistical information that can be inferred about the process it describes (on the set of times upon which it is defined). While here we will construct the comb for Example~\ref{ex::class_state} by explicitly solving the system-environment dynamics, it is important to note that it could be reconstructed experimentally by means of measurements on the system alone, without any access to or knowledge of the environmental degrees of freedom, through a generalized tomographic scheme~\cite{pollock_non-markovian_2018}.

In slight deviation from the notation of the main text, in this appendix, for better orientation, here we explicitly write the labels of the Hilbert spaces a comb acts on, and the times it is defined upon, as sub- and superscripts, respectively.

As described in Example~\ref{ex::class_state}, we start with a system-environment state $\eta_{se}(t_0=0) = \rho_s(t_0=0)\otimes \ket{\varphi^e}\bra{\varphi^e}$ where $\ket{\varphi^e}$ is fixed. As shown in Fig.~\ref{fig::HilberspacesExample}, the initial system state $\rho_s(t_0)$ is associated with the Hilbert space with label 1. The channel 
\begin{figure}
	\centering
	\includegraphics[scale=2.]{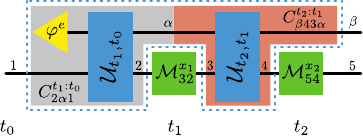}
	\caption{{\bf{Labeling of Hilbert spaces used for the comb description of Example~\ref{ex::class_state}.}} The grey box contains the comb $C^{t_1:t_0}_{2\alpha 1}$ and the red box the comb $C^{t_2:t_1}_{4\beta 3\alpha}$. The comb $C^{t_2:t_1:t_0}_{4\beta 321}$ corresponds to everything inside the dashed box and consists of the contraction of the two combs $C^{t_1:t_0}_{2\alpha 1}$ and $C^{t_2:t_1}_{4\beta 3 \alpha}$.}
	\label{fig::HilberspacesExample}
\end{figure}

\begin{align}
	\mathcal{C}^{t_1:t_0}(\rho_s)=&\Ucal_{t_1,t_0}\rho_s\otimes\ket{\varphi^e}\bra{\varphi^e} 
\end{align}
maps the initial system state to the full system-environment state at time $t_1$ directly before the intervention. The corresponding channel in comb description is given by
\begin{align}
	C^{t_1:t_0}_{2\alpha 1}=&\sum_{i,j}U_{t_1,t_0} \left( \ketbra{i}{j}_2\otimes\int_{-\infty}^{\infty}dp\int_{-\infty}^{\infty}dq f(p) f^*(q) \ketbra{p}{q}_{\alpha}\right)U^\dagger_{t_1,t_0} \otimes \ketbra{i}{j}_1 \nonumber \\
	=&\sum_{i,j}\int_{-\infty}^{\infty}dp\int_{-\infty}^{\infty}dq f(p) f^*(q) e^{i(\phi_ip-\phi_j q)t_1} \ketbra{i}{j}_2\otimes \ketbra{p}{q}_{\alpha} \otimes \ketbra{i}{j}_1,
\end{align}
where the superscripts denote the intervention times and the subscripts the Hilbert spaces on which the comb is acting. The object $C^{t_1:t_0}_{2\alpha 1}$ above is nothing other than the Choi state associated with the channel. The dynamics from time $t_1$ to time $t_2$ is similarly given by the channel
\begin{align}
	\mathcal{C}^{t_2:t_1}(\rho_{se})=&\Ucal_{t_2,t_1}\rho_{se}
\end{align}
applied to the combined system-environment state directly after the first intervention. Again, this channel admits a Choi state description
\begin{align}
	C^{t_2:t_1}_{4\beta 3\alpha}=&\sum_{i,j} \int_{-\infty}^{\infty}dp\int_{-\infty}^{\infty}dq e^{i(\phi_i p-\phi_j q) (t_2-t_1)} \ketbra{ipip}{jqjq}_{4\beta 3\alpha}.
\end{align}	
The next step is to eliminate the explicit description of the environment state on Hilbert space $\alpha$. To do this, we contract the Choi states of the two channels described above using the link product $\star$ described in Refs.~\cite{chiribella_quantum_2008,chiribella_theoretical_2009}. This leaves us with the comb describing the dynamics on both times
\begin{align}
	C^{t_2:t_1:t_0}_{4\beta 321}=&C^{t_2:t_1}_{4\beta 3\alpha} \star C^{t_1:t_0}_{2\alpha 1} \nonumber \\
	=&\tr_\alpha \left[\left(\ident_{4\beta 3}\otimes {C^{t_1:t_0}_{2\alpha 1}}^{T_\alpha}\right)\left(C^{t_2:t_1}_{4\beta 3\alpha}\otimes\ident_{21}\right)\right] \nonumber \\
	=&\int_{-\infty}^{\infty}ds \bra{s}_\alpha \sum_{i,j} \int_{-\infty}^{\infty}dp\int_{-\infty}^{\infty}dq f(p) f^*(q) e^{i(\phi_i p -\phi_j q) t_1} \ketbra{iqi}{jpj}_{2\alpha 1} \nonumber \\
	&\sum_{k,l} \int_{-\infty}^{\infty}dr\int_{-\infty}^{\infty}dt e^{i(\phi_k r-\phi_l t)(t_2-t_1)}\ketbra{krkr}{ltlt}_{4\beta 3\alpha} \ket{s}_\alpha \nonumber \\
	=&\sum_{i,j,k,l}\idotsint_{-\infty}^{\infty}ds\  dp\ dq\ dr\ dt\ \delta(s-q)  \delta(s-t) \delta(p-r) f(p) f^*(q) e^{i(\phi_i p -\phi_j q) t_1} \ketbra{ii}{jj}_{21}\nonumber\\
	&   e^{i(\phi_k r-\phi_l t)(t_2-t_1)}\ketbra{krk}{ltl}_{4\beta 3}  \nonumber \\
	=&\sum_{i,j,k,l}\iint_{-\infty}^{\infty}ds\  dp\   f(p) f^*(s) e^{i(\phi_i p -\phi_j s) t_1} \ketbra{ii}{jj}_{21}
	   e^{i(\phi_k p-\phi_l s)(t_2-t_1)}\ketbra{kpk}{lsl}_{4\beta 3}  \nonumber \\
	=&\sum_{i,j,k,l}\iint_{-\infty}^{\infty}ds\  dp\   f(p) f^*(s) e^{i(\phi_i p -\phi_j s) t_1} e^{i(\phi_k p-\phi_l s)(t_2-t_1)} \ketbra{kpkii}{lsljj}_{4\beta 321}. 
\end{align}

We can also describe the projectors corresponding to the observed measurement outcomes using Choi states, e.g., if we measured in the eigenbasis of $\hat{\sigma}_x$ and obtained outcome +, the corresponding Choi state is given by 
\begin{align}
    M^+=\ketbra{+}{+}\otimes\ident \sum_{i,j} \ketbra{ii}{jj}\ketbra{+}{+}\otimes\ident=\frac{1}{4}\sum_{i,j,k,l}\ketbra{ij}{lk}.
\end{align}
Again, using the link product, we can obtain the unnormalized joint system-environment state directly after the second intervention at time $t_2$, conditioned on the initial state of the system $\rho_s(0)$ and the interventions $M^{x_1}, M^{x_2}$ (where the superscripts $x_i$ refer to the outcomes) as follows
\begin{align}
	\rho_{se}^{(x_2,x_1)}(t_2)_{5\beta}=& C^{t_2:t_1:t_0}_{4\beta 321}\star\rho_s(t_0)_1\star M^{x_1}_{32}\star M^{x_2}_{54} \nonumber \\
	=&\tr_{4321}\left[\rho_s(t_0)^T_1\otimes {M^{x_1}_{32}}^{T_2}\otimes {M^{x_2}_{54}}^{T_4} C^{t_2:t_1:t_0}_{4\beta 321}\right].
\end{align}
For instance, if we observed the outcome + twice, the joint state after the second intervention is given by
\begin{align}
	\rho_{se}^{(+,+)}(t_2)_{5\beta}=&\sum_{i,j,k,l,m,n,x,y,a,b,c,d,f,g,h,o} \bra{fgho}_{4321}\rho_{mn}\ketbra{n}{m}_1\otimes\frac{1}{4}\ketbra{cx}{dy}_{32}\otimes\frac{1}{2}\ketbra{+}{+}_5\otimes\ketbra{a}{b}_\beta \nonumber \\
	&\iint_{-\infty}^{\infty}ds\  dp\   f(p) f^*(s) e^{i(\phi_i p -\phi_j s) t_1} \ketbra{ii}{jj}_{21}   e^{i(\phi_k p-\phi_l s)(t_2-t_1)}\ketbra{kpk}{lsl}_{4\beta 3} \ket{fgho}_{4321} \nonumber \\
	=&\frac{1}{8}\ketbra{+}{+}_5\otimes\sum_{i,j,k,l} \rho_{ij} \iint_{-\infty}^{\infty}ds\  dp\   f(p) f^*(s) e^{i(\phi_i p -\phi_j s) t_1}   e^{i(\phi_k p-\phi_l s)(t_2-t_1)} \ketbra{p}{s}_\beta \nonumber \\
	=&\frac{1}{8}\ketbra{+}{+}_5\otimes \iint_{-\infty}^{\infty}dp\  ds\   \widetilde{f}^{(+,+)}_{2;t_2,t_1}(p,s)  \ketbra{p}{s}_\beta,
\end{align}
where we have introduced 
\begin{align}
	\widetilde{f}^{(+,+)}_{2;t_2,t_1}(p,s)=&\sum_{i,j,k,l} \rho_{ij} f(p) f^*(s) e^{i(\phi_i p -\phi_j s) t_1}   e^{i(\phi_k p-\phi_l s)(t_2-t_1)}  \nonumber \\
	=&f(p) f^*(s) \left(\rho_{00} e^{i(p -s) t_1}  + \rho_{01} e^{i(p +s) t_1}  + \rho_{10} e^{-i(p +s) t_1}  +\rho_{11} e^{-i(p -s) t_1}  \right)  \nonumber\\
	&\left(e^{i(p-s)(t_2-t_1)}+	e^{i(p+s)(t_2-t_1)} +e^{-i(p+s)(t_2-t_1)}+
	e^{-i(p-s)(t_2-t_1)}\right) \nonumber \\
	=&f^{(+,+)}_{2;t_2,t_1}(p,s)
\end{align}
and checked consistency with the direct description in Appendix~\ref{app::Absence}.

Since we are mainly interested in the question of whether the obtained measurement statistics can be explained classically, we restrict our attention to the unnormalized state of the system alone, because the probability of obtaining a specific sequence of measurement outcomes is encoded in the trace of the corresponding system state. Therefore we eliminate the description of the environment by tracing over the Hilbert space $\beta$, which we can do directly at the level of the comb itself
\begin{align}
	\widetilde{C}^{t_2:t_1:t_0}_{4321}=&\tr_\beta[ C^{t_2:t_1:t_0}_{4\beta 321}]
	\nonumber\\
	=&\sum_{i,j,k,l}\int_{-\infty}^{\infty}dq \left|f(q)\right|^2 e^{i(\phi_i-\phi_j)qt_1} e^{i(\phi_k-\phi_l)q(t_2-t_1)}\ketbra{kkii}{lljj}_{4321}.
\end{align}
Following the same procedure as above, we then obtain the system state after the second intervention
\begin{align}
	\rho^{(x_2,x_1)}_s(t_2)_5=& \widetilde{C}^{t_2:t_1:t_0}_{4\beta 321}\star\rho_s(t_0)_1\star M^{x_1}_{32}\star M^{x_2}_{54}.
\end{align}
Similarly, the probability to obtain, e.g., twice the measurement result + is given by
\begin{align}
	\Pprob_2(+,t_2;+,t_1)=&\tr\left[\rho_s^{(+,+)}(t_2)_5\right].
\end{align}
If we introduce $\tau_n:=t_n-t_{n-1}$, by way of induction, we find that
\begin{figure}
	\centering
	\includegraphics[scale=2.]{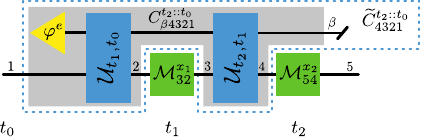}
	\caption{\textbf{Dilation for Example~\ref{ex::class_state}.} Pictorial representation of the quantum combs describing Example~\ref{ex::class_state} with two interventions.}
	\label{fig::HilberspacesExampleRela}
\end{figure}
\begin{align}
	C^{t_n::t_0}=&\sum_{i_{2n}...i_1, j_{2n}...j_1}\iint_{-\infty}^{\infty}dp\ dq\ f(p) f^*(q) \ketbra{p}{q}\bigotimes_{a=1}^{2n} e^{i(\phi_{i_a}p-\phi_{j_a}q)\tau_a}\ketbra{i_a i_a}{j_a j_a}_{2a,2a-1},\nonumber \\
	\widetilde{C}^{t_n::t_0}=&\sum_{i_{2n}...i_1, j_{2n}...j_1}\int_{-\infty}^{\infty}dp\ |f(p)|^2 \bigotimes_{a=1}^{2n} e^{i(\phi_{i_a}-\phi_{j_a})p\tau_a}\ketbra{i_a i_a}{j_a j_a}_{2a,2a-1}, 
\end{align}
where we suppressed the subscripts of the combs. As above, $C^{t_n::t_0}$ denotes the comb including the outgoing environment and $\widetilde{C}^{t_n::t_0}$ the comb describing the system alone, see Fig.~\ref{fig::HilberspacesExampleRela} for a pictorial representation. Therefore, the joint probability distribution for sequences of measurement outcomes is given by
\begin{align}
	\Pprob_n(x_n,t_n;\ldots;x_1,t_1)=&\tr\left[\rho_s(0)_1^T\bigotimes_{a=1}^n {\left(M_{2a+1,2a}^{x_a}\right)}^T \widetilde{C}^{t_n::t_0} \right].
\end{align}
\end{widetext}

\section{Alternative example for non-classical dynamics that do not create coherences}
\label{app::MIC_Ex}
Here, we provide an alternative example of a process where the state of the system is diagonal in the computational basis at all times but does not yield classical statistics. To this end, consider the following circuit (see Fig.~\ref{fig::MIC_Kolmo}): Let the initial system-environment state at time $t_0$ be a maximally entangled two qubit state $\varphi^+$ that undergoes trivial evolution between $t_0$ and $t_1$. At $t_1$ the system alone is thus in a maximally mixed state $\rho_{t_1}$ Between $t_1$ and $t_2$, the system and the environment undergo a CPTP map $\Ecal_{t_2,t_1}$ (which could---in principle---be dilated to a unitary map~\cite{stinespring1955}, but for conciseness, we restrict ourselves to the relevant part of it), that yields output $\ket{0}$ on the system, if system and environment are in the state $\varphi^+$, and $\ket{1}$ otherwise, i.e., when the system-environment state is orthogonal to $\varphi^+$. Consequently, its action can be written as  
\begin{gather}
    \Ecal_{t_2,t_1}[\eta] = \tr(\varphi^+\eta)\ketbra{0}{0} + \tr[(\ident - \varphi^+)\eta] \ketbra{1}{1} \, .
\end{gather}
It is easy to check that $\Ecal_{t_3,t_2}$ is indeed CPTP, and the state of the system at $t_2$ is a convex mixture of $\ketbra{0}{0}$ and $\ketbra{1}{1}$ for all possible experimental interventions at $t_1$; there are thus no coherences in the state of the system at any of the times $\{t_1,t_2\}$. However, this process does not satisfy the Kolmogorov condition.

To see this, consider the probabilities for a measurement in the computational basis at $t_2$, with no operation performed at $t_1$. In this case, the system-environment state before the action of $\Ecal_{t_2,t_1}$ is equal to $\varphi^+$, which means that we have $\rho_{t_2} = \ketbra{0}{0}$. Consequently, a measurement in the computational basis at $t_2$ yields the probabilities
\begin{gather}
\label{eqn::Prob}
    \Pprob_1(0,t_2) = 1 \quad \text{and} \quad  \Pprob_1(1,t_2) = 0\, .
\end{gather}
On the other hand, performing a measurement at $t_1$ and discarding the outcomes amounts to performing the completely dephasing map $\Delta_1$. Immediately after this map, i.e., right before $\Ecal_{t_2,t_1}$, the system-environment state is of the form 
\begin{gather}
    \eta_{t_1}^{se} = \frac{1}{2} \sum_{x_1} \ketbra{x_1}{x_1} \otimes \ketbra{x_1}{x_1} = \frac{1}{2}\left(\varphi^+ + \varphi^-\right)\, ,
\end{gather}
where $\varphi^- = (\sigma_z\otimes \ident)\varphi^+ (\sigma_z\otimes \ident)$ is a Bell state. Consequently, in this case the final system state $\rho_{t_2}$ is of the form  $\rho_{t_2} = \frac{1}{2}\left(\ketbra{0}{0} + \ketbra{1}{1}\right)$. Finally, the obtained probabilities for a measurement in the computational basis at $t_2$ are
\begin{align}
\notag
    &\Pprob^{\Delta_{1}}_1(0,t_2) = \sum_{x_1}\Pprob_2(x_1,t_1;0,t_2) = \frac{1}{2} \\
    \text{and} \quad &\Pprob_1^{\Delta_{1}}(1,t_2) = \sum_{x_1}\Pprob_2(x_1,t_1;1,t_2) = \frac{1}{2}\, ,
\end{align}
which does not coincide with~\eqref{eqn::Prob}. Even though the state of the system is incoherent at all considered times, i.e., appears to be classical, the multi-time statistics do not satisfy the Kolmogorov condition.

\begin{figure}
    \centering
    \includegraphics[width=0.75\linewidth]{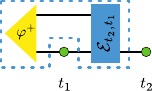}
    \caption{{\bf{Non-classical process that does not display coherences.}} The state of the system is classical, i.e., it does not contain coherences with respect to the classical basis, at any step of the process. The corresponding statistics do not satisfy the Kolmogorov conditions, though. Potential measurements are depicted as green circles. The blue dotted line signifies the comb of the process (see Sec.~\ref{sec::combs}).}
    \label{fig::MIC_Kolmo}
\end{figure}

\section{Measure of non-classicality}\label{measure_derivation}
In this appendix we derive the optimal solution of the game which defines our measure of non-classicality M(C) and show that it can be formulated as a linear program. 
We also derive the dual of this problem for completeness.

In our game, Bob can choose the points in time at which he wants Rudolph to perform projective measurements and those for which Rudolph should not interfere with the natural evolution of the system. This defines a sequence of measurements $T_i(\vec{x})=\otimes_{t_j\in \tau_i}\Phi_j^{+} \otimes_{t_k\in \tau_i^c} P_{x_k}$. Given the choice of any sequence of this form and labelling the obtained outcome sequence of the experiment by $\vec{x}$, the best strategy for Bob is to announce that the comb that was tested is $C$ if the probability for measuring outcome $\vec{x}$ with said sequence $T_i(\vec{x})$ is higher for $C$ than for $C^{\textup{Cl.}}$ (i.e., if $\tr[(C^{\textup{Cl.}}-C)T_i(\vec{x})]<0$), and announcing $C^{\textup{Cl.}}$ otherwise. The probability that he is correct when announcing $C$, given that the outcome obtained was $\vec{x}$, is given by:
\begin{align}
\mathbbm{P}(C|\vec{x})=\frac{\mathbbm{P}(C,\vec{x})}{\mathbbm{P}(\vec{x})}=\frac{\mathbbm{P}(\vec{x}|C)}{\mathbbm{P}(\vec{x})} \mathbbm{P}(C),
\end{align}
where the prior probability is $\mathbbm{P}(C)=1/2$. Denoting by $S^{\textup{Cl.}}$ the set of all $\vec{x}$ such that $\tr[(C^{\textup{Cl.}}-C)T_i(\vec{x})]>0$ and $S^{\textup{Cl.}}_c$ its complement, the probability 
that Bob wins the game is given by
\begin{align}
&\sum_{\vec{x}\in S^{\textup{Cl.}}_c} \mathbbm{P}(C|\vec{x})\mathbbm{P}(\vec{x})+\sum_{\vec{x}\in S^{\textup{Cl.}}} \mathbbm{P}(C^{\textup{Cl.}}|\vec{x})\mathbbm{P}(\vec{x}) \notag \\
&=\frac{1}{2} \left(\sum_{\vec{x}\in S^{\textup{Cl.}}_c}\mathbbm{P}(\vec{x}|C)+\sum_{\vec{x}\in S^{\textup{Cl.}}} \mathbbm{P}(\vec{x}|C^{\textup{Cl.}})\right)
\notag \\
&=\frac{1}{2} \left(\sum_{\vec{x}\in S^{\textup{Cl.}}_c}\tr[C T_i(\vec{x})]+\sum_{\vec{x}\in S^{\textup{Cl.}}}\tr[C^{\textup{Cl.}} T_i(\vec{x})]\right)
\notag \\
&=\frac{1}{2} \left(1+\sum_{\vec{x}\in S^{\textup{Cl.}}}(-\tr[C T_i(\vec{x})]+\tr[C^{\textup{Cl.}} T_i(\vec{x})])\right).
\end{align}

Assuming that both Alice and Bob play ideally, using Lemma~\ref{lem::coherence_non_Markov}, the probability $\mathbbm{P}_B(C)$ that Bob wins is given by
\begin{align}
\mathbbm{P}_B(C)=\frac{1}{2}\left(1+ M(C)\right),
\end{align}
where $M(C)$ is the solution of
\begin{alignat*}{2}
&	\text{minimize:}\ &&\max_{i} \sum_{\vec{x}\in S^{\textup{Cl.}}} \tr[(C^{\textup{Cl.}}-C) T_i(\vec{x})] \\
& 	\text{subject to:}\ &&S^{\textup{Cl.}}=\{\vec{x}|\tr[(C^{\textup{Cl.}}-C)T_i(\vec{x})]\geq 0\},\\
& &&C^{\textup{Cl.}}=\sum_{y_K,\ldots,y_1} \mathbbm{P}_K(\vec{y}) P_{y_K}\otimes\cdots \otimes P_{y_1}+\chi, \\
& &&\tr[\chi \; \cdot \; (\otimes_{t_j\in \tau_i}A_j \otimes_{t_k\in \tau_i^c} P_{z_k}) ]=0 ,\\
& &&C^{\textup{Cl.}}\geq 0,\\
& &&\tr_{{K^\inp}}[ C^{\textup{Cl.}}]=\ident_{{(K-1)^\out}}\otimes\Theta_{K-1},\\
& &&\vdots\\
& &&\tr_{{2^\inp}}[ \Theta_2]=\ident_{{1^\out}} \otimes \rho_{1^\inp},\\
& &&\mathbbm{P}_K(\vec{y})\; \mathrm{joint\, prob.\, distribution},
\end{alignat*}
where we defined $A_j:=\Phi_j^{+}-D_j$ and $\rho_{1^\inp}$ is a valid quantum state. The hierarchy of partial trace conditions on the comb written above ensure that the overall action of any instrument at a later time cannot influence previous statistics~\cite{chiribella_transforming_2008,chiribella_theoretical_2009}.

Starting from the above program, we see that $\chi$ does not contribute to the trace, as $\tr[\chi T_i(\vec{x})]$ is, by definition, a marginal of a zero-distribution (due to the third constraint above), see also the proof of Lemma~1). This leaves us with contributions only from the diagonal parts of the operator $C^{\textup{Cl.}}$, where the non-zero entries are those that correspond to $\mathbbm{P}_K(\vec{y}) P_{y_K}\otimes\cdots \otimes P_{y_1}$, which must satisfy $\tr[C^{\textup{Cl.}}]=1$ and $C^{\textup{Cl.}}\geq 0$ due to the requirement that $\mathbbm{P}_K(\vec{y})$ is a valid probability distribution. Note that for any such an operator, there exists a $\chi$ such that the total operator satisfies the additional requirements in the above program, since one simply must add terms of the form $\sum_{y_K,\ldots,y_1} \mathbbm{P}_K(\vec{y}) P_{y_K,z_K}\otimes\cdots \otimes P_{y_1,z_1}$, where the  $P_{y_j,z_j}$ are projectors up to a permutation on the input basis (i.e., $P_{y_j,z_j}=\ketbra{y_j}{y_j}_\out \otimes\ketbra {z_j}{z_j}_\inp$). We are then left with:
\begin{alignat*}{2}
&	\text{minimize:}\ &&\max_{i} \sum_{\vec{x}\in S^{\textup{Cl.}}} \tr[(C^{\textup{Cl.}}-C) T_i(\vec{x})] \\
& 	\text{subject to:}\ && 	S^{\textup{Cl.}}=\{\vec{x}|\tr[(C^{\textup{Cl.}}-C)T_i(\vec{x})]\geq 0\},\\
&	&&C^{\textup{Cl.}}=\sum_{y_K,\ldots,y_1} \mathbbm{P}_K(\vec{y}) P_{y_K}\otimes\cdots \otimes P_{y_1}, \\
& &&\mathbbm{P}_K(\vec{y})\; \mathrm{joint\, prob.\, distribution}.
\end{alignat*}
Since both $C$ and $C^{\textup{Cl.}}$ represent (up to a non-contributing $\chi$) deterministic quantum combs, we have
\begin{align}
\sum_{\vec{x}} \tr[(C^{\textup{Cl.}}-C) T_i(\vec{x})]=0
\end{align}
and thus
\begin{align}
\sum_{\vec{x}} \left|\tr[(C^{\textup{Cl.}}-C) T_i(\vec{x})]\right|=2 \sum_{\vec{x}\in S^{\textup{Cl.}}} \tr[(C^{\textup{Cl.}}-C) T_i(\vec{x})].
\end{align}
This allows us to express $M(C)$ as half of the solution of
\begin{alignat*}{2}
&	\text{minimize:}\ &&\max_{i} \sum_{\vec{x}} \left|\tr[(C^{\textup{Cl.}}-C) T_i(\vec{x})] \right|  \\
& 	\text{subject to:}\  &&	C^{\textup{Cl.}}=\sum_{y_K,\ldots,y_1} \mathbbm{P}_K(\vec{y}) P_{y_K}\otimes\cdots \otimes P_{y_1}, \\
& &&\mathbbm{P}_K(\vec{y})\; \mathrm{joint\, prob.\, distribution}.
\end{alignat*}
In order to transform this program into an LP, for every testing sequence $\{T_i(\vec{x})\}_{\vec{x}}$, we define an arbitrary order of the outcomes $\vec{x}$, i.e, we label them as $\vec{x}_{j}$. Then
\begin{align}
\max_i& \sum_{\vec{x}} \left|\tr[(C^{\textup{Cl.}}-C) T_i(\vec{x})] \right|
\end{align}
is the solution of
\begin{alignat*}{2}
&	\text{minimize:}\ && a \\
& 	\text{subject to:}\ &&a \ge  \sum_{j} \left|\tr[(C^{\textup{Cl.}}-C) T_i(\vec{x}_j)] \right|, 
\end{alignat*}
which is equivalent to
\begin{alignat*}{2}
&	\text{minimize:}\ && a \\
& 	\text{subject to:}\ &&a\ge s_{i}, \\
& &&s_{i}=\sum_j b_{ij}, \\
& &&b_{ij}\ge c_{ij}\ge -b_{ij},\\
& &&c_{ij}=\tr\left[(C^{\textup{Cl.}}-C)T_i(\vec{x}_{j})\right].
\end{alignat*}
Combining this with the outer minimization, we finally have that $M(C)$ is half of the solution of
\begin{alignat*}{2}
&	\text{minimize:}\ && a \\
& 	\text{subject to:}\ &&a\ge s_{i}, \\
& &&s_{i}=\sum_j b_{ij}, \\
& &&b_{ij}\ge c_{ij}\ge -b_{ij},\\
& &&c_{ij}=\tr\left[(C^{\textup{Cl.}}-C)T_i(\vec{x}_{j})\right], \\
& &&	C^{\textup{Cl.}}=\sum_{y_K,\ldots,y_1} \mathbbm{P}_K(\vec{y}) P_{y_K}\otimes\cdots \otimes P_{y_1}, \\
& &&\mathbbm{P}_K(\vec{y})\; \mathrm{joint\, prob.\, distribution},
\end{alignat*}
which is a linear program.

In order to simplify the numerical implementation and the derivation of the dual program, we will also order the vectors $\vec{y}$ (arbitrarily), identify $p_k$ with $\mathbbm{P}_K(\vec{y}(k))$, and define $\alpha_{ijk}$
\begin{align}
\tr\left[C^{\textup{Cl.}} T_i(\vec{x}_{j})\right]=&\sum_k p_k \alpha_{ijk}
\end{align}
for all $p_k$, i.e., 
\begin{align}
\alpha_{ijk}=\tr\left[ P_{y_K(k)}\otimes\cdots \otimes P_{y_1(k)} T_i(\vec{x}_j) \right]
\end{align} 
for the sequence $y_K(k),...,y_1(k)$ corresponding to $\vec{y}(k)$. In addition, we define 
\begin{align}
\beta_{ij}=\tr \left[ C T_i\left(\vec{x}_j\right)\right],
\end{align}
which allows us to write 
\begin{align}
c_{ij}=&\tr\left[(C^{\textup{Cl.}}-C)T_i(\vec{x}_{j})\right] \notag \\
=& \sum_k p_k \alpha_{ijk}-\beta_{ij}.
\end{align}

Then, the above optimization problem is equivalent to
\begin{alignat*}{2}
&	\text{minimize:}\ && a \\
& 	\text{subject to:}\ &&\sum_j b_{ij}-a\le0, \\
& &&\sum_k p_k \alpha_{ijk}-\beta_{ij} -b_{ij}\le0, \\
& &&-\sum_k p_k \alpha_{ijk}+\beta_{ij} -b_{ij}\le0,\\
& &&\sum_k p_k -1=0, \\
& &&	p_k, a, b_{ij}\ge0. 
\end{alignat*}
The Lagrangian corresponding to this problem is
\begin{align}
L(a, &p_k, b_{ij}, X_i, Y_{ij}, Z_{ij}, W) \nonumber \\
=&a \left[1-\sum_i X_i \right]+\sum_{ij} b_{ij} \left[X_i-Y_{ij}-Z_{ij}\right] \nonumber \\
&+\sum_k p_k \left[\sum_{ij}\alpha_{ijk}\left(Y_{ij}-Z_{ij}\right)-W\right] \nonumber \\
&+W+\sum_{ij} \beta_{ij}\left(Z_{ij}-Y_{ij}\right)
\end{align}
and the dual function explicitly written
\begin{align}
q( X_i,& Y_{ij}, Z_{ij}, W) \nonumber\\
=&\inf_{p_k\ge0, a, b_{ij}} L(a, p_k, b_{ij}, X_i, Y_{ij}, Z_{ij}, W),
\end{align}
where we used that $a, b_{ij} \ge0$ is implicit in the remaining conditions.
The dual problem is then given by
\begin{alignat*}{2}
&	\text{maximize:}\ && W+\sum_{ij} \beta_{ij}(Z_{ij}-Y_{ij}) \\
& 	\text{subject to:}\ &&\sum_i X_i=1, \\
& &&X_i-Y_{ij}-Z_{ij}=0\quad\forall\; ij , \\ 
& &&\sum_{ij} \alpha_{ijk}\left(Y_{ij}-Z_{ij}\right)-W\ge0 \quad\forall\; k,\\
& &&X_i, Y_{ij}, Z_{ij}\ge0, \\
& &&W \in \mathbbm{R},
\end{alignat*} 
which can straightforwardly be reformulated as
\begin{alignat*}{2}
&	\text{maximize:}\ &&  \Omega\\
& 	\text{subject to:}\ && \Omega \leq \sum_{ij} \left(\alpha_{ijk}-\beta_{ij}\right)\left(2 Y_{ij}-X_{i}\right)\quad\forall \ k, \\
& && \sum_i X_i=1, \\
& &&X_i, Y_{ij}, X_i-Y_{ij} \ge0,  \\
& && \Omega \in \mathbbm{R}.
\end{alignat*}
Evidently, the above considerations are amenable to many extensions but that is the matter of future work.

\section{Non-discord-creating maps}
\label{app::discord_zero}
Here, for comprehensiveness, we characterize the set of maps $\Gamma: \Bcal(\Hcal^\inp_s \otimes \Hcal^\inp_e) \rightarrow \Bcal(\Hcal^\out_s \otimes \Hcal^\out_e)$ that map discord-zero states to discord-zero states, where we mean discord-zero with respect to the classical basis. Such system-environment maps form a subset of the NDGD maps of Definition~\ref{def::NDGD} (in the sense that a set of them would satisfy Eq.~\eqref{eqn::NDCG}) and would thus lead to classical statistics on the level of the system. However, for classical statistics, it is not necessary that the underlying maps do not create discord.  

To facilitate notation, throughout this Appendix, we will denote discord-zero states as \DOpr states, and maps that do not create discord as \DOpr maps. We have the following lemma:
\begin{lemma}[Structure of \DOpr maps]
\label{lem::DO}
The Choi state $G$ of a \DOthm map $\Gamma: \Bcal(\Hcal^\inp_s \otimes \Hcal^\inp_e) \rightarrow \Bcal(\Hcal^\out_s \otimes \Hcal^\out_e)$ is of the form 
\begin{gather}
\label{eqn::FormDO}
    G = \sum_{k,j = 1}^{d_s} p_{k|j} \Pi_k^{\out} \otimes \Pi_j^{\inp} \otimes O_{jk}^{\out \inp} + G^\perp\, ,
\end{gather}
where $\{\Pi_l^{\inp/\out}\}$ are orthogonal rank-$1$ projectors on $\Hcal_{s}^{\inp/\out}$ that are diagonal in the computational basis, $O_{jk}^{\out\inp} \in  \Bcal(\Hcal^\out_e \otimes \Hcal^\inp_e)$ is the Choi state of a CPTP map $\Omega_{jk}: \Bcal(\Hcal^\inp_e) \rightarrow \Bcal(\Hcal^\out_e)$, $p_{k|j}$ is a conditional probability distribution, i.e., $\sum_k p_{k|j} = 1$ and $p_{k|j} \geq 0$, and $G^\perp \in \Bcal(\Hcal^\out_s \otimes \Hcal^\out_e \otimes \Hcal^\inp_s \otimes \Hcal^\inp_e)$ is orthogonal to the set of \DOthm states, i.e., $\tr[(\mathbbm{1} \otimes \rho) G^\perp] = 0$ for all \DOthm states $\rho \in \Bcal(\Hcal^\inp_s \otimes \Hcal^\inp_e)$.
\end{lemma}
Before we prove this lemma, we emphasize its structural relation to the representation of MIOs, i.e., the structure of maps $\Fcal: \Bcal(\Hcal^\inp_s) \rightarrow \Bcal(\Hcal^\out_s)$ that map incoherent states $\rho \in \Xi \subset \Bcal(\Hcal^\inp_s)$ onto incoherent states $\rho' = \Fcal[\rho] \in \Xi \subset \Bcal(\Hcal^\out_s)$, where $\Xi$ denotes the set of incoherent states with respect to the computational basis. The Choi state $F$ of the map $\Fcal$ is a positive element of $\Bcal(\Hcal^\out_s \otimes \Hcal^\inp_s)$. Choosing a basis $\{\tau^{\out}_k\}_{k=1}^{d_s^2}$ and $\{\omega^{\inp}_j\}_{j=1}^{d_s^2}$ for $\Bcal(\Hcal^\out_s)$ and  $\Bcal(\Hcal^\inp_s)$, respectively, any $F$ can be written as 
\begin{gather}
\label{eqn::FormF}
    F = \sum_{j,k} f_{j k} \, \tau^{\out}_k \otimes \omega^\inp_j\, ,
\end{gather}
where $f_{jk} \in \mathbbm{R}$. We can choose the basis $\{\omega^{\inp}_j\}$ to consist of the $d_s$ rank-$1$ projectors $\Pi^\inp_j$ in the computational basis and $d_s(d_s-1)$ elements $\widetilde \Pi^\inp_s$ that are orthogonal to these projectors, i.e., such that $\tr(\Pi^\inp_j\widetilde{\Pi}^\inp_s) = 0$ (e.g., one could choose the off-diagonal elements $\ketbra{m}{n} + \ketbra{n}{m}$ and $\mathbbm{i}(\ketbra{m}{n} - \ketbra{n}{m})$). With this choice of basis elements Eq.~\eqref{eqn::FormF} reads
\begin{gather}
    F = \sum_{j,k} f_{kj} \, \tau^{\out}_k \otimes \Pi^\inp_j + \sum_{r,s} \widetilde f_{rs}  \, \tau^{\out}_r \otimes \widetilde \Pi_{s}^{\inp}\, .
\end{gather}
Imposing the requirement that $\Fcal$ does not create coherences with respect to the classical basis then yields 
\begin{gather}
\label{eqn::FormFincoherent}
    F = \sum_{j,k} p_{k|j}\, \Pi_k^\out \otimes \Pi_j^\inp + \sum_{r,s} \widetilde f_{rs} \, \tau^{\out}_r \otimes \widetilde \Pi_{s}^{\inp}\, ,
\end{gather}
where $p_{k|j} \geq 0$, $\sum_k p_{k|j} = 1$, and $\tau_{r}^\out \in \Bcal(\Hcal_s^{\out})$. Indeed, an $F$ of the form of Eq.~\eqref{eqn::FormFincoherent} yields an incoherent output state for any incoherent input state $\rho_\mathrm{cl} = \sum_{r=1}^{d_s} q_r \Pi_r^{\inp} \in \Xi$:
\begin{gather}
    \Fcal[\rho_\mathrm{cl}] = \tr_\inp\left[\left(\mathbbm{1}^\out \otimes \rho_{\mathrm{cl}}^{\mathrm{T}}\right) F \right] = \sum_{kr} p_{k|r} q_r \Pi_k^\out\,.
\end{gather}
Importantly, Eq.~\eqref{eqn::FormFincoherent} constitutes a decomposition of the form $F = F^{\parallel} + F^\perp$, where $F^{\parallel}  = \sum_{j,k} p_{k|j}\, \Pi_k^\out \otimes \Pi_j^\inp$ encapsulates the action of $\Fcal$ on incoherent states, and $F^\perp$ is such that all incoherent states lie in its kernel, i.e., $\tr(\rho F^\perp) = 0$ for all $\rho \in \Xi$. The fact that $F^\perp$ does not have to vanish in order for $\Fcal$ to be an MIO demonstrates in a transparent way the (well-known) fact that there are MIOs that necessitate coherent resources for their implementation~\cite{chitambar2016critical,chitambar2016comparison,marvian2016quantify}.

As emphasized throughout the main body of this paper, \DOpr states reduce to incoherent ones when the environment is trivial. Consequently, \DOpr maps are the natural extension of MIOs, and the proof of Lemma~\ref{lem::DO} follows similar logic to the above proof for the structural properties of MIOs:
\begin{proof}
Employing the reasoning that led to Eq.~\eqref{eqn::FormFincoherent}, any \DOpr map $\Gamma$ has a Choi state $G$ of the form 
\begin{align}
\notag
 G = &\sum_{kj\mu\nu} g_{k\mu j\nu}\,  \tau^{\out}_k \otimes \Pi_j^{\inp} \otimes N_{\mu \nu}^{\out \inp} \\
 \label{eqn::Discord0maps}
 &+ \sum_{r s \mu \nu} \widetilde g_{r\mu s\nu}\,  \tau^{\out}_r \otimes \widetilde \Pi_s^{\inp} \otimes N_{\mu \nu}^{\out \inp}\,,
 \end{align}
where $g_{k\mu j\nu}, \widetilde g_{r\mu s\nu} \in \mathbbm{R}$ and $\{N_{\mu\nu}^{\out \inp}\}_{\mu,\nu=1}^{d_e^2}$ is a basis of $\Bcal(\Hcal_e^{\out} \otimes \Hcal_e^{\inp})$. As for the case of MIOs, Eq.~\eqref{eqn::Discord0maps} constitutes a decomposition $G = G^{\parallel} + G^{\perp} $, where $G^\perp$ is orthogonal to the set of \DOpr states. Consequently, the action of $\Gamma$ on any \DOpr state is entirely encapsulated in $G^{\parallel}$ and it remains to show that this term is of the form given in the lemma. To this end, we note that a map $\Gamma$ is \DOpr iff it maps \textit{any} state of the form $\Pi^{\inp}_\ell \otimes \eta^{\inp}_\ell$ to a \DOpr state. Letting $\Gamma$ act on such a product state, we obtain
\begin{align}
\notag
    \Gamma[\Pi^{\inp}_\ell \otimes \eta^{\inp}_\ell] &= \tr_{\inp}\left\{\left[\mathbbm{1}^\out\otimes \left(\Pi^{\inp}_\ell \otimes \eta^{\inp}_\ell\right)^{\mathrm{T}}\right] G^\parallel\right\} \\
    \label{eqn::DOpar1}
    &=\sum_{k\mu\nu} g_{k\mu \ell\nu}\, \tau^\out_k \otimes \tr_{\inp}\left[\left(\mathbbm{1}^\out \otimes \eta_\ell^{\inp \mathrm{T}}\right) N_{\mu \nu}^{\out \inp}\right] \\
    \label{eqn::DOpar2}
    &\stackrel{!}{=} \sum_r p_{r|\ell}\, \Pi_t^{\out} \otimes \xi_{r|\ell}^\out\, ,
\end{align} 
where $\sum_r p_{r|\ell} = 1$ and $p_{r|\ell} \geq 0$, and $\xi_r^{\out} \in \Bcal(\Hcal_e^{\out})$ are states of the environment. The last line of Eq.~\eqref{eqn::DOpar2} stems from the requirement that $\Gamma$ is a \DOpr map, and the remaining open index $\ell$ signifies that the resulting output state depends on the input state $\Pi^{\inp}_\ell \otimes \eta^{\inp}_\ell$. In the same way as above, we can choose the basis $\{\tau_k^\out\}$ to consist of projectors $\{\Pi^\out_k\}$ onto the computational basis and elements that are orthogonal to these projectors. Then, comparing Eqs.~\eqref{eqn::DOpar1} and~\eqref{eqn::DOpar2}, we see that all of the terms of $G^\parallel$ where $\tau^\out_k$ is not a projector onto the computational basis must vanish. Finally, the terms $N_{\mu \nu}^{\out \inp}$ have to be such that $\tr_{\inp}\left[\left(\mathbbm{1} \otimes \eta_l^{\inp \mathrm{T}}\right) \sum_{\mu \nu} g_{k\mu l\nu} N_{\mu \nu}^{\out \inp}\right]$ yields the correct output state $p_{k|\ell} \xi_{k|\ell}^\out$. Consequently, $\sum_{\mu \nu} g_{k\mu l\nu} N_{\mu \nu}^{\out\inp}$ can be chosen to be (up to normalization $p_{k|\ell}$) the Choi state $O^{\out\inp}_{k\ell}$ of a CPTP map. Putting these observations together yields Eq.~\eqref{eqn::FormDO}.
\end{proof}

\section{Proof that NDGD \texorpdfstring{$\Rightarrow$}{} classical process}
\label{app::proof_NDGD}
For the proof of Theorem~\ref{thm::NDCG}, we employ the fact that the completely dephasing map has no influence on the outcomes of a measurement in the computational basis, i.e., 
\begin{gather}
\label{eqn::Commute_class}
    \Pcal_{x_j} = \Delta_j \circ \Pcal_{x_j} = \Pcal_{x_j} \circ \Delta_j \quad \forall \ x_j\, .
\end{gather}
The probability $\Pprob_k(x_k,\dots,x_1)$ to measure outcomes $\{x_k,\dots,x_1\}$ at times $\{t_k,\dots,t_1\}$ is given by 
(see Eq.~\eqref{eq:extra5})
\begin{gather}
\label{eqn::NDGD_class_Prob}
    \tr\{(\Pcal_{x_k}{\otimes \Ical^e})\circ \cdots \circ \Gamma_{t_2,t_1} \circ (\Pcal_{x_1}\otimes \Ical^e)[\eta^{se}_{t_1}]\}\, ,
\end{gather}
where $\{\Gamma_{t_j,t_{j-1}}\}$ are system-environment CPTP maps and $\eta^{se}_{t_1}$ is the system-environment state at time $t_1$.
Summing this probability distribution over the outcomes at time $t_j$ amounts to replacing $\Pcal_{x_j}$ in~\eqref{eqn::NDGD_class_Prob} by $\Delta_j$. `Zooming in' on the relevant time
(and leaving the $\Ical^e$ implicit), we see that
\begin{align}
\notag
    & \Pcal_{x_{j+1}} \circ \Gamma_{t_{j+1},t_j} \circ \Delta_j \circ \Gamma_{t_j,t_{j-1}} \circ \Pcal_{x_{j-1}}\\
        \notag&= \Pcal_{x_{j+1}} \circ \Delta_{j+1} \circ \Gamma_{t_{j+1},t_j} \circ \Delta_j \circ \Gamma_{t_j,t_{j-1}} \circ \Delta_{j-1} \circ \Pcal_{x_{j-1}}\\
\label{eqn::Classical_1}
   & = \Pcal_{x_{j+1}} \circ \Gamma_{t_{j+1},t_j} \circ \Ical_j \circ \Gamma_{t_j,t_{j-1}} \circ \Pcal_{x_{j-1}}\, ,
\end{align}
where we have used Eq.~\eqref{eqn::Commute_class} in the first line, and both the fact that the dynamics is NDGD 
and Eq.~\eqref{eqn::Commute_class} in the second line. As Eq.~\eqref{eqn::Classical_1} holds for arbitrary times $t_j$, it implies that for NDGD dynamics, the completely dephasing map cannot be distinguished from the identity map when the process is probed by measurements in the computational basis, 
which implies that the Kolmogorov condition holds for any joint probabilities with at least 3 different times. For the 2-time joint probabilities, we can exploit, along with the NDGD property, the fact that the initial state is zero-discord. We have
\begin{align}
\notag &\sum_{x_1} \Pprob_2(x_2,x_1)  \nonumber\\
&= \tr\{\Pcal_{x_2}\circ \Gamma_{t_2,t_1} \circ \Delta_1\circ\Gamma_{t_1,t_0}[\eta^{se}_{t_0}]\} \nonumber\\
&= \tr\{\Pcal_{x_2}\circ \Delta_2\circ \Gamma_{t_2,t_1} \circ \Delta_1\circ\Gamma_{t_1,t_0}[\eta^{se}_{t_0}]\} \nonumber\\
& =\tr\{\Pcal_{x_2}\circ \Delta_2\circ \Gamma_{t_2,t_1} \circ \Delta_1\circ \Gamma_{t_1,t_0}\circ\Delta_0[\eta^{se}_{t_0}]\} \nonumber \\
& =\tr\{\Pcal_{x_2}\circ \Gamma_{t_2,t_0}[\eta^{se}_{t_0}] \}
= \Pprob_1(x_2),
\end{align}
where we used Eq. \eqref{eqn::NDGD_class_Prob} and $\sum_{x_1}\Pcal_{x_1}=\Delta_1$ in the first line,
Eq.~\eqref{eqn::Commute_class} in the second line,
the invariance of the initial zero-discord state with respect to $\Delta_0$
in the third line, and finally the definition of NDGD dynamics,
Eq.~\eqref{eqn::Commute_class}, and the invariance of $\eta^{se}_{t_0}$ in the fourth line. Consequently, the resulting statistics satisfy all of the Kolmogorov conditions and are thus classical.

\section{Classicality \texorpdfstring{$\cancel \Rightarrow$}{} NDGD}
\label{app::Class_non_NDGD}

Here, we provide an example of dynamics that are not NDGD, yet lead to classical dynamics, thus demonstrating that it is not necessary for a dynamics to be NDGD in order for it to appear classical. 
\begin{figure}
\vspace{.5cm}
\centering
 \includegraphics[width=0.95\linewidth]{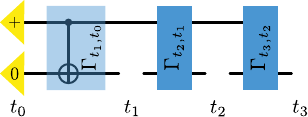}
 \caption{\textbf{Non-NDGD dynamics that leads to classical statistics.} The first map $\Gamma_{t_1,t_0}$ (blue transparent box) performs a CNOT gate on the system and the environment. The subsequent CPTP map $\Gamma_{t_2,t_1}$ maps $\varphi^+$ and $\mathbbm{1}/4$ onto two different system-environment states with the same reduced system state $\rho_{t_3} = \mathbbm{1}/2$. The final CPTP map $\Gamma_{t_3,t_2}$ is such that it induces a unital dynamics on the system. Consequently, the system state at $t_1, t_2,$ and $t_3$ is maximally mixed independent of whether the completely dephasing, or the identity map was implemented at $t_1$ and $t_2$.}
 \label{fig::Non_NDGD}
 \vspace{-.5cm}
\end{figure}
We consider the following situation (see Fig.~\ref{fig::Non_NDGD} for a graphical representation): Let the system of interest be a qubit that is initially in state $\ket{0}$ and let the initial environment be in a plus state, i.e., $\tau_{t_0}^e = \tfrac{1}{\sqrt{2}}(\ket{0}+\ket{1}))$. The first evolution $\Gamma_{t_1,t_0}$ from $t_0$ to $t_1$ is a CNOT gate, such that the system-environment state at $t_1$ is a maximally entangled state. The second evolution $\Gamma_{t_2,t_1}$ from $t_1$ to $t_2$ is such that it yields a system-environment state $\mathbbm{1}_s/2 \otimes \ketbra{0}{0}$ if the $s e'$ input state is $\varphi^+_{se'}$, and $\mathbbm{1}_s/2 \otimes \ketbra{1}{1}$ otherwise. Consequently, when the completely dephasing map is applied at $t_1$, the system-environment state at $t_2$ is $\mathbbm{1}_s/2 \otimes \mathbbm{1}_e/2$, while it is equal to $\mathbbm{1}_s/2 \otimes \ketbra{0}{0}$ if the identity map was implemented, and as such, the dynamics is not NDGD. However, the system state is always maximally mixed, independent of whether $\Delta_1$ or $\Ical_1$ was implemented at time $t_1$. To make the example non-trivial, we add a third free dynamics $\Gamma_{t_3,t_2}$ from $t_2$ to $t_3$. We choose $\Gamma_{t_3,t_2}$ such that it induces a unital dynamics on the level of the system, independent of the environment state at $t_2$. This happens, e.g., when the corresponding system-environment Hamiltonian is of product form, i.e., $H_{se} = H_s \otimes H_e$, independent of the explicit form of the respective terms~\cite{sakuldee_non-markovian_2018}. With this final dynamics, the system state at each of the times $t_1,t_2,$ and $t_3$ is maximally mixed, and the resulting statistics satisfy Kolmogorov conditions, i.e., they are classical. 

\section{Example of a genuinely quantum process}
\label{app:genuinelyquantumprocess}

\begin{figure*}
	\centering
	\includegraphics[scale=0.95]{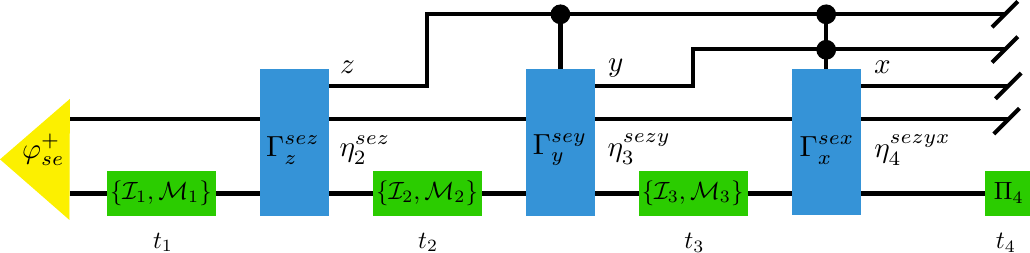}
	\caption{\textbf{Genuinely quantum process.} The system-environment begin in a maximally entangled Bell state $\varphi^+_{se}$. As described in the text, the process dynamics consists of a sequence of maps, $\Gamma_z^{sez}, \Gamma_y^{sey}, \Gamma_x^{sex}$, that either output $\varphi^+_{se}$ or else bias the system in either the $z$-, $y$- or $x$-basis respectively (see Eqs.~\eqref{eq:gammazapp} -- \eqref{eq:gammaxapp}). The overall implementation of each of these maps is controlled on the joint state of all previous classical flag states $z,y,x$, which encode whether or not the system has already been biased. We show that this process is genuinely quantum by tracking the system-environment state throughout the dynamics, conditioned on whether the identity map $\mathcal{I}$ or an arbitrary CPTP map $\Mcal_i$ was implemented at time $t_i$; the labels $\eta_2^{sez}, \eta_3^{sezy}$ and $\eta_4^{sezyz}$ refer to the overall joint state immediately prior to the interrogation at the relevant time (see Eqs.~\eqref{eq:appetajointIII}, \eqref{eq:appetajointIIL}, \eqref{eq:appetajointILI3}, \eqref{eq:appetajointILI}, \eqref{eq:appetajointLII2} and \eqref{eq:appetajointLII}). In particular, we show that there does not exist a non-pathological POVM $\Pi_4$ that an experimenter can implement at $t_4$ such that the four sequences $\{ \mathcal{I}_1,\mathcal{I}_2,\mathcal{I}_3\}$, $\{ \Mcal_1,\mathcal{I}_2,\mathcal{I}_3\}$, $\{ \mathcal{I}_1,\Mcal_2,\mathcal{I}_3\}$ and $\{ \mathcal{I}_1,\mathcal{I}_2,\Mcal_3\}$ cannot be distinguished, thereby proving that the process is genuinely quantum.}
	\label{fig:appgenuinelyquantum}
\end{figure*}

Consider the following process, depicted in Fig.~\ref{fig:appgenuinelyquantum}, which is a variation on that presented in Section~\ref{sec:genuinelyquantumprocess}. The process begins with a two-qubit system-environment state in the Bell pair $\varphi^+_{se}$, the system part of which the experimenter has access to measure at $t_1$. Following this, the process `performs' the CPTP system-environment map $\Gamma_z^{sez} : \mathcal{B}(\mathcal{H}^{s^\inp} \otimes \mathcal{H}^{e^\inp}) \to \mathcal{B}(\mathcal{H}^{s^\out} \otimes \mathcal{H}^{e^\out} \otimes \mathcal{H}^{z^\out})$, whose action is as follows: it measures its joint inputs in the Bell basis and if the measurement outcome corresponds to $\varphi^+_{se}$, it outputs a $\varphi^+_{se}$ system-environment state as well as a classical \emph{flag} state $\ket{0}_z$; on the other hand, if the measurement outcome does not correspond to $\varphi^+_{se}$, it outputs a system-environment state whose system part is a pure state in the $z$-basis, and sets the flag state to $\ket{1}_z$ to indicate that the system state has been biased in the $z$-basis. The action of the map is thus as follows: 
\begin{align}\label{eq:gammazapp}
    \Gamma_z^{sez}[\eta_{se}] =& \tr{[\varphi_{se}^+  \eta_{se}]} \varphi_{se}^+ \otimes \ket{0}\bra{0}_{z} \notag \\ +& \tr{[(\mathbbm{1}_{se} - \varphi_{se}^+) \eta_{se}]}  \ket{0} \bra{0}_{s} \otimes \tau_{e} \otimes \ket{1}\bra{1}_{z}.
\end{align}
For this map (and all that follow in this example), the output state of the environment when the $\varphi^+_{se}$ outcome is \emph{not} recorded is irrelevant for our argument; as such, we simply write a generic quantum state $\tau_{e}$.

Following this part of the dynamics, the experimenter has access to measure the system at time $t_2$. The subsequent dynamics of the process is controlled on the state of the classical $z$ flag: if it is in the state $\ket{0}_z$, the system-environment is subject to a similar dynamics as before, $\Gamma_y^{sey} : \mathcal{B}(\mathcal{H}^{s^\inp} \otimes \mathcal{H}^{e^\inp}) \to \mathcal{B}(\mathcal{H}^{s^\out} \otimes \mathcal{H}^{e^\out} \otimes \mathcal{H}^{y^\out})$; however, this time if the Bell basis measurement outcome does not correspond to $\varphi^+_{se}$, the system is biased in the $y$-basis, e.g., set to the $-1$ eigenstate of $\sigma^{(y)}$, $\ket{-^{(y)}} := \tfrac{1}{\sqrt{2}} (\ket{0} - i \ket{1})$, with a classical $y$ flag set to the state $\ket{1}_y$ and sent forward. If, on the other hand, the $z$ flag is in the state $\ket{1}_z$, the system-environment undergoes trivial dynamics (i.e., is subject to the identity map) and the $y$ flag is set to $\ket{0}_y$. In either case, the previous $z$ flag state is also sent forward unperturbed. Thus,  between $t_2$ and $t_3$, the system-environment evolves conditionally according to
\begin{align}\label{eq:gammayapp}
    &\underline{z = 0}: \notag \\
    &\Gamma_y^{sey}[\eta_{se}] = \tr{[\varphi_{se}^+  \eta_{se}]} \varphi_{se}^+ \otimes \ket{0}\bra{0}_{ y}  \notag \\
    &+ \tr{[(\mathbbm{1}_{se} - \varphi_{se}^+)  \eta_{se}]} \ket{-^{(y)}} \bra{-^{(y)}}_{s} \otimes \tau_{e} \otimes \ket{1}\bra{1}_{ y} \notag \\
    &\underline{z = 1}: \notag \\
    &\mathcal{I}^{se}[\eta_{se}] \otimes \ket{0}\bra{0}_{y}.
\end{align}
Following this, the experimenter has access to the system at $t_3$.

The final portion of the dynamics between $t_3$ and $t_4$ follows a similar construction to above, but the implementation of the map $\Gamma_x^{sex} : \mathcal{B}(\mathcal{H}^{s^\inp} \otimes \mathcal{H}^{e^\inp}) \to \mathcal{B}(\mathcal{H}^{s^\out} \otimes \mathcal{H}^{e^\out} \otimes \mathcal{H}^{x^\out})$ is controlled on the \emph{joint} state of the $z$ and $y$ classical flags. If $zy=00$, the system-environment is measured in the Bell basis: if the measurement outcome does not correspond to $\varphi^+$, the system is biased in the $x$-basis, e.g., set to the $-1$ eigenstate of $\sigma^{(x)}$, $\ket{-^{(x)}} := \tfrac{1}{\sqrt{2}} (\ket{0} - \ket{1})$, with a classical $x$ flag set to the state $\ket{1}_x$ and sent forward. If, on the other hand, $zy \neq 00$, the system-environment undergoes trivial dynamics (i.e., is subject to the identity map) and the $x$ flag is set to $\ket{0}_x$. Mathematically, the controlled dynamics is described as 
\begin{align}\label{eq:gammaxapp}
    &\underline{zy = 00}: \notag \\
    &\Gamma_x^{sex}[\eta_{se}] = \tr{[\varphi_{se}^+  \eta_{se}]} \varphi_{se}^+ \otimes \ket{0}\bra{0}_{ x}  \notag \\
    &+ \tr{[ (\mathbbm{1}_{se} - \varphi_{se}^+) \eta_{se}]} \ket{-^{(x)}} \bra{-^{(x)}}_{s} \otimes \tau_{e} \otimes \ket{1}\bra{1}_{ x} \notag \\
    &\underline{zy = 10, 01}: \notag \\
    &\mathcal{I}^{se}[\eta_{se}] \otimes \ket{0}\bra{0}_{x}.
\end{align}
Note that the flag state $zy=11$ cannot occur. Finally, the environment and all flag states are discarded and the experimenter has access to the system at $t_4$, concluding the process.

We now show that there exist no unrestricted measurement scheme for this process such that the statistics observed are classical, i.e., we prove that the process is genuinely quantum. As in the main text, we do this by considering the state of the system to be measured at the final time $t_4$ conditioned on a history of identity maps and arbitrary CPTP maps $\{\Mcal_1,\Mcal_2,\Mcal_3 \}$ implemented at various sets of earlier times. In each case, by demanding classicality we end up with a different constraint on the structure of the POVM at the final time, and the only valid POVMs that simultaneously satisfy all conditions are the pathological ones that do not reveal anything about the process. The conclusion is that any non-pathological POVM at $t_4$ will be able to distinguish between previous implementations of the identity map or an arbitrary non-pathological instrument at a given time, therefore picking up on the invasiveness of (at least some of) the previous interrogations and leading to non-classical statistics.

Consider first the scenario where the experimenter implements identity maps at the first three times, $\mathcal{I}_1, \mathcal{I}_2, \mathcal{I}_3$. In this case, the overall state immediately prior to the measurement at $t_4$ is
\begin{align}\label{eq:appetajointIII}
    \eta_4^{sezyx}(\mathcal{I}_1, \mathcal{I}_2, \mathcal{I}_3) = \varphi_{se}^+ \otimes \ket{000}\bra{000}_{zyx}.
\end{align}
The reduced system state is then maximally mixed:
\begin{align}\label{eq:appetasystemIII}
    \eta_4^{s}(\mathcal{I}_1, \mathcal{I}_2, \mathcal{I}_3) = \frac{\mathbbm{1}}{2}.
\end{align}

Next, consider the case where the experimenter implements the identity map at the first two times, $\mathcal{I}_1, \mathcal{I}_2$, followed by an arbitrary CPTP map $\Mcal_3\neq \mathcal{I}_3$ at $t_3$. The system-environment joint state immediately prior to $t_3$ is $\varphi^+_{se}$, since the previous identity maps on the system and the dynamics $\Gamma_z^{sez}, \Gamma_y^{sey}$ leading up to $t_3$ preserve the initial state; moreover, the $zy$ flag is in the joint state $00$, since both previous Bell basis measurements are necessarily successful. Now, the system-local CPTP map $\Mcal_3 \neq \mathcal{I}_3$ will perturb the joint system-environment state, and so the map $\Gamma_x^{sex}$ (which is implemented due to the joint state of the input flags) only successfully records the outcome corresponding to $\varphi^+_{se}$ with some probability $r = \tr{\left[\varphi^+_{se} (\Mcal_3^s \otimes \mathcal{I}^e) [\varphi^+_{se}] \right]} < 1$; otherwise, the system is biased in the $x$-basis. The total joint state immediately prior to $t_4$ in this scenario is then
\begin{align}\label{eq:appetajointIIL}
    &\eta_4^{sezyx}(\mathcal{I}_1, \mathcal{I}_2, \Mcal_3) = r \varphi_{se}^+ \otimes \ket{000}\bra{000}_{zyx} \notag \\
    &+ (1-r) \ket{-^{(x)}} \bra{-^{(x)}}_{s} \otimes \tau_{e} \otimes \ket{001}\bra{001}_{zyx}.
\end{align}
The reduced system state is thus biased in the $x$-basis:
\begin{align}\label{eq:appetasystemIIL}
    \eta_4^{s}(\mathcal{I}_1, \mathcal{I}_2, \Mcal_3) = \frac{r}{2} \mathbbm{1} + (1-r) \ket{-^{(x)}} \bra{-^{(x)}}.
\end{align}

Next, consider the case where the experimenter implements the identity map at the first and third time, $\mathcal{I}_1, \mathcal{I}_3$, with an arbitrary CPTP map $\Mcal_2\neq \mathcal{I}_2$ implemented in between at time $t_2$. The system-environment joint state immediately prior to $t_2$ is $\varphi^+_{se}$, since the previous identity map on the system and the dynamics $\Gamma_z^{sez}$ prior to $t_2$ again preserve the initial state; moreover, the $z$ flag is in the state $0$, since the earlier Bell basis measurement is necessarily successful. Again, the system-local CPTP map $\Mcal_2 \neq \mathcal{I}_2$ will perturb the joint system-environment state, and so the map $\Gamma_y^{sey}$ (which is implemented due to the state of the input flag) only successfully records the outcome corresponding to $\varphi^+_{se}$ with some probability $q = \tr{\left[\varphi^+_{se} (\Mcal_2^s \otimes \mathcal{I}^e) [\varphi^+_{se}] \right]} < 1$; otherwise, the system is biased in the $y$-basis. The total joint state immediately prior to $t_3$ in this scenario is then
\begin{align}\label{eq:appetajointILI3}
    &\eta_3^{sezy}(\mathcal{I}_1, \Mcal_2) = q \varphi_{se}^+ \otimes \ket{00}\bra{00}_{zy} \notag \\
    &+ (1-q) \ket{-^{(y)}} \bra{-^{(y)}}_{s} \otimes \tau_{e} \otimes \ket{01}\bra{01}_{zy}.
\end{align}
In this case, the experimenter then implements the identity map to the system at $t_3$, which leaves the overall state invariant. The subsequent system-environment dynamics $\Gamma_x^{sex}$ will be enacted when $zy=00$, i.e., with probability $q$: in each such run, the system-environment state is guaranteed to be in the state $\varphi^+_{se}$, thus the system-environment state output by $\Gamma_x^{sex}$ will be also. In the other cases, when $zy\neq00$, the subsequent dynamics will be trivial. Thus, the total joint state immediately prior to $t_4$ in this scenario is 
\begin{align}\label{eq:appetajointILI}
    &\eta_4^{sezyx}(\mathcal{I}_1, \Mcal_2, \mathcal{I}_3) = q \varphi_{se}^+ \otimes \ket{000}\bra{000}_{zyx} \notag \\
   &+ (1-q) \ket{-^{(y)}} \bra{-^{(y)}}_{s} \otimes \tau_{e} \otimes \ket{010}\bra{010}_{zyx}.
\end{align}
The final reduced system state is thus biased in the $y$-basis:
\begin{align}\label{eq:appetasystemILI}
    \eta_4^{s}(\mathcal{I}_1, \Mcal_2, \mathcal{I}_3) = \frac{q}{2} \mathbbm{1} + (1-q) \ket{-^{(y)}} \bra{-^{(y)}}.
\end{align}

Lastly, consider the scenario where the experimenter first implements an arbitrary CPTP map $\Mcal_1\neq \mathcal{I}_1$ at $t_1$, followed by identity maps at the second and third time, $\mathcal{I}_2, \mathcal{I}_3$. Just as in the main text, $\Mcal_1 \neq \mathcal{I}_1$ will perturb the initial system-environment state and so the map $\Gamma_z^{sez}$ will only successfully record the outcome corresponding to $\varphi^+_{se}$ with some probability $p = \tr{\left[\varphi^+_{se} (\Mcal_1^s \otimes \mathcal{I}^e) [\varphi^+_{se}] \right]} < 1$; otherwise, the system will be biased in the $z$-basis. The total joint state immediately prior to $t_2$ in this scenario is then
\begin{align}\label{eq:appetajointLII2}
    &\eta_2^{sez}(\Mcal_1) = p \varphi_{se}^+ \otimes \ket{0}\bra{0}_{z} \notag \\
    &+ (1-p) \ket{0} \bra{0}_{s} \otimes \tau_{e} \otimes \ket{1}\bra{1}_{z}.
\end{align}
The identity map implemented by the experimenter on the system at $t_2$ does not change this state. Thus, $\Gamma_y^{sey}$ will subsequently be enacted with probability $p$, i.e., when $z=0$: in such cases, the system-environment is in the state $\varphi^+_{se}$ and the output of the map $\Gamma_y^{sey}$ will be so also, accompanied by the classical $y$ flag with the value $0$. In the other cases, the system-environment undergoes trivial dynamics. Again, at $t_3$ implementation of the identity map on the system leaves the joint state unperturbed. Only when the joint state of $zy$ is $00$ will the map $\Gamma_x^{sex}$ be implemented: in each such run, the system-environment is guaranteed to be in the state $\varphi^+_{se}$, and thus so too will be the output of the map. In the other cases, trivial dynamics ensues. The overall joint state in this scenario immediately prior to $t_4$ is then
\begin{align}\label{eq:appetajointLII}
    &\eta_4^{sezyx}(\Mcal_1, \mathcal{I}_2, \mathcal{I}_3) = p \varphi_{se}^+ \otimes \ket{000}\bra{000}_{zyx} \notag \\
   &+ (1-p) \ket{0} \bra{0}_{s} \otimes \tau_{e} \otimes \ket{100}\bra{100}_{zyx},
\end{align}
and so the reduced system state is biased in the $z$-basis:
\begin{align}\label{eq:appetasystemLII}
    \eta_4^{s}(\Mcal_1, \mathcal{I}_2, \mathcal{I}_3) = \frac{p}{2} \mathbbm{1} + (1-p) \ket{0} \bra{0}.
\end{align}

We are now in a position to prove the claim that we set out to, namely that the process considered is genuinely quantum. Demanding classicality means that the experimenter cannot distinguish whether an identity map or a dephasing map was implemented at any subset of previous times: to allow for arbitrary and possibly unrestricted interrogation schemes, here we have considered the more general case where the experimenter is allowed to implement \emph{arbitrary} CPTP maps, of which any POVM measurement followed by an arbitrary preparation is a special case. This more general notion of classicality (with respect to a general, possibly unrestricted, interrogation scheme) means that the experimenter cannot distinguish between the implementation of the identity map or the CPTP map at any subset of previous times and thereby provides a valid notion of a genuinely quantum process. Above, in Eqs.~\eqref{eq:appetasystemIIL},~\eqref{eq:appetasystemILI} and~\eqref{eq:appetasystemLII}, we have calculated the system state that would be measured at $t_4$ conditioned on the fact that a CPTP map was implemented at each one of the previous three times (as well as the case where only a sequence of identity maps was implemented in Eq.~\eqref{eq:appetasystemIII}). Intuitively, in each of the three cases, the system is biased in one of the $x$-, $y$- or $z$-basis directions, and in the case where the experimenter interacts only trivially with the system, it is completely unbiased. The only way that a measurement at $t_4$ cannot distinguish between these four scenarios is if it is blind to biases in \emph{every} basis. The only types of POVM that can achieve this are trivial, with all elements proportional to the identity matrix, $\{ \Pi_4^{(x_4)} \} \propto \mathbbm{1} \ \forall \ x_4$. Thus, there is no (non-trivial) measurement scheme for this process such that the full statistics appears classical, and thus it is a genuinely quantum process.

\end{document}